\documentclass[aps,prb,onecolumn,groupedaddress,showpacs,superscriptaddress,amssymb,amsmath,longbibliography,]{revtex4-2}
\usepackage{amsmath,amsthm}
\usepackage{amssymb}
\usepackage[margin=1in]{geometry}
\usepackage{mathtools}
\usepackage{dsfont}
\usepackage{xcolor}
\usepackage{algorithm2e}
\usepackage{graphicx}
\usepackage{subcaption}
\usepackage{hyperref}
\usepackage{cleveref}
\usepackage{nicefrac}

\usepackage{thmtools}
\usepackage{thm-restate}

\usepackage{mathrsfs}
\usepackage{nccmath}
\usepackage{placeins}
\usepackage{physics}
\usepackage{comment}
\usepackage{ulem}

\usepackage{tikz}
\usetikzlibrary{quantikz}

\usepackage[compat=0.4]{yquant}
\useyquantlanguage{groups}

\newtheorem{theorem}{Theorem}
\newtheorem{lemma}[theorem]{Lemma}
\newtheorem{corollary}[theorem]{Corollary}
\newtheorem{proposition}[theorem]{Proposition}
\newtheorem{definition}[theorem]{Definition}
\newtheorem{problem}[theorem]{Problem}

\usepackage{todonotes}

\newcommand{\bea}{\begin{eqnarray}}
\newcommand{\eea}{\end{eqnarray}}

\newcommand{\Braket}[3]{\langle #1|#2|#3\rangle}

\newcommand{\ceil}[1]{\left\lceil #1 \right\rceil}

\DeclareMathOperator*{\argmin}{\arg\!\min}
\DeclareMathOperator{\sign}{sgn}
\DeclareMathOperator{\erfunc}{erf}

\begin{document}
\title{Constant-Factor Improvements in Quantum Algorithms for Linear Differential Equations}
\author{Matthew Pocrnic}
\affiliation{Department of Physics, 60 Saint George St., University of Toronto, Toronto, Ontario,  M5S 1A7, Canada} 
\author{Peter D. Johnson}
\affiliation{Apollo Quantum LLC, Cambridge, MA 02139 USA}
\affiliation{Zapata AI Inc., Boston, MA 02110 USA}
\author{Amara Katabarwa}
\affiliation{Zapata AI Inc., Boston, MA 02110 USA}
\author{Nathan Wiebe}
\affiliation{Department of Computer Science, University of Toronto, Toronto ON, Canada}
\affiliation{Pacific Northwest National Laboratory, Richland Wa, USA}
\affiliation{Canadian Institute for Advanced Research, Toronto ON, Canada}

\begin{abstract}
Finding the solution to linear ordinary differential equations of the form $\partial_t u(t) = -A(t)u(t)$ has been a promising theoretical avenue for \textit{asymptotic} quantum speedups. However, despite the improvements to existing quantum differential equation solvers over the years, little is known about \textit{constant factor} costs of such quantum algorithms. 
This makes it challenging to assess the prospects for using these algorithms in practice.
In this work, we prove constant factor bounds
for a promising new quantum differential equation solver, the linear combination of Hamiltonian simulation (LCHS) algorithm. Our bounds are formulated as the number of queries to a unitary $U_A$ that block encodes the generator $A$.
In doing so, we make several algorithmic improvements
such as 
tighter truncation and discretization bounds on the LCHS kernel integral,
a more efficient quantum compilation scheme for the SELECT operator in LCHS,
as well as use of a constant-factor bound for oblivious amplitude amplification, which may be of general interest. 
To the best of our knowledge, our new formulae improve over previous state of the art by at least two orders of magnitude, where the speedup can be far greater if state preparation has a significant cost.
Accordingly, for any previous resource estimates of time-independent linear differential equations for the most general case whereby the dynamics are not \textit{fast-forwardable},
these findings provide a 100-200x reduction in runtime costs.
This analysis contributes towards establishing more promising applications for quantum computing.
\end{abstract}

\maketitle

\section{Introduction}
\label{sec:introduction}
Differential equations are foundational to modeling phenomena in science and engineering, with applications including fluid dynamics, control systems, population models. As a result, the efficient numerical solution of ordinary differential equations (ODEs) remains a core objective in high-performance computing \cite{utkarsh2024automated}. In parallel, quantum computing has emerged as a potentially transformative platform for simulating dynamical systems.
However, determining the extent to which quantum computing can be transformative
within the landscape of differential equations applications remains an important problem for the field~\cite{li2025potential,aaronson2015read,costa2019quantum}.

Among applications of differential equations, simulating quantum systems has surfaced as one of the more promising uses for future quantum computers \cite{lloyd1996universal, whitfield2011simulation}.
Many approaches to quantum simulation involve the task of Hamiltonian simulation in which the quantum computation encodes an evolution under the Schrodinger equation. 
Algorithms for the task of Hamiltonian simulation have been placed on solid theoretical foundations, culminating in asymptotically optimal methods such as quantum signal processing and qubitization \cite{low2017optimal,low2019hamiltonian}. 
Beyond asymptotic analysis, several works \cite{jennings2023efficient, berry2024doubling} have established non-asymptotic upper bounds on the costs of Hamiltonian simulation.
Such non-asymptotic analyses are needed to predict the practicality of running such algorithms of future fault-tolerant quantum computers.
In some cases, researchers are able to use non-asymptotic analyses to estimate the resources (e.g.  wall clock time and qubit count) of future quantum computations~\cite{rubin2024quantum}.
As quantum devices continue to improve, such non-asymptotic analyses have become increasingly important \cite{goings2022reliably, penuel2024feasibility, watts2024fullerene, bellonzi2024feasibility, agrawal2024quantifying}, shedding light on which applications might benefit most from quantum computing.

Outside of quantum simulation, non-asymptotic analyses and resource estimations for other differential equations applications are lacking.
To our knowledge there is only one work in the literature which provides non-asymptotic query counts \cite{jennings2024cost} and one which estimates the quantum gate costs \cite{penuel2024feasibility} for solving differential equations on quantum computer.
One reason for this dearth is that algorithmic developments for these algorithms have lagged behind those for Hamiltonian simulation.
A difficulty with simulating dynamical systems on a quantum computers is that, while the quantum operations are unitary, general dynamical evolutions are non-unitary.
This challenge was first directly addressed by Berry \cite{berry2014high}, who proposed an algorithm for solving linear ODEs on a quantum computer by leveraging a quantum linear system solver. Since then, several improvements have been made to the original approach, including both tighter asymptotic analyses \cite{berry2017quantum, krovi2023improved} and entirely new algorithmic frameworks \cite{jin2022quantum, fang2023time, an2023linear}. 
Recent work has also studied the complexity of block encoding differential operators at the circuit level, providing deeper insight into quantum solutions for PDEs with a variety of boundary conditions \cite{kharazi2024explicit, schleich2025arbitrary, lubasch2025quantum}. A new mapping has also emerged that equates an ODE solution to a Lindbladian evolution \cite{shang2024design}, allowing one to solve a system of ODEs with any modern quantum algorithm for Lindbladian dynamics \cite{li2022simulating, pocrnic2023quantum, ding2024simulating, cleve2016efficient}. However, constant factor analyses are also lacking on this front. At the same time, limitations in practical applicability have been noted for particular differential equations \cite{montanaro2016quantum}. 

A recent proposal by Li et al.~\cite{li2025potential} suggested an exponential speedup for solving the Lattice Boltzmann equations using quantum algorithms, though this claim was supported only by asymptotic analysis. In contrast, Jennings et al.\ \cite{jennings2024cost} provided the first detailed constant-factor cost analysis of a quantum ODE solver and demonstrated constant-factor improvements relative to prior work. Yet, follow-up work by Penuel et al.\ \cite{penuel2024feasibility} analyzed the costs of this algorithm for solving the Lattice Boltzmann equations in the context of computational fluid dynamics and found that the resource requirements were prohibitive even for classically-realizable instances. Their analysis identified the ODE solver subroutine as the dominant contributor to the overall cost, raising important questions about whether alternative methods might revive the promise of quantum advantage for solving non-linear ODEs.

In contrast to the linear systems solver approach \cite{berry2014high, krovi2023improved, jennings2023efficient}, another direction for simulating classical dynamics on a quantum computer involves directly propagating a quantum state that encodes the initial-time boundary condition of the differential equation. Two notable approaches in this category are the time-marching method \cite{fang2023time} and the Linear Combination of Hamiltonian Simulations (LCHS) method \cite{an2023linear, an2023quantum}. 
An appealing property of both of these algorithms is the fact that the number of queries to the initial state preparation is independent of the final time, whereas for linear solver based methods this number of calls grows with the final time.
However, in the more general non-unitary setting, these approaches exhibit different scaling behavior: the time-marching method suffers from quadratic scaling in the evolution time $t$, while the LCHS method achieves the optimal linear scaling. 
In comparison to the time marching approach, however, LCHS scales slightly sub-optimally in the desired precision; when all other parameters are treated as constant as LCHS has complexity $O\left (\log(1/\epsilon)^{1/\beta}\right)$, where $\beta \in (0,1)$ is a free parameter that only affects the constants otherwise. This scaling becomes optimal when $\beta=1$, and this precision scaling has been achieved by other approaches including time marching. 

Motivated by this near optimal dependence on all parameters, and the fact that LCHS relies heavily on Hamiltonian simulation as a subroutine, a highly optimized subroutine in quantum computing, we focus in this work on the LCHS method and leave a detailed, non-asymptotic analysis of the time-marching approach and other methods to future study.  
We investigate the LCHS algorithm as proposed in \cite{an2023linear, an2023quantum}, with the goal of assessing its relative viability through a full constant-factor resource analysis. To this end, we construct an explicit implementation of the method and derive detailed bounds on its query complexity. \\

Our main result is summarized as follows. Note that the terminology and notation used in the following can be found in Section \ref{sec:preliminaries}.  \\

\textbf{Main Result} (Informal version of Theorems \ref{thm:amplified_prep} and \ref{thm:error}.)
\textit{Given access to the unitary $U_A$, a $(\alpha_A, m_A, 0)$ block encoding of the differential operator $A$, and a state-preparation unitary $U_0$, then there exists a quantum algorithm to solve the linear differential equation}
\begin{equation}
    \frac{\dd u(t)}{\dd t} = -Au(t), \: \: \: u(0)=u_0,
\end{equation}
\textit{by approximately encoding the time-$t$ solution $u(t)$ in a quantum state $\ket{u(t)} = u(t)/\|u(t)\|$ up to arbitrary precision $\epsilon$, that uses at most $C_A$ calls to the block encoding of $A$, where }
\begin{equation}
    C_A=C_{\mathrm{FPOAA}}\left(\Delta, \frac{\epsilon}{8\|u(t)\|}\right)\ceil{e \sqrt{1+K^2}\alpha_A t + 2\ln \left (\frac{2304 \sqrt{1+1/e} \ln \left( \frac{256 \sqrt{2}}{3\sqrt{\pi} \epsilon}\right) \|c_1\| \|u_0\|}{(3\sqrt{2\pi}e^{1/13}) \epsilon }\right ) }.
\end{equation}
\textit{In this expression, $K$ is an upper bound on the interval of truncation for the LCHS integral formula and $t$ is the total simulation time. $K$ can be calculated using the corresponding bound given in Lemma \ref{lem:K_formula} by setting $\epsilon_{\mathrm{trunc}}= \epsilon/4 \|u_0\|$. The function $C_{\mathrm{FPOAA}}(\Delta,\epsilon_{\mathrm{AA}})$ represents the complexity of fix-point oblivious amplitude amplification, of which we analyze the constant factor query complexity in Appendices \ref{app:polynomial_degree_bound} and \ref{app:robust_AA}. In the context of our analysis, this represents the number of calls by fixed-point amplitude amplification to the LCHS algorithm, and also the number of calls to the state preparation oracle $U_0$. For this reason, it is labeled as $C_{\mathrm{LCHS}}$ after this point. This function has the form} 
\begin{equation}
    C_{\mathrm{FPOAA}}(\Delta,\epsilon_{\mathrm{AA}}) = \ceil{\sqrt{8\ceil{\frac{4}{\Delta^2}\ln(8/\pi \epsilon_{\mathrm{AA}}^2) e^2 }\ln\left (\frac{64\frac{\sqrt{2}}{\Delta}\ln^{1/2}(8/\pi \epsilon_{\mathrm{AA}}^2)}{3\sqrt{\pi} \epsilon_{\mathrm{AA}}}\right )}+1},
\end{equation}
\textit{where $\Delta$ is at most twice the amplitude of the state output by LCHS. In this case, for the purposes of simplification, we have used $\epsilon_{\mathrm{AA}} = \epsilon/8\|u(t)\|$ and take $\Delta = 2(\|u(t)\|-\epsilon)/\|u_0\|\|c\|_1$. Here $\|c\|_1$ is the summed absolute values of the coefficient weights in said integral (computed exactly in Appendix \ref{sec:gauss}).}  \\

Note that for purposes of expedition, compared to the main results in Theorem \ref{thm:amplified_prep} and \ref{thm:error}, here we have simplified our results to the case where the block encoding and state preparation are error free, as this makes the query formulas simpler to state. In order to state the aforementioned result, choices were made to distribute sources of error for analytical convenience, sacrificing some tightness for the purpose of clarity. The analysis within this manuscript presents the full technical details of how all errors contribute. For full accounting of errors required for detailed costing of an application, please refer to Section \ref{sec:lchs_implementation}.  

Equipped with the above constant-factor expressions for the cost of LCHS, we benchmark this implementation against the randomization linear solver (RLS) approach developed in \cite{jennings2024cost}. We find that the LCHS method yields approximately a 100-200x reduction in block-encoding query complexity for the same task (see Fig. \ref{fig:cost_plot}). Additionally, numerical studies indicate that for fixed $\epsilon$ this advantage may grow like $O(\log(\alpha T))$ in certain parameter regimes (see Fig. \ref{fig:ratio_plot}), further suggesting that LCHS may yield more practical quantum differential equation solvers.
We also compare the costs of the two methods as a function of the relative costs of initial state preparation and $A$-matrix block encoding.
We find, for example, that when the initial state preparation is ten times more costly than the $A$-matrix block encoding, then LCHS speedup factor increases to over 1000 (see Fig. \ref{fig:speedup_plot}).

Additionally, at the end of Appendix \ref{app:polynomial_degree_bound}, we provide a slightly tighter version of $C_{\mathrm{FPOAA}}(\Delta,\epsilon_{\mathrm{AA}})$ than that stated above in terms of non-elementary functions, stemming from a constant factor bound on the degree of a polynomial that approximates the function $\sign(x)$, which may be of more general interest. Note that for simplicity we do not use this form in the numerical analysis, and doing so would not significantly affect the results, as each ODE solver produces an output that must be amplified. Bounds of a similar nature involving a different approximation construction can also be found in Refs. \cite{wan2022randomized, berntson2025complementary}.

The remainder of this paper is organized as follows. Section~\ref{sec:preliminaries} reviews key background material and notation. In Section~\ref{sec:lchs_integral}, we analyze the discretization and truncation approximations in the LCHS formulation. Section~\ref{sec:lchs_implementation} presents the quantum implementation of the LCHS operator and derives performance guarantees. In Section~\ref{sec:numerics}, we provide a numerical comparison of LCHS to previous methods. Finally, Section~\ref{sec:discussion_and_outlook} discusses implications of our findings and future research directions.

\section{Preliminaries}
\label{sec:preliminaries}
Given that this paper pertains to solving general linear ODEs, we are interested in finding a solution vector $u(t)$ at time $t$, which is not a quantum state. Therefore, in this work, we consider the cost of encoding the non-normalized solution vector of the ODE $u(t)$ into a normalized quantum state $\ket{u(t)}$ such that $\ket{u(t)} = u(t)/ \|u(t)\|$. If not written as a ket, all vectors written are of arbitrary norm unless stated otherwise. Throughout this work we use $\widetilde{A}$ to mean an imperfect construction of the operator $A$, which may arise for a variety of reasons, all of which will be due to either representing the problem or computing the solution on a fault tolerant quantum computer. \\

This paper relies heavily on the well-known block encoding access model for quantum computing. This model assumes we can access the operator $A$ via its encoding in a larger unitary $U_A$, such that
\begin{equation}
    U_A = \begin{bmatrix}
        A/\alpha_A & * \\
        * & *\\
    \end{bmatrix},
\end{equation}
where $\alpha_A \geq \|A\|$ the sub-normalization constant to ensure $U_A$ is unitary, and we are only concerned of the upper left block of this operator. $\widetilde{U}_A$ is called an $(\alpha_A, m_A, \epsilon_A$), block encoding of $A$ if $m_A$ qubits are used to encode the operator, and if the encoding is done up to error $\epsilon_A$, where this error is defined as  
\begin{equation}
\label{eq:block_encoding}
  \|\alpha_A(\bra{0}^{m_A} \otimes \mathbb{I}) \widetilde{U_A} (\ket{0}^{m_A} \otimes \mathbb{I}) - A\| \leq \epsilon_A.  
\end{equation}
Here, we have again used the notation $\widetilde{(\cdot)}$ to indicate an imperfect operator. \\

Throughout the paper we track a variety of errors and cost-determining parameters and objects that arise in the end-to-end analysis. For bookkeeping and reader referral, we list them all in Table \ref{tab:notation}.

\begin{table}[htbp!] 
    \centering
    \begin{tabular}{| c | c | c | c |}
    \hline
        Notation & Description  \\
        \hline
        $C_A$ & Algorithmic cost defined as queries to the block encoding $U_A$  \\
        $t$ & Total Simulation time \\
        $\beta$ & A parameter $\in (0,1)$ that gives LCHS slightly sub-optimal asymptotic scaling w.r.t $\epsilon$\\
        $u(t)$ & The exact solution vector at time $t$ \\
        $\cal{L}$ & LCHS propagator: $\mathcal{L} := \sum_jc_j U(t,k_j)$ \\
        $C_{\mathrm{LCHS}}$ & The number of calls to the LCHS propagator due to amplitude amplification \\
        $v(t)$ & The approximate solution vector generated by LCHS: $({v(t)} :=\mathcal{L} {u_0})$ \\
        $\epsilon$ & The total error incurred by the end-to-end algorithm \\
        $K$ & The truncation parameter of the LCHS integral bounded in Lemma \ref{lem:K_formula} \\
        $M$ & The total number of terms in the discretized LCHS integral bounded in Lemma \ref{lem:class_error}\\
        $\mathcal{Q}$ & The unitary approximating the SELECT operator in LCHS \\ 
        $\mathcal{S}$ & The effective Hamiltonian generating the LCHS propagator \\        $\epsilon_{\mathrm{trunc}}$ & The error incurred by truncating the LCHS integral \\
        $\epsilon_{\mathrm{disc}}$ & The error incurred by discretizing the LCHS integral \\
        $\epsilon_{v}$ & The error in approximating the exact solution vector with the LCHS-generated vector\\
        & $\left ( \epsilon_{v} := 
        \left\| u(t) - v(t) \right \| \right )$ 
        \\
        $\epsilon_A$ & The block encoding error of $\widetilde{U}_A$ \\
        $\epsilon_0$ & Initial state preparation error \\
        $\epsilon_{c}$ & The error associated with the oracles that form the state preparation pair $(O_{c,l}, O_{c,r})$ \\
        $\epsilon_{\text{exp}}$ & The error due to imperfect exponentiation of the effective Hamiltonian\\
        $\epsilon_Q$ & The total error due to Hamiltonian simulation by qubitization \\
        $\epsilon_{\mathrm{AA}}$ & The error due to imperfect amplitude amplification \\        
        \hline
    \end{tabular} 
    \caption{Summary of notation used in this work.}
    \label{tab:notation}
\end{table} 

\subsection{Computational Problem}
The goal of quantum algorithms to solve differential equations is to encode in a quantum state $\ket{u(t)}$, the solution of to the differential equation $u(t)$, such that $\ket{u(t)} = u(t) / \|u(t)\|$. Therefore, an algorithm that solves this problem outputs a quantum state encoding of the solution up to a constant of proportionality. This can be viewed as slightly less general to a classical solution of the same problem, given that we do not have direct access to the elements of the solution vector in the quantum case, as we would have to perform a tomography routine, which would erase the hopes of any quantum speedup. Therefore, what we can efficiently extract from the quantum encoding are global properties of the solution, such as an observable $\Braket{u(t)}{O}{u(t)}$, the solution norm $\|u(t)\|$, etc. \\

The computational problem can then be stated as follows: 

\begin{problem}
Let $u(t) \in \mathbb{C}^d$  and let $u_0 \in \mathbb{C}^d$ is the initial condition. Further, let $A \in \mathbb{C}^{d\times d}$ be the an operator such that 
\begin{equation} \label{eq:prob1}
    \frac{\dd u(t)}{\dd t} = -Au(t), \: \: \: u(0)=u_0.
\end{equation}
where $t \in \mathbb{R}^+$ and the solution vector at time $t$ is $u(t)=u(t)$. Let  $U_A$ be a block encoding of the operator $A$ and let $U_0$ be a unitary that prepares the initial state $\ket{u_0}$, a total simulation time $t$, and an error tolerance $\epsilon$, output a vector $v(t) \in \mathbb{C}^d$ such that $\|v(t) - u(t)\| \leq \epsilon$. 
\end{problem}

\subsection{Linear Combinations of Hamiltonian Simulation}
LCHS or Linear Combinations of Hamiltonian Simulations is a recently developed quantum algorithm \cite{an2023linear, an2023quantum} that provides the means for solving the following differential equation:
\begin{equation}\label{eq:ode}
    \frac{\dd u(t)}{\dd t} = -A(t)u(t) + b(t),
\end{equation}
where $A$ need not be Hermitian. A notable feature of LCHS is that is also solves time-dependent linear ODEs, although for this work we will analyze only the time-independent homogeneous case. The high-level idea of the algorithm is the following: \\
First, use the $A$ matrix to construct the following two Hermitian matrices 
\begin{equation}
    L(t) \equiv \frac{A(t)+A(t)^\dagger}{2} \: \: \: H(t) \equiv \frac{A(t)-A(t)^\dagger}{2i},
\end{equation}
such that $A=L+iH$. Then using the main Theorem of Ref. \cite{an2023quantum}, we can rewrite the exact solution to the ODE (in the simpler case where $b(t)=0$) in the following way  
\begin{equation} \label{eq:LCHS}
    \mathcal{T}e^{-i \int_0^t A(s) \dd s} = \int_{\mathbb{R}}\frac{f(k)}{1-ik}\mathcal{T}e^{-i \int_0^t (kL(s) + H(s))\dd s} \dd k,
\end{equation}
given that $L\succeq 0$, and $f(k)$ is chosen from a class of functions obeying certain analyticity criteria. This is equivalent to the requirement that the real part of the solution is ``non-increasing", and is required to satisfy stability conditions in the proof. This condition can always be satisfied, however, through introducing a shift in the eigenvalues. This can be done be redefining the solution $u(t) = e^{ct}w(t)$, yielding
\begin{equation}
    \partial_t w(t) = -(L(t) + cI + iH)w(t) + e^{-ct}b(t),
\end{equation}
where $c = \abs{\argmin_{t,j} \lambda_j(L(t))}$.
At first glance, the LCHS formula seems un-intuitive, however, some intuition can be gained through first observing the original formula introduced in \cite{an2023linear}:
\begin{equation}
    \mathcal{T}e^{-i \int_0^t A(s) ds} = \int_{\mathbb{R}}\frac{1}{\pi(1+k^2)}\mathcal{T}e^{-i \int_0^t (kL(s) + H(s))ds}dk.
\end{equation}
This can be understood as a matrix adaptation of a more general version of the following Fourier property:\\
\begin{equation}
    \hat{f}(k) = \frac{1}{2\pi}\int_\mathbb{R} e^{-|x|} e^{ikx} dk = \frac{1}{\pi (1+k^2)}.
\end{equation}
This formula trivially coincides with that of LCHS in the case $H=0$, where in the matrix integral $L$ would assume the role of $x$ above. The analytical work of LCHS therefore generalizes this to the more general case being investigated here. \\

For this integral identity to be promoted to the level of quantum algorithm, it must be converted to a form amenable to a quantum computer. In order to achieve this, we must first make the interval of integration finite by performing a truncation, such that $\mathbb{R} \to [-K, K]$, where $K$ is chosen such that the approximation error is bounded by $\epsilon$. The formula must then be made discrete, by converting the integral to a sum. There are a variety of ways this can be done, including the simple Riemann sum (this work and the improved LCHS work \cite{an2023quantum} use the method of Gaussian Quadrature which is briefly reviewed in Appendix \ref{sec:gauss}. Again, the number of points to well approximate integral must be chosen such that the discrete approximation error is bounded by $\epsilon$. By truncating and discretizing the integral in Equation \ref{eq:LCHS}, and writing $g(k)=\frac{f(k)}{1-ik}$ and $U(t,k) = \mathcal{T}e^{-i \int_0^t (kL(s) + H(s))ds}$, we see that the solution to the ODE can be written as a linear combination of Hamiltonian simulations:
\begin{equation}
    \mathcal{T}e^{-i \int_0^t A(s) ds} \approx \sum_k g(k) U(t, k).
\end{equation}
For our purposes, we choose to more closely examine the case where $A=A(t)$ and $b(t) = 0$. Therefore wish to solve the ODE 
    \begin{equation}
        \frac{\dd u(t)}{\dd t} = -Au(t) .
    \end{equation}
    
Here, we study an application of the LCHS algorithm while providing a precise constant-factor resource estimate of obtaining said solution. In doing so, we will provide algorithmic cost formulas in terms of queries to an oracle $U_A$, or a block encoded version of $A$.

Since our problem instance is time independent, the LCHS formula in this case has the following simplification:
\begin{equation}
    e^{-At} = \int_\mathbb{R} \frac{f(k)}{1-ik} e^{-i t(kL + H)}dk
\end{equation}

\begin{theorem}[Time-Independent LCHS Identity] \label{thm:t-indep} Given a time $t \in \mathbb{R}^+$, a matrix $A \in \mathbb{C}^{dxd}$ such that $\mathfrak{R}(A) \succeq 0$, and a kernel function $f(k)$ that satisfies the properties of Theorem 5 in Ref. \cite{an2023quantum}, then the following identity holds
\begin{equation}
    e^{-At} = \int_\mathbb{R} \frac{f(k)}{1-ik} e^{-i t(kL + H)}\text{d}k.
\end{equation}
\end{theorem}
\begin{proof}
    One can observe that by setting $A(t) = A$ in Lemma 6 of Ref. \cite{an2023quantum}, such that the time ordered Dyson series reduces to the simple exponential propagator $\mathcal{T}e^{-i \int_0^t A(s) ds} \to e^{-At}$, it is clear that $\mathcal{P}\int_\mathbb{R} f(k) e^{-i t(kL + H)}dk = 0$, where $\mathcal{P}$ denotes the Cauchy principal value of the integral. As a result, the proof of [Ref. \cite{an2023quantum}, Theorem 5] yields the result in the Theorem statement. 
\end{proof}

\noindent Ref. \cite{an2023quantum} studies a family of kernel functions that satisfy the Theorem above. They find that the most performant kernel function, or the kernel function that decays most rapidly, to be 
\begin{equation} \label{eq:kernel}
    f(k) = \frac{1}{C_\beta}e^{-(1+ik)^\beta}, 
\end{equation}
where $C_\beta = 2\pi e^{-2^\beta}$ is a normalization constant, and $\beta \in (0,1)$. Ideally, we want $f(k)$ to decay exponentially fast such that size the integral truncation interval is only required to grow logarithmically with the desired precision. The free parameter $\beta$ keeps this scaling slightly sub-optimal. In Ref. \cite{an2023quantum} it is shown that this feature cannot be improved purely by finding better kernel functions, therefore, the problem of eliminating the $\beta$ dependence on the truncation interval is an open problem. We do not address this problem this work, however, we give constant factor bounds that show how the $\beta$ dependence can be slightly improved. We provide a detailed study of the truncation and discretization of the integral in Theorem \ref{thm:t-indep} in the following Section. 

\section{Analysis of the LCHS Integral}
\label{sec:lchs_integral}
In this Section we perform constant factor analysis of the numerical parameters affecting the convergence of the LCHS formula. Specifically, we provide constant factor bounds for the truncation cut-off $K$ and the number of quadrature points $Q$ in the integral. Since we are interested in the time-independent case, we begin each Subsection with discussions on whether this simplification leads to constant factor improvements in the LCHS parameters $K$ and $M$. With simple examples, we answer this question in the negative and find that for the case $A(t)=A$, the improvements come strictly via the Hamiltonian simulation algorithm, which, to those with a deep familiarity with LCHS techniques, may have been expected from the outset. 

\subsection{Integral Truncation}
The goal of this section is to derive a constant factor bound on the truncation parameter $K$ for the LCHS integral. First, we briefly show that the scaling of the integral truncation parameter does not change in the time-independent case. In the absence of any simplifications, we utilize the bound proven in Ref. \cite{an2023quantum}. By using special functions to exactly solve this bound, we find tighter constant factor formulas than the naive bound.\\
Let us write the truncated LCHS formula in the following way:
\begin{equation} \label{eq:trunc_LCHS}
    \int_\mathbb{R} g(k) U(t,k) \approx \int_{-K}^K g(k) U(t,k),
\end{equation}
where for the moment, allow the unitary to have some time dependence. Here we wish to analyze the truncation error as a function of $K$. Observe the following:
\begin{align}
    \left | \left |\int_\mathbb{R} g(k) U(t,k) - \int_{-K}^K g(k) U(t,k) \right | \right | &= \left | \left | \int_{K}^\infty g(k) U(t,k) + \int_{-\infty}^{-K} g(k) U(t,k) \right | \right | \\
    & \leq \int_{K}^\infty |g(k)|  + \int_{-\infty}^{-K} |g(k)|.
\end{align}
Observe that regardless of the time dependence of $U$, given the fact that for $U$ unitary  $\|U(t)\|_\infty = 1 \: \forall \: U,t$. This leads to the same truncation bound results from Ref. \cite{an2023quantum}. In the referenced work, it is argued that the kernel function that yields near optimality in terms of decay along the real axis is the following 
\begin{equation}
    g(k) = \frac{1}{C_\beta (1-ik) e^{(1+ik)^\beta}},\label{eq:kernel}
\end{equation}
where $C_\beta = 2\pi e^{-2^\beta}$. In this case, [\cite{an2023quantum}, Lemma 9] gives us the following bound, which then also holds in the time independent case: 

\begin{equation} \label{eq:trunc_bound}
    \left\|\int_\mathbb{R} g(k) U(t,k)dk - \int_{-K}^K g(k) U(t,k) dk \right\| 
    \leq \frac{2^{\left\lceil 1/\beta \right\rceil + 1} \left\lceil 1/\beta \right\rceil !}{C_\beta \cos{(\beta \pi/2)}^{\left\lceil 1/\beta \right\rceil}}
    \frac{1}{K} e^{-\frac{1}{2} K^\beta \cos(\beta \pi/2)} \leq \epsilon_{\text{trunc}}.
\end{equation}
From this bound we wish to extract a non-asymptotic or constant factor formula for $K(\epsilon)$. The most naive way to get an expression for $K(\epsilon)$ is to upper-bound $1/K \leq 1$, and then just solve to obtain the following:
\begin{equation} \label{eq:naive_K}
    K = \left (\frac{2 \ln \left( \frac{2^{\left\lceil 1/\beta \right\rceil + 1} \left\lceil 1/\beta \right\rceil !}{C_\beta \cos{(\beta \pi/2)}^{\left\lceil 1/\beta \right\rceil}}
    \frac{1}{\epsilon_{\text{trunc}}}\right)}{\cos(\beta \pi /2)} \right )^{1/\beta}, 
\end{equation} 
which is where the scaling $K \in O(\log(1/\epsilon_{\text{trunc}})^{1/\beta})$ is derived in Ref \cite{an2023quantum}. \\

However, in pursuit of our goal of reducing constant factors for LCHS methods, we would like to avoid performing this constant upper bound of $1/K$. Instead, to achieve a tight bound, we use special functions to solve this expression exactly. The results of solving Equation \ref{eq:trunc_bound} are given in the following Lemma.

\begin{lemma} \label{lem:K_formula}
Let $K \in \mathbb{R}$ be the truncation point of the LCHS integral defined by Equation \ref{eq:trunc_LCHS} and $\beta \in (0, 1)$ be a tunable parameter. Then to achieve a truncation error $\epsilon_\text{trunc}$ defined in Equation \ref{eq:trunc_bound}, $K$ takes the following form 

\begin{equation} \label{eq:lambert_K}
    K = \left (\frac{2}{\beta \cos (\beta \pi /2)} W_0\left ( \left (\frac{B_\beta}{\epsilon_{\text{trunc}}} \right )^{\beta} \frac{\beta \cos(\beta \pi /2)}{2} \right ) \right )^{1/\beta},
\end{equation}
and $$
B_\beta \equiv \frac{2^{\left\lceil 1/\beta \right\rceil + 1} \left\lceil 1/\beta \right\rceil !}{2\pi e^{-2^\beta} \cos{(\beta \pi/2)}^{\left\lceil 1/\beta \right\rceil}},
$$
which is a tightened version of the bound from \cite{an2023quantum}, and $W_0$ is the principal branch of the Lambert W-function.
    \begin{proof}
        Observing the right hand side of Equation \ref{eq:trunc_bound}, we can use a substitution to rewrite it in the form $x = ye^y$. First recall $C_\beta = 2\pi e^{-2^\beta}$ and in turn define $$B_\beta \equiv \frac{2^{\left\lceil 1/\beta \right\rceil + 1} \left\lceil 1/\beta \right\rceil !}{C_\beta \cos{(\beta \pi/2)}^{\left\lceil 1/\beta \right\rceil}},$$ which leaves us with 
\begin{equation}
    B_\beta \frac{1}{K} e^{-\frac{1}{2} K^\beta \cos(\beta \pi/2)} = \epsilon_{\text{trunc}}.
\end{equation}
If we define $y\equiv \frac{\beta K^\beta}{2}\cos(\beta \pi/2)$ and make the substitution $K = \left(\frac{2y}{\beta\cos(\beta\pi/2)}\right)^{1/\beta}$, we can rearrange the equation into the following form
\begin{align}
    \frac{B_\beta}{\epsilon_{\mathrm{trunc}}} &= \left ( \frac{2ye^y}{\beta \cos(\beta \pi /2)}\right )^{1/\beta} \notag \\
    \left (\frac{B_\beta}{\epsilon_{\mathrm{trunc}}}\right ) ^\beta \frac{\beta \cos(\beta \pi /2)}{2}&= ye^y,
\end{align}
Then if we define $x\equiv \left (\frac{B_\beta}{\epsilon_{\text{trunc}}} \right )^{\beta} \frac{\beta \cos(\beta \pi /2)}{2}$ we obtain
\begin{align}
    x = ye^y.
\end{align}
Now the function $x=ye^y$ has well known solutions for $y$ called \textit{Lambert W-functions}, of which its branches are enumerated $y=W_k(x)$ for integer $k$. Given $x \in \mathbb{R}$ and $x\geq 0$, which is implied by the restriction $\beta \in (0,1)$, our equation is solved by the principal branch $y=W_0(x)$. Solving for $K$ in terms of y now gives the stated result.
    \end{proof}
\end{lemma}
To provide additional intuition to the result, note that the principal branch of the Lambert W-function is bounded like so
\begin{equation}
    W_0(x) \leq \ln(x) - \ln\ln(x) + \frac{e}{e-1} \frac{\ln \ln(x)}{\ln(x)},
\end{equation}
which holds $\forall \: x \geq e$ \cite{hoorfar2008inequalities}.  To compare the constant factor bound proved in Lemma \ref{lem:K_formula} to that of Equation \ref{eq:lambert_K}, we plot $K(\epsilon)$ in Figure \ref{fig:Kplot}.

\subsection{Integral Discretization}
In terms of discretization, we first follow the scheme presented in the original paper. The method used here is that of Gaussian quadrature, which is specifically chosen due to the fact that the kernel function $f(k)$ can be well approximated by polynomials, which is the regime in which Gaussian quadrature enjoys rapid convergence. \\

Thus far we have a truncation bound for the integral, which we wish to further approximate:

\begin{equation} \label{eq:LCHSdisc}
    \int_{-K}^K g(k) U(t,k)\text{d}k = \sum_{m=-K/h}^{K/h-1}\int_{mh}^{(m+1)h} g(k) U(t,k) \text{d}k \approx \sum_{m=-K/h}^{K/h-1} \sum_{q=1}^Q c_{q,m} U(t,k_{q,m}). 
\end{equation}
In the above expression, $h$ can be viewed as a short time step that is used to break up the integrals into smaller pieces, and chosen such that $K/h$ is an integer. In the second approximate equality, the Gaussian quadrature is employed from numerical integration, where $Q$ is the number of quadrature nodes. We thus need to bound the number of nodes $Q$ as a function of $\epsilon$, as an encompassing bound on $2QK/h$ will gives us the overall number of terms in the summation. We provide a bound on $Q$ in the following Lemma which provides a tighter estimate of the sufficient value of $Q$ than the analysis of \cite{an2023quantum} which yields $
    Q = \ceil{\frac{1}{2}\log_2 \left (\frac{8K}{3 C_\beta \epsilon_{\mathrm{disc}}} \right )}
$ which will be an improvement largely from the use of Lambert W functions rather than their approximate forms as logarithmic functions.
\begin{lemma}[Quadrature points] \label{lem:Qbound}
    Let $K/h$ be an integer, $g(k)$ be the kernel function described in~\eqref{eq:kernel} and let $U(t,k) = e^{-it(kL+H)}$ be a unitary operator that is $C^\infty$ with respect to $k$, then for any $\epsilon_{\rm disc}>0$ there exists $Q>0$ such that
    $$
     \left\|\int_{-K}^K g(k) U(t,k)dk - \sum_{m=-K/h}^{K/h-1} \sum_{q=1}^Q c_{q,m} U(t,k_{q,m})\right\| \le \epsilon_{\rm disc}
    $$
    and further it suffices to pick $Q$ such that
\begin{equation}
    Q = \ceil{-\frac{\log(e)}{4} W_{-1}\left (\frac{-3 C_\beta}{2\pi e^{1/3} \log(e)} \frac{\epsilon_{\mathrm{disc}}}{K}\right )},
\end{equation}
\end{lemma}
\begin{proof}
The starting point for this is using the known error bounds on the Gaussian Quadrature method, which scale with the higher-order derivatives of the integrand: \\
\begin{equation} \label{eq:disc_bound}
    \left | \left | \int_{mh}^{(m+1)h} g(k) U(t,k) dk -  \sum_{q=1}^Q c_{q,m} U(t,k_{q,m})  \right | \right | \leq \frac{(Q!)^4 h^{2Q+1}}{(2Q+1)((2Q)!)^3 } \|(g(k) U(t,k))^{(2Q)}\|.
\end{equation}
where $g(k)$ is the kernel function given in~\eqref{eq:kernel}.
Following notation from the LCHS work, the superscript $(2Q)$ indicates the order of the $k$-derivative of the integrand. We can see that the method scales with the higher order $k$-derivatives of $U(t,k)$. We quickly verify here that the bounds go through similarly for a time independent $A$ matrix. 

To proceed with bounding the discretization error, we need to bound the high order $k$-derivatives of $U(t,k)$. In this section, we do not assume $A$ to be normal, and follow techniques in \cite{an2023quantum}. By definition
\begin{equation}
    \frac{dU(t,k)}{dt} = -i(kL + H)U(t,k),\label{eq:Ueqn}
\end{equation}
and after differentiating~\eqref{eq:Ueqn} $q$ times with respect to $k$ we have 
\begin{equation}
    \frac{dU^{(q)}}{dt} = -i(kL + H)U^{(q)} - iqLU^{(q-1)}.
\end{equation}
Now with the initial conditions $U(0,k) = \mathbb{I}$ for all $k$ it then follows that all higher $k$-derivatives are zero, i.e. $U(0,k)^{(q)} = 0 \: \forall \: q$, we can solve this differential equation using the method of integrating factor, given the fact that our matrices are time-independent. 
\begin{align}
    \frac{d}{dt} \left [e^{i\int (kL+H) dt} U^{(q)} \right ] &= e^{it(kL+H)} (i(kL+H)) U^{(q)} + e^{it(kL+H)} \frac{dU^{(q)}}{dt} \\
    &= e^{it(kL+H)} (-iqLU^{(q-1)}),
\end{align}
now by integrating both sides and applying the inverse of the unitary we arrive at the following
\begin{align}
     \int_0^t \frac{d}{dt} \left [e^{it (kL+H) } U^{(q)} \right ] dt & = \int_0^t e^{it(kL+H)} (-iqLU^{(q-1)}) dt \\
     U^{(q)} &= e^{-it(kL+H)}  \int_0^t e^{it(kL+H)} (-iqLU^{(q-1)}) dt.
\end{align}
Note that in integrating the LHS we also used the condition $U(0,k)^{(q)} = 0$ . Now we use the sub-multiplicative property of the spectral norm to see that
\begin{align}
    \max_t \|U^{(q)}\| &\leq \int_0^t \left | \left | e^{ix(kL+H)} (-iqLU^{(q-1)}) \right | \right| dx \\
    & \leq tq \|L\| \max_t \|U^{(q-1)}\|,
\end{align}
and recursing this expression over $q$ times and recalling that $\|U^{(0)}\| = \|U\|=1$ gives 
\begin{equation} \label{eq:k_derivative}
    \max_t \|U^{(q)}\| \leq t^q q^q \|L\|^q, 
\end{equation}
which reproduces the exact bound of Lemma 10 in \cite{an2023quantum} aside from the maximum taken over time-dependent $L$. This allows for the application of the Faa Di Bruno formula such as [\cite{an2023quantum}, Equation 196], which allows for the upper bounding of the norm of the derivatives in Equation \ref{eq:disc_bound}.  
For clarity, this expression is simply obtained by applying the time independent bound derived in Equation \ref{eq:k_derivative} into the proof of Lemma 10 in Ref. \cite{an2023quantum}. In fact, the only difference so far is the absence of taking the maximum over time for the $\|L\|$ term. Next, we plug this into Equation \ref{eq:disc_bound} to obtain
\begin{align}
    \left | \left | \int_{mh}^{(m+1)h} g(k) U(t,k) dk -  \sum_{q=1}^Q c_{q,m} U(t,k_{q,m})  \right | \right |& \leq \frac{(Q!)^4 h^{2Q+1}}{(2Q+1)((2Q)!)^3 } \frac{4}{3C_\beta}(2Q)! (2Q+1) (et\|L\|)^{2Q} 
    \\
    & = \frac{(Q!)^4 h^{2Q+1}}{((2Q)!)^2 } \frac{4}{3C_\beta} (et\|L\|)^{2Q}.
\end{align}
Now we apply the Robbins bounds on the factorial $\sqrt{2\pi n}(n/e)^ne^{1/(12n+1)} < n!<\sqrt{2\pi n}(n/e)^ne^{1/12n}$ to keep the constants tight 
\begin{align}
    \frac{(Q!)^4 h^{2Q+1}}{((2Q)!)^2 } \frac{4}{3C_\beta} (et\|L\|)^{2Q} &\leq \frac{\pi Q e^{1/3Q}}{2^{4Q}e^{2/(24Q+1)}}  \frac{4}{3C_\beta} h^{2Q+1} (et\|L\|)^{2Q} \\
    &\leq \frac{\pi Q e^{1/3}}{2^{4Q}}  \frac{4}{3C_\beta} h^{2Q+1} (et\|L\|)^{2Q}. 
\end{align}
In the top line we lower bound $e^{2/(24Q+1)}$ by taking $Q \to \infty$ and upper bound $e^{1/3Q}$ using $Q=1$. Next, we sum this value over all $2K/h$ short intervals
\begin{align}
    \sum_{m=-K/h}^{K/h-1} \left | \left |\int_{mh}^{(m+1)h} g(k) U(t,k) dk -  \sum_{q=1}^Q c_{q,m} U(t,k_{q,m}) \right | \right | &\leq \frac{\pi Q e^{1/3}}{2^{4Q}}  \frac{8K}{3C_\beta} h^{2Q} (et\|L\|)^{2Q}   \\  
    & =\frac{\pi e^{1/3}}{e^{Q} Q}  \frac{2K}{3C_\beta},
\end{align}
where we set the time step $h=1/(et\|L\|)$. Now, setting this equal to $\epsilon_{\mathrm{disc}}$ and solving for $Q$ we first make the change of variables $Q = \frac{x}{4} \log(e)$, which yields the form 
\begin{equation}
    \underbrace{\frac{2 \pi e^{1/3} \log(e)}{3C_\beta}\frac{K}{\epsilon}}_y = \frac{e^x}{x},
\end{equation}
and by setting the LHS equal to $y$ and $x=-z$ we once again obtain an expression that can be solved by the \textit{Lambert-W function} where 
\begin{equation}
    ze^{z} = -1/y
\end{equation}
has solution
\begin{equation}
    z=W_k(-1/y).
\end{equation}
However, this is not true in general, as we must be careful to ensure the correct branch of $W_k$ is applied and that the solution is unique. Since $y$ is real and strictly negative, $W_{-1}(-1/y)$ is a unique solution iff $-1/y \in [-1/e, 0)$. Writing out $-1/y$ and making constants as negative as possible we obtain 
$$-1/y \geq \frac{-3}{e^1e^{1/3}\log(e)}\frac{\epsilon}{K},$$
where we have applied the definition $C_\beta = 2\pi e^{-2^\beta}$, and taken the limit $\beta \to 0$. By computing this expression, one can see that the condition $-1/y \in [-1/e, 0)$ is satisfied for any $\epsilon/K \lesssim 0.67$, which is true in all interesting cases. Therefore $W_{-1}(-1/y)$ is a unique solution. Thus, by substituting all variables and enforcing $Q$ to be an integer, we obtain the following sufficient value of $Q$ 
\begin{equation}
    Q = \ceil{-\frac{\log(e)}{4} W_{-1}\left (\frac{-3 C_\beta}{2\pi e^{1/3} \log(e)} \frac{\epsilon_{\mathrm{disc}}}{K}\right )}.
\end{equation} 
\end{proof}

The improvement obtained by the tightening of this bound can also be seen in Figure \ref{fig:Mplot}. We also remark that for the case where $A$ is a normal matrix we have that $[A, A^\dagger ] = 0$ which allows these constants to be further tightened. This is because we can then take the $k$ derivatives of $U(t,k)$ directly by first performing  Trotterization: $U(t,k) = e^{-itkL}e^{-iHt}$, and then differentiating. Doing so reproduces Equation \ref{eq:k_derivative}, without the factor of $q^q$. However, due to the factorial terms that remain in the overall Gaussian quadrature bound, this does not drastically change the scaling, and may require a variable time-step $h_m$ to see significant improvements in algorithmic performance. We leave these questions to future work. 

\begin{figure}[t!]
    \centering
        \begin{subfigure}[b]{.49\textwidth}
            \includegraphics[width=1\textwidth]{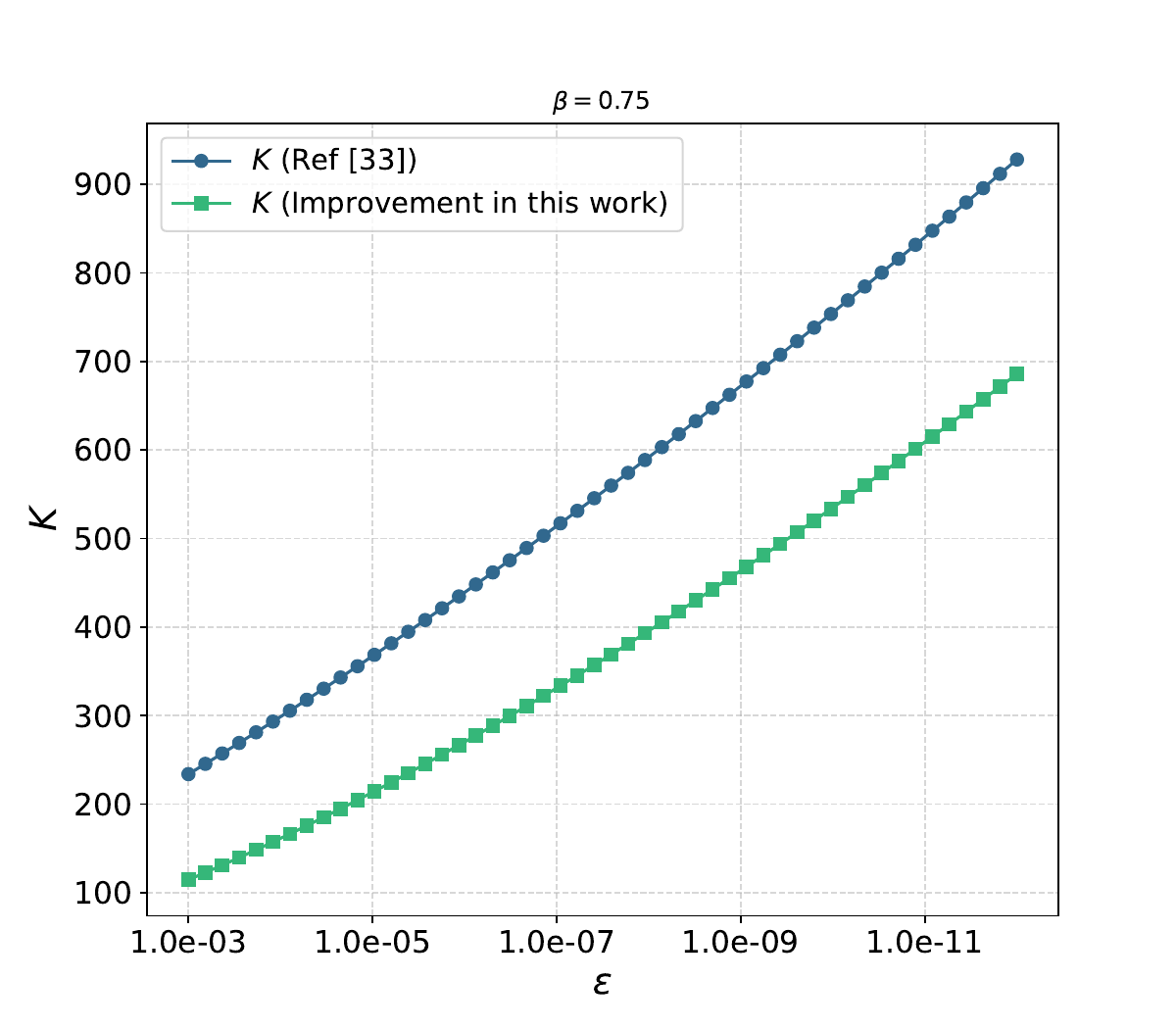}
            \caption{}\label{fig:Kplot}
        \end{subfigure}
        \begin{subfigure}[b]{.49\textwidth}
            \includegraphics[width=1\textwidth]{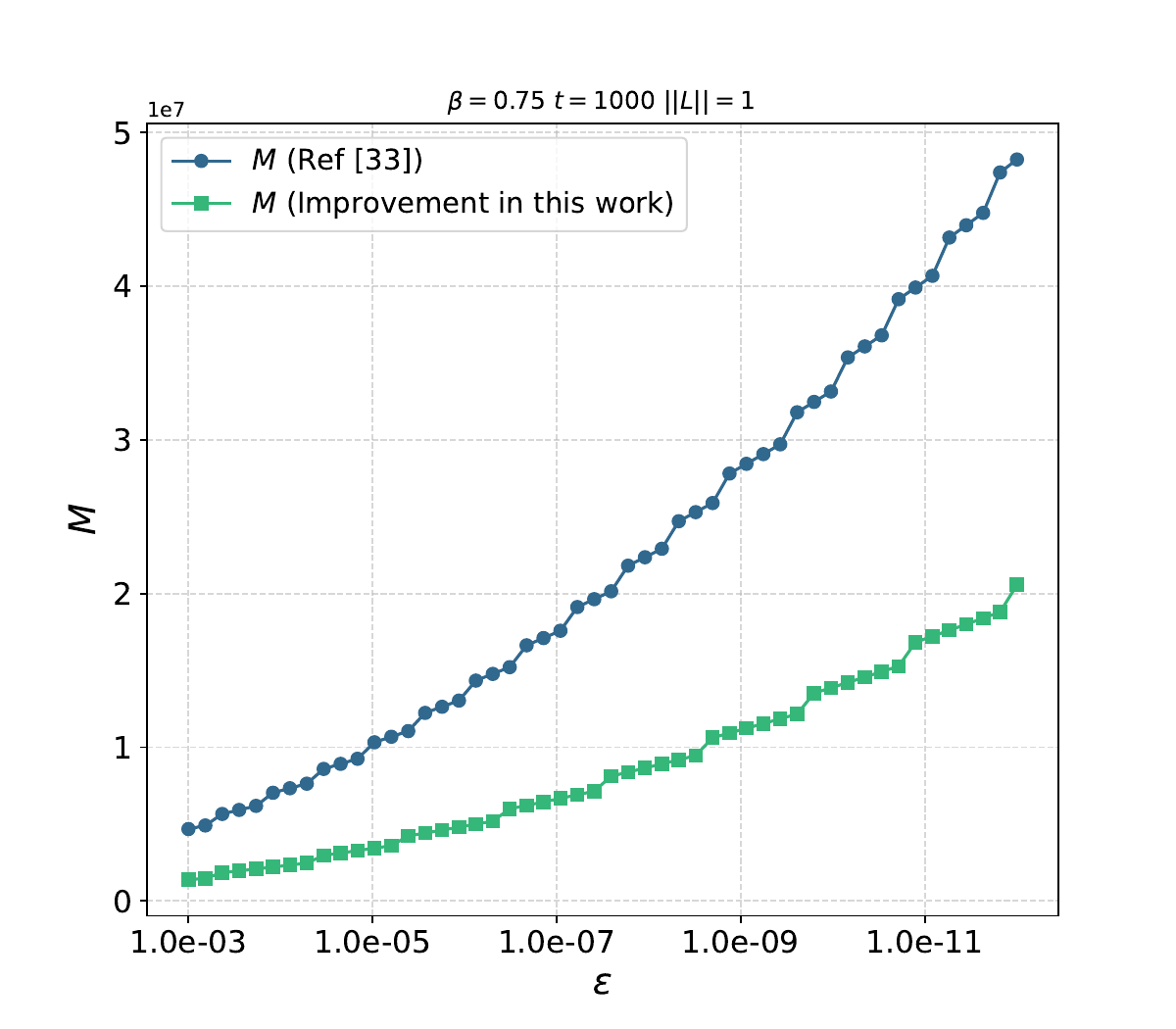}
            \caption{}\label{fig:Mplot}
        \end{subfigure}
        \caption{\textit{a) A comparison of $K$ upper bounds, which is the size of the truncation region of the integral which subsequently impacts the block encoding sub-normalization factor of the LCHS algorithm. b) A comparison of $M$ upper bounds which modulates the number of ancilla qubits and rotation gates required in constructing the block encoding of the LCHS formula. The improvement here is due to both improvements in $K$ and $M$ bounds. 
        }} 
\end{figure}
\FloatBarrier

\subsection{Total number of Summands}
Given that we now have a bound on the number of quadrature points in the integral, we can bound the total costs and establish error constraints.
The costs we consider are the number of summands in the LCHS formula, which we will denote $M$, and the 
truncation level $K$, which will inform the quantum implementation costs (see Lemma \ref{subsec:select}).
We also establish how the controllable errors $\epsilon_{\mathrm{trunc}}$ and $\epsilon_{\mathrm{disc}}$ contribute to the total approximation error.

\begin{lemma}[LCHS Cost and Error] \label{lem:class_error}
Let $A\in \mathbb{C}^{2n \times 2n}$ be a bounded operator and let $u:\mathbb{R}\rightarrow\mathbb{C}^{2^n}$ be a time-dependent vector that describes the exact solution to the differential equation $\frac{du(t)}{dt} = -Au(t)$ with initial conditions $u(0)=u_0$. When acting on the initial state vector $u_0$, the LCHS integrator $\mathcal{L}=\sum_{j=1}^M c_j e^{-it(k_j L +H)}$, as defined in Eq. \ref{eq:LCHSdisc}, generates a vector $v(t)$ satisfying $\|v(t)-u(t)\|\le \epsilon_v$, using a number of terms no more than
\begin{align}
\label{eq:trunc_and_disc_cost}
    M &\leq t \|L\| e \left (\frac{2\beta}{\cos (\beta \pi /2)} W_0\left ( \left (\frac{B_\beta}{\epsilon_{\text{trunc}}} \right )^{1/\beta} \frac{\cos(\beta \pi /2)}{2\beta} \right ) \right )^{1/\beta} \nonumber \\
    &\times \ceil{-\frac{\log(e)}{4} W_{-1}\left (\frac{-3 C_\beta}{2\pi e^{1/3} \log(e)} \frac{\epsilon_{\mathrm{disc}}}{\left (\frac{2\beta}{\cos (\beta \pi /2)} W_0\left ( \left (\frac{B_\beta}{\epsilon_{\text{trunc}}} \right )^{1/\beta} \frac{\cos(\beta \pi /2)}{2\beta} \right ) \right )^{1/\beta}}\right )},
\end{align}
and using a truncation level 

\begin{equation} \label{eq:k_bound}
    K = \left (\frac{2\beta}{\cos (\beta \pi /2)} W_0\left ( \left (\frac{B_\beta}{\epsilon_{\text{trunc}}} \right )^{1/\beta} \frac{\cos(\beta \pi /2)}{2\beta} \right ) \right )^{1/\beta}
\end{equation}

chosen as per Lemma \ref{lem:K_formula}, as long as $\epsilon_{\text{trunc}}$ and $\epsilon_{\text{disc}}$ satisfy $\|u_0\| ( \epsilon_{\mathrm{trunc}}+\epsilon_{\mathrm{disc}})\leq\epsilon_v$.
\end{lemma}
\begin{proof}
    Observing Equation \ref{eq:LCHSdisc}, we see that the total number of Hamiltonian simulations of $e^{-it(kL+H)}$, or unitaries in the LCU expression, is equal to the total number of terms in the double sum, which is $$M = \frac{2KQ}{h},$$ where $[-K,K]$ is the truncation interval on the real line (hence the factor of 2), $Q$ is the number of quadrature points, and $h$ is the time step. By inserting the bound on $K$ from Lemma \ref{lem:K_formula} into the result of Lemma \ref{lem:Qbound}, which was obtained with a choice of time step $h=1/(et\|L\|)$, we obtain the expression for $M$ in the lemma statement. 
    Next, as seen in Eq. \ref{eq:LCHSdisc}, the largest magnitude among the quadrature point values $k_j$ is no more than the truncation threshold $K$, which establishes Eq. \ref{eq:k_bound}.
    Last, we establish what constraints on $\epsilon_{\text{trunc}}$ and $\epsilon_{\text{disc}}$ ensure that 
    $\|v(t) - u(t) \|\leq\epsilon_v$.
    The discretized LCHS formula has the form $\sum_j c_j e^{-it(k_j L +H)}$. For simpler notation, here we write $U_j = e^{-it(k_j L +H)}$. 
    The norm of the difference is then upper bounded as,
    \begin{align}
        \|v(t) - u(t) \| = & \left | \left |  \sum_j c_j U_j u_0 - \int_\mathbb{R} g(k) U(k) \dd k \; u_0  \right | \right | \\
        \leq & \left | \left | \sum_j c_j U_j - \int_\mathbb{R} g(k) U(k) \dd k  \right | \right | \|u_0\| \\
        \leq & \left | \left | \sum_j c_j U_j - \int_{-K}^{K} g(k) U(k) \dd k \right | \right | \|u_0\|  \\
        & + \left | \left | \int_{-K}^K g(k) U(k) \dd k - \int_\mathbb{R} g(k) U(k) \dd k  \right | \right | \|u_0\| \\
        \leq &   \|u_0\| (\epsilon_{\mathrm{disc}} + \epsilon_{\mathrm{trunc}} ).
    \end{align}
Therefore, by ensuring that $\|u_0\| (\epsilon_{\mathrm{disc}} + \epsilon_{\mathrm{trunc}} )\leq\epsilon_v$, we have that $\|v(t) - u(t) \|\leq\epsilon_v$ is satisfied.
\end{proof}

Note that this is bound is also tightened by the result of using \textit{Lambert-W} functions to more tightly bound $K$, given that $Q$ is also a function of $K$. This value $M$ is important as it gives the space complexity required to express the sum as an LCU circuit, which will be used to bound the number of qubits later. The number of terms in the sum also to some extent controls the coefficient 1-norm for the LCHS block encoding, or in other words, the sum of all quadrature coefficients.

\section{Construction and Analysis of the LCHS Implementation}
\label{sec:lchs_implementation}

Equipped with bounds on the parameters of the numerical integral, we can now proceed by constructing a quantum implementation of the LCHS formula:
\begin{equation} \label{eq:propagator}
    \sum_{m=-K/h}^{K/h-1} \sum_{q=1}^Q c_{q,m} U(t,k_{q,m}) = \sum_{j=1}^M c_j U(t, k_j)=\sum_{j=1}^M c_je^{-it(k_j L+H)},
\end{equation}
and as a reminder, $L=(A+A^{\dagger})/2$ and $H=(A-A^{\dagger})/2i$.
Our goals are 1) to construct a block encoding for $\sum_{m=-K/h}^{K/h-1} \sum_{q=1}^Q c_{q,m} U(t,k_{q,m})$ made from block encodings of $A$ and 2) to
establish a constant factor upper bound on the number of queries to the block encoding of $A$ required to ensure a target error rate. 
Compared to the more general case considered in \cite{an2023linear}, since we assume that $A$ is time independent, we can use the block encoding model 
described in Eq. \ref{eq:block_encoding}.

\subsection{Construction of Block Encoding}
Here we summarize our approach to constructing a block encoding of the approximate propagator in Equation \ref{eq:propagator}.
We will save the costing and error accounting for Sections \ref{subsec:select} and \ref{subsec:soln_be}.
The LCHS operator $\mathcal{L}$ is, by design, a linear combination of unitaries.
Therefore, we will use the standard SELECT-PREP approach to block encoding \cite{childs2012hamiltonian},
\begin{align}
    \mathcal{L} = \alpha(\textup{PREP}_l\otimes \mathbb{I})\textup{SELECT}(\textup{PREP}_r^{\dagger}\otimes \mathbb{I}),
\end{align}
where $\alpha$ is the subnormalization and, in our case
\begin{align}
    \textup{SELECT} \equiv \mathcal{Q} = \sum_j\ketbra{j}{j}\otimes e^{-it(k_j L+H)}
\end{align}
and, as will be described below, $\textup{PREP}_l$ and $\textup{PREP}_r$ are used to encode the coefficient vector $c$.

First, we will overview our approach for constructing the SELECT operator $\mathcal{Q}$.
The standard construction might be to construct a controlled unitary for each $e^{-it(k_j L+H)}$ and use a technique such as unary iteration \cite{babbush2018encoding}.
Instead, we develop a more efficient construction by exploiting structure shared among the target unitaries $e^{-it(k_j L+H)}$.
We construct the entire SELECT operator as a single Hamiltonian simulation operator with respect to an effective Hamiltonian.
Let $\mathcal{H}_j = k_jL +H $ be the Hamiltonian in each simulation, then define the effective Hamiltonian operation as:
\begin{equation}\label{eq:sel1}
    \mathcal{S} = \sum_{j=1}^M \ketbra{j}{j} \otimes \mathcal{H}_j
\end{equation}
Note that this effective Hamiltonian can be expressed linearly in $A$, which will be a key property in implementing a block encoding of $\mathcal{S}$ in terms of block encodings of $A$:
\begin{align}
    \mathcal{S} = \sum_{j=1}^M \ketbra{j}{j} \otimes \left(k_j\frac{A+A^{\dagger}}{2}+\frac{A-A^{\dagger}}{2i}\right)
\end{align}
Applying Hamiltonian simulation to this effective Hamiltonian, and assuming that $\sum_j \ketbra{j}{j}=\mathbb{I}$ yields the desired SELECT operator
\begin{align}
    \mathcal{Q} &\equiv \exp{-it \mathcal{S}}\\
    &=\exp{-it\sum_{j=1}^M \ketbra{j}{j} \otimes \mathcal{H}_j}\\
     &= \sum_{j=1}^M \ketbra{j}{j} \otimes e^{-it\mathcal{H}_j}.
\end{align}
The costs and error associated with this step will be established in Lemma \ref{lem:lcu_select}.
Another factor that makes this approach more efficient is the fact that $\mathcal{S}$ can be constructed from just two calls to the block encoding of $A$ (see Lemma \ref{lem:SH}), rather than invoking $A$ for each of the $M$ terms.

Having established a construction for the SELECT operator $\mathcal{Q}$ using the the effective Hamiltonian $\mathcal{S}$, next we describe the PREP operators. 
To formalize these, we use the concept of a state preparation pair.
We follow Ref. \cite{gilyen2019quantum} and provide the following definition:
\begin{definition}[State-Preparation-Pair] \label{def:state_prep_pair}
    Let $z\in \mathbb{C}^M$ and $\gamma \geq \|z\|_1$. The unitaries $(P_{l}, P_{r})$ are called a $(\gamma, m, \epsilon)$ state-preparation-pair for the vector $z$ if $P_{l} \ket{0}^{m} = \sum_{j=1}^M {x_j}\ket{j}$, and $P_{r} \ket{0}^{m} = \sum_{j=1}^M y_j\ket{j}$ such that $\epsilon_{c} \geq \sum_{j=1}^M |\gamma x_j^*y_j - z_j|$. 
\end{definition}

We will assume that the unitaries $(O_{c,l}, O_{c,r})$ form an $(\alpha_c, m_c=\lceil\log_2(M)\rceil, \epsilon_c)$ state preparation pair for the vector $c$. With this, the following circuit can be verified to yield a block encoding of the LCHS propagator 
\begin{align}
    \mathcal{U}=(O_{c,l}^\dagger \otimes \mathbb{I}) \mathcal{Q} (O_{c, r} \otimes \mathbb{I}).
\end{align}

Finally, we show how an approximation to the final solution state is obtained using these operations.
Applying $\mathcal{U}$ to an encoding of the initial vector $O_u \ket{0^{m_u}} = \ket{u_0} = u_0/\|u_0\|$ and ancilla qubits in $\ket{0^{m_c}}$, and conditioning on the block encoding ancilla qubits being in the state $\ket{0^{m_c}}$, yields the approximate solution vector $v(t)\approx u(t)$:
\begin{align}
        (\bra{0^{m_c}}\otimes\mathbb{I})(O_{c,l}^\dagger \otimes \mathbb{I}) \mathcal{Q} (O_{c, r} \otimes \mathbb{I}) (\mathbb{I} \otimes O_u ) (\ket{0^{m_c}}\otimes \ket{0}^{m_u} )&= \sum_{j=1}^M\bra{0^{m_c}}O_{c,l}^\dagger\ketbra{j}{j}O_{c, r}\ket{0^{m_c}} e^{-it\mathcal{H}_j}  (\mathbb{I} \otimes O_u\ket{0}^{m_u} )\\
        &= \frac{1}{\alpha_c\|u_0\|}\sum_{j=1}^M c_j e^{-it\mathcal{H}_j}u_0\\
        &= \frac{1}{\alpha_c\|u_0\|}v(t).
\end{align}
The costs and error incurred in implementing the block encoding of $v(t)$, and hence of $u(t)$, will be established in Theorem \ref{lem:lchs_be}.
It may be desirable to prepare a state proportional to the solution vector with success probability arbitrarily close to one, rather than having a conditional preparation.
Such a preparation can be achieved using amplitude amplification, and so we will refer to this as an {amplified} state preparation as opposed to the conditional state preparation.
An amplified state preparation can help reduce the overall runtime costs of estimating properties of the solution state \cite{penuel2024feasibility}.
 
Alternatively, one can view the construction of the amplified solution state as a means of reducing the sub-normalization of the block encoding.
Loosely, this enables a stronger statistical signal when estimating properties of this vector.
Corollary \ref{thm:amplified_prep} describes how robust fixed-point oblivious amplitude amplification can be used to achieve this.

Critically, the methods used to construct the imperfect block encoding of $u(t)$ enable systematic reduction of all error sources, assuming the error in the block encoding of $A$ can be controlled.
The essential content of the main results in this section is the accounting of these controllable errors.
A road map of the overall construction is as follows:
\begin{enumerate}
    \item (Lemma \ref{lem:SH}) Construct block encoding of $    \mathcal{S} = \sum_{j=1}^M \ketbra{j}{j} \otimes (k_jL+H)$ from block encoding of $A$.
    \item (Lemma \ref{lem:lcu_select}) Construct LCU SELECT operator $\mathcal{Q} = \sum_{j=1}^M \ketbra{j}{j} \otimes e^{-it\mathcal{H}_j}$ from block encoding of $    \mathcal{S}$.
    \item (Lemma \ref{lem:lchs_be}) Construct the LCHS unitary $\mathcal{U}$ by linearly combining $\exp(-it\mathcal{H}_j)$ with coefficients $c_j$ using $\mathcal{Q}=\sum_{j=1}^M \ketbra{j}{j} \otimes e^{-it\mathcal{H}_j}$ as SELECT and $O_{c,r}$ and $O_{c,l}$ as PREP unitaries (or a state preparation pair).
    \item (Theorem \ref{thm:amplified_prep}) Construct the amplified solution state preparation using amplitude amplification with $\mathcal{U}$ and the block encoding of $u(0)$.
\end{enumerate}

\subsection{SELECT circuit construction}
\label{subsec:select}

The following lemma establishes a construction for the block encoding of the effective Hamiltonian $\mathcal{S} = \sum_{j=1}^M \ketbra{j}{j} \otimes \mathcal{H}_j$. It is then used to generate the controlled time evolution operator $\mathcal{Q} =\sum_{j=1}^M \ketbra{j}{j} \otimes e^{-it\mathcal{H}_j}$. 
\begin{lemma}[$\mathcal{S}$ block encoding]\label{lem:SH}
Assuming access to $U_A$, an $(\alpha_A, m_A, \epsilon_A)$-block encoding of $A$, and $\epsilon_R$-accurate multi-controlled $R_Z(\phi)$ gates,
there exists a 

$(\sqrt{1+k_{\textup{max}}^2}\alpha_A, m_A+2, \epsilon)$-block encoding of $\mathcal{S}=\sum_{j=1}^M \ketbra{j}{j} \otimes (k_jL +H)$, where $k_{\textup{max}}=\max_j |k_j| \leq K$, as long as the error parameters satisfy
\begin{align}
\epsilon_{A}\sqrt{1+k_{\textup{max}}^2}+2M\sqrt{1+k_{\textup{max}}^2}\epsilon_R\alpha_{A}\leq\epsilon.
\end{align}
The construction uses one call each to $U_A$ and to $U_A^{\dagger}$, and $2M$ multi-controlled rotation calls, where $M$ is the total number of points in the LCHS integral bounded in Lemma \ref{lem:class_error}.
\end{lemma}
\begin{proof}

First, we establish a decomposition of the operator $\mathcal{S}$ into the operators $A$ and $D=\sum_j k_j\ketbra{j}$:
\begin{align}
    \mathcal{S} &= \sum_{j=1}^M \ketbra{j}{j} \otimes (k_jL +H)\\
    &= \mathbb{I}\otimes H + \left(\sum_{j=1}^M k_j\ketbra{j}{j} \right)\otimes L\\
    &= \mathbb{I}\otimes \frac{(A-A^\dagger)}{2i} + \left(\sum_{j=1}^M k_j\ketbra{j}{j} \right)\otimes \frac{A+A^\dagger}{2} \\
    &= \mathbb{I}\otimes \frac{(A-A^\dagger)}{2i} + D\otimes \frac{A+A^\dagger}{2} \\    
    &= \frac{-i\mathbb{I}+D}{2}\otimes A + \frac{i\mathbb{I}+D}{2}\otimes A^\dagger \\
    &= \frac{1}{2} \left ( (-i\mathbb{I}+D) \otimes A + (i\mathbb{I}+D) \otimes A^\dagger \right ).
    \label{eq:SH_block_encoding}
\end{align}

To block encode $\mathcal{S}$, we first notice that $-i\mathbb{I}+D$ is diagonal in the computational basis.
This lets us use the diagonal block encoding method of Prop. 12 in \cite{takahira2021quantum}.
With $M$, the number of diagonal entries, this block encoding is constructed from $2M$ multi-controlled $R_Z(\phi)$ gates, each with error $\epsilon_R$.
The block encoding parameters in this method depend on the largest magnitude of the diagonal entries.
In our case, defining $k_{\textup{max}}\equiv\max_j |k_j|$, the largest magnitude entry is $\max_j|-i+k_j|=\sqrt{1+k^2_{\textup{max}}}$.
Prop. 12 in \cite{takahira2021quantum} then constructs a
$(\sqrt{1+k^2_{\textup{max}}}, 1, 2M\sqrt{1+k^2_{\textup{max}}}\epsilon_R)$ block encoding of $-i\mathbb{I}+D$. Note that while \cite{takahira2021quantum} 
assumes perfect gates and arrives at a zero-error block encoding, we have accounted for the error in the $2M$ multi-controlled $R_Z(\phi)$ gates using a triangle inequality upper bound.

Next, we form a block encoding of the tensor product
\begin{align}
(-i\mathbb{I}+D)\otimes A.
\end{align}
Lemma 30 of \cite{gilyen2019quantum} provides such an implementation, yielding a $(\alpha_{A}\sqrt{1+k_{\textup{max}}^2}, m_{A}+1, \epsilon_{A}\sqrt{1+k_{\textup{max}}^2}+2M\sqrt{1+k^2_{\textup{max}}}\epsilon_R\alpha_{A})$ block encoding of $(-i\mathbb{I}+D)\otimes A$.
Finally, we use the above block encoding and its Hermitian conjugate to form a block encoding of their sum.
Noting that, $(H, H)$ forms a $(1,1,0)$ state preparation pair for the vector $[\frac{1}{2}, \frac{1}{2}]$, Proposition 5 of \cite{takahira2021quantum} gives a 
\begin{align}
(\alpha_{A}\sqrt{1+k_{\textup{max}}^2},m_{A}+2, \epsilon_{A}\sqrt{1+k_{\textup{max}}^2}+2M\sqrt{1+k^2_{\textup{max}}}\epsilon_R\alpha_{A})
\end{align}
block encoding of 
\begin{align}
    \mathcal{S}=\frac{1}{2} \left ( (-i\mathbb{I}+D) \otimes A + (i\mathbb{I}+D) \otimes A^\dagger \right ).
\end{align}
This yields the block encoding with error $\epsilon$ as long as
\begin{align}
\epsilon_{A}\sqrt{1+k_{\textup{max}}^2}+2M\sqrt{1+k^2_{\textup{max}}}\epsilon_R\alpha_{A}\leq\epsilon.
\end{align}
\end{proof}

The above result yields all necessary block encoding parameters to consider a time evolution of the select operator. Now we use the above construction for the block encoding of $\mathcal{S}$ to establish a block encoding for $\mathcal{Q}(t)=e^{-it\mathcal{S}}$.

\begin{lemma}[LCU SELECT Operator]
\label{lem:lcu_select}
There exists a $(1,m_A+4,\epsilon_Q)$-block encoding of the $\textup{SELECT}$ operator 
$\mathcal{Q}(t) = \sum_{j=1}^M \ketbra{j}{j} \otimes e^{-it(k_j L + H)}=e^{-it\mathcal{S}}$
via qubitization using the following subroutines with the associated query counts:
\begin{itemize}
\item controlled implementations of $U_A$, an $(\alpha_A, m_A, \epsilon_A)$-block encoding of $A$, with 
\begin{align} \label{eq:qubitize_queries}
\ceil{e \sqrt{1+K^2}\alpha_A t + 2\ln \left (\frac{2\eta}{\epsilon_{\mathrm{exp}}}\right ) }
\end{align}
queries to $c$-$U_A$ and six queries to $c$-$c$-$U_A$, with $\eta= 4(\sqrt{2\pi}e^{1/13})^{-1} \approx 1.47762.$
\item $\epsilon_R$-accurate multi-controlled $R_Z(\phi)$ gates, with query count of $M$ times the number of calls to $c$-$U_A$ as given above
\end{itemize}
as long as the error parameters $\epsilon_{\mathrm{exp}}$, $\epsilon_A$, and $\epsilon_R$ satisfy
\begin{align}
\label{eq:lcu_select_error_budget}
\epsilon_{\mathrm{exp}}+\sqrt{1+K^2}t\epsilon_A+ 2Mt\sqrt{1+K^2}\alpha_A\epsilon_R\leq \epsilon_Q.    
\end{align}
\begin{proof}
Lemma \ref{lem:flex_ham_sim} establishes that one can implement a $(1, m+2, \epsilon)$ block encoding of $e^{-iHt}$ using $\ceil{\frac{e}{2}\alpha_\mathcal{S} t + \ln \left (\frac{2c}{\epsilon_{\mathrm{exp}}}\right ) }$ controlled queries to $U_H$ or $U_H^\dagger$, a $(\alpha, m, \epsilon_H)$ block encoding of a hermitian operator $H$, as long as $\epsilon_{\mathrm{exp}}+t\epsilon_H\leq\epsilon$, where $\epsilon_H$ is the block encoding error of the Hamiltonian. Note that we are neglecting the additional 2 queries in Lemma \ref{lem:flex_ham_sim} as they are a trivial addition to the complexity. Our proof follows by modifying this result for the LCHS select operator $\mathcal{S}$.

Let $U_{\mathcal{S}}$ be the block encoding of the Hamiltonian $\mathcal{S}$ in this case. Note that it meets the necessary criteria as it is Hermitian. Then we can use the referenced theorems to encode a block encoding of the following:
\begin{align}
    e^{-it\mathcal{S}} &= e^{-it\sum_{j=1}^M \ketbra{j}{j} \otimes k_jL +H} \\
    &= \sum_{j=1}^M \ketbra{j}{j} \otimes e^{-it(k_jL +H)} \equiv \mathcal{Q}(t).
\end{align}
Applying Lemma \ref{lem:flex_ham_sim}, an $\epsilon$-accurate block encoding of $e^{-it\mathcal{S}}$ can be implemented using no more than $\ceil{\frac{e}{2}\alpha_\mathcal{S} t + \ln \left (\frac{2\eta}{\epsilon_{\mathrm{exp}}}\right ) }$ controlled queries to $U_{\mathcal{S}}$ and $U_{\mathcal{S}}^\dagger$ (each of which requires 2 queries to $U_A$), as long as we ensure that
$\epsilon_{\mathrm{exp}}+t\epsilon_{\mathcal{S}}\leq\epsilon$. We can see this by decomposing the error using a series of triangle inequalities:
\begin{align}
    \| \mathcal{Q}(t, A) - \widetilde{\mathcal{Q}}(t, \widetilde{A}) \| & \leq \|\mathcal{Q}(t, A) - {\mathcal{Q}}(t, \widetilde{A})\| + \|\mathcal{Q}(t, \widetilde{A}) - \widetilde{\mathcal{Q}}(t, \widetilde{A}) \| \\
    & = \|e^{-it\cal{S}} - e^{-it\widetilde{\cal{S}}} \| + \epsilon_{\mathrm{exp}} \\
    & \leq t \|\cal{S} - \widetilde{\cal{S}}\| +\epsilon_{\mathrm{exp}} \\
    &= t\epsilon_{A}\sqrt{1+k_{\textup{max}}^2}+2Mt\sqrt{1+k^2_{\textup{max}}}\epsilon_R\alpha_{A} + \epsilon_{\mathrm{exp}}\\
    &\leq t\epsilon_{A}\sqrt{1+K^2}+2Mt\sqrt{1+K^2}\epsilon_R\alpha_{A} + \epsilon_{\mathrm{exp}},
\end{align} 
where we use $\widetilde{\cal{Q}}$ to mean the approximate unitary produced by qubitization and $\widetilde{\cal{S}}$ to mean the imperfect block encoding of $\cal{S}$, as well as leverage that $k_{\textup{max}}\leq K$.
In the third line we use the inequality $\| e^{-iHt} - e^{-i\widetilde{H}t}\| \leq t\|H - \widetilde{H}\|$ from Ref. \cite{chakraborty2018power}, and then recognize that this is the block encoding error of $\cal{S}$ in the fourth line. Finally, we use 
Lemma \ref{lem:SH}, which shows that each query to $U_{\mathcal{S}}$ uses two queries to $c$-$U_A$ and $M$ times that many queries to $\epsilon_R$-accurate multi-controlled $R_Z(\phi)$ gates, and that the parameters of the $\mathcal{S}$ block encoding are, 
\begin{align}
\alpha_{\mathcal{S}}&=\sqrt{1+k_{\textup{max}}^2}\alpha_A\\
m_{\mathcal{S}}&=m_A+2\\ 
\epsilon_{\mathcal{S}}&=\sqrt{1+k_{\textup{max}}^2}\epsilon_A+ 2M\sqrt{1+k_{\textup{max}}^2}\alpha_A\epsilon_R.
\end{align}
From this we have that robust Hamiltonian simulation can be used to implement a
\begin{align}
    (1,m_A+4,\epsilon)
\end{align}
block encoding of $e^{-it\mathcal{S}}$
using the query counts in the Lemma statement (using that $k_{\textup{max}}\leq K$), as long as the error parameters satisfy
the inequality in Equation \ref{eq:lcu_select_error_budget}.
\end{proof}
\end{lemma}

\subsection{Solution vector block encodings}
\label{subsec:soln_be}
Building from the previous lemmas, we establish a series of results that provide the costs and error budget requirements for block encodings of the solution vector. 
To start, Lemma \ref{lem:lchs_be} analyzes the task of conditionally preparing a state proportional to the solution vector. 
That is, given the output state, the solution vector can be prepared by measuring a register and, with a \textit{pre-fixed} failure-probability, post-selecting on a certain outcome.
While this does yield a block encoding of the solution vector, it may be preferable to implement the state preparation with a controllable failure probability, enabling the user to pay an implementation cost to set the failure rate to be arbitrarily low.
This can be achieved using amplitude amplification.
We are also motivated to consider the amplitude amplification costs in order to establish a fair comparison between the costs of the LCHS algorithm and the RLS algorithm; we assess the cost of implementing block encodings with with the same probability of success (or, nearly equivalently, with the same subnormalization).
Finally, Theorem \ref{thm:amplified_prep}, Corollary \ref{coro:error_budgeted}, and Theorem \ref{thm:error} analyze the costs of this amplified state preparation in several different settings.

\begin{lemma}
    [LCHS Block Encoding] 
\label{lem:lchs_be}
Consider the task of preparing a state on the quantum computer that is proportional to $u(t) = e^{At}u_0$.
Let the LCHS propagator of Eq. \ref{eq:LCHSdisc} be labeled as $\mathcal{L}\equiv \sum_{j=1}^M c_je^{-it(k_j L + H)}$ and let $c$ be the vector of coefficients $c_j$.
We can construct a $(\|c\|_1\|u_0\|,[m_{c}+m_Q,m_{c}+m_Q+m_N], \epsilon)$-block encoding of $u(t) \in\mathbb{C}^{2^n}$ using a single call to each of
\begin{itemize}
\item $U_Q$, an $(\alpha_Q, m_Q, \epsilon_Q)$ block encoding of $\mathcal{Q}(t) = \sum_{j=1}^M \ketbra{j}{j} \otimes e^{-it(k_j L + H)}$, 
\item $(O_{c,l}, O_{c,r})$, an $m_{c}$-qubit $(\epsilon_{c})$ state preparation pair for the coefficient vector $c$ (see Definition \ref{def:state_prep_pair}), 
\item $U_0$, a $(\|u_0\|,[0,m_N],\epsilon_0)$ block encoding of $u_0$,
\end{itemize}
as long as the error parameters $\epsilon_Q$, $\epsilon_{c}$, $\epsilon_0$, and $\epsilon_v$ satisfy 
\begin{align}
\epsilon_v + \|u_0\|(\epsilon_{c}+\|c\|_1\epsilon_Q)+\|c\|_1\epsilon_0\leq \epsilon,
\end{align}
where $\|u(t)-\mathcal{L}u_0\|\leq \epsilon_v$.
The block encoding yields a conditional state preparation that succeeds with probability lower bounded by
    \begin{align}
    \textbf{\textup{Pr}}(\textup{success}) &\geq\frac{(\|u(t)\|-\epsilon')^2}{\|u_0\|^2 \|c\|_1^2},
\end{align}
where $\epsilon'=\epsilon_v + \|u_0\|(\epsilon_{c}+\|c\|_1\epsilon_Q)+\|c\|_1\epsilon_0$.
\end{lemma}
\begin{proof}
The operator $\mathcal{Q}(t) = \sum_{j=1}^M \ketbra{j}{j} \otimes e^{-it(k_j L + H)}$ serves as a SELECT operation in an LCU circuit, as it is a sum of controlled unitaries. Now, using $(O_{c,l},O_{c, r})$, assumed to be a $(\|c\|_1,m_{c},\epsilon_{c})$-state preparation pair for the vector $c$,
we can construct the state preparation block encoding $U_u=(O_{c,l}^\dagger \otimes \mathbb{I}) \mathcal{Q} (O_{c, r} \otimes \mathbb{I})$.
We label the imperfect implementations of the unitaries as $\tilde{Q}$, and similarly for others. 
With this, can establish that $\tilde{U}_u$ is a $(\|c\|_1, m_{c}+m_Q, \epsilon_{c}+\|c\|_1\epsilon_Q)$ block encoding of the LCHS propagator $\mathcal{L}$: 
\begin{align}
        \left\|\mathcal{L}-\|c\|_1\bra{0^{ m_{c}+m_Q}}\tilde{U}_u  \ket{0^{m_{c}+m_Q}} \right\| 
        \leq & \left\|\mathcal{L}-\|c\|_1\bra{0^{m_{c}}}\tilde{O}_{c,l}^\dagger\mathcal{Q} \tilde{O}_{c, r} \ket{0^{m_{c}}} \right\|\\
        &+\|c\|_1\left\|\bra{0^{m_{c}}}\tilde{O}_{c,l}^\dagger\mathcal{Q} \tilde{O}_{c, r} \ket{0^{m_{c}}}-\bra{0^{m_{c}+m_Q}}\tilde{U}_u \ket{0^{m_{c}+m_Q}} \right\|\\
        = &\left\|\sum_{j=1}^M c_je^{-it(k_j L + H)}-\|c\|_1\bra{0^{m_{c}}}\tilde{O}_{c,l}^\dagger\ketbra{j}{j} \tilde{O}_{c, r} \ket{0^{m_{c}}}e^{-it(k_j L + H)} \right\|\\
        &+\|c\|_1\left\|\mathcal{Q} -\bra{0^{m_Q}}\tilde{U}_Q\ket{0^{m_Q}} \right\|\\
        \leq & \sum_{j=1}^M \left\|c_j-\|c\|_1\bra{0^{m_{c}}}\tilde{O}_{c,l}^\dagger\ketbra{j}{j} \tilde{O}_{c, r} \ket{0^{m_{c}}} \right\|+\|c\|_1\epsilon_Q\\
        \leq & \epsilon_{c}+\|c\|_1\epsilon_Q,
\end{align}
where we have used the triangle inequality, the unitarity of $e^{-it(k_j L + H)}$, and the definition of state preparation pair.
Finally,  we consider a block encoding of the solution state. We have by assumption that $\tilde{U}_0$ satisfies
$|u_0-\|u_0\| \tilde{U}_0\ket{0^{n}}|\leq \epsilon_0$.
Then $\tilde{U}_u$ can be applied to this approximate initial vector block encoding to conditionally prepare an approximation to $\mathcal{L}u_0$. 
Specifically, we have that $\tilde{U}_u\tilde{U}_0$ is a 
\begin{align}
(\|c\|_1\|u_0\|,[m_{c}+m_Q,m_{c}+m_Q+m_N],\|u_0\|(\epsilon_{c}+\|c\|_1\epsilon_Q)+\|c\|_1\epsilon_0) 
\end{align}
block encoding of $\mathcal{L}u_0$.

\begin{align}
\label{eq:v(t)_error}
\left\|\mathcal{L}u_0-\right.&\left.\|u_0\|\|c\|_1 \,
 \left( \bra{0^O 0^Q} \otimes \mathbb{I} \right) \tilde{U}_u \tilde{U}_0 \left( \ket{0^O 0^Q 0^S} \otimes \mathbb{I} \right) \right\|  \\
\leq & \left\|\mathcal{L}u_0-\|c\|_1 
\left( \bra{0^O 0^Q} \otimes \mathbb{I} \right) \tilde{U}_u \left( \ket{0^O 0^Q} \otimes \mathbb{I} \right)
u_0\right\|\\
&+\left\|\|c\|_1 
\left( \bra{0^O 0^Q} \otimes \mathbb{I} \right) \tilde{U}_u \left( \ket{0^O 0^Q} \otimes \mathbb{I} \right)
u_0-\|u_0\|\|c\|_1 \left( \bra{0^O 0^Q} \otimes \mathbb{I} \right) \tilde{U}_u \tilde{U}_0 \left( \ket{0^O 0^Q 0^S} \otimes \mathbb{I} \right)\right\|\\
\leq & \left\|\mathcal{L}-\|c\|_1 \left( \bra{0^O 0^Q} \otimes \mathbb{I} \right) \tilde{U}_u \left( \ket{0^O 0^Q} \otimes \mathbb{I} \right)\right|\|u_0\|+\|c\|_1\left| u_0-\|u_0\| \tilde{U}_0\ket{0^{n}}\right\|\\
\leq & \|u_0\|(\epsilon_{c}+\|c\|_1\epsilon_Q)+\|c\|_1\epsilon_0
\label{eq:L_bias}
\end{align}
Next, if $|u(t)-\mathcal{L}u_0|\leq \epsilon_v$, then by the triangle inequality we have that $\tilde{U}_u\tilde{U}_0$ is also a 
\begin{align}
(\|c\|_1\|u_0\|,[m_{c}+m_Q,m_{c}+m_Q+m_N], \epsilon_v + \|u_0\|(\epsilon_{c}+\|c\|_1\epsilon_Q)+\|c\|_1\epsilon_0)
\end{align}
block encoding of $u(t)$.
The probability that the conditional quantum state $\left( \bra{0^O 0^Q} \otimes \mathbb{I} \right) \tilde{U}_u \tilde{U}_0 \left( \ket{0^O 0^Q 0^S} \otimes \mathbb{I} \right)$ is obtained can be lower bounded as
\begin{align}
    \textbf{Pr}(\text{success}) &= \left\|\left( \bra{0^O 0^Q} \otimes \mathbb{I} \right) \tilde{U}_u \tilde{U}_0 \left( \ket{0^O 0^Q 0^S} \otimes \mathbb{I} \right)\right\|^2\\
 &\geq\frac{(\|u(t)\|-\epsilon')^2}{\|u_0\|^2 \|c\|_1^2},
\end{align}
where $\epsilon'=\epsilon_v + \|u_0\|(\epsilon_{c}+\|c\|_1\epsilon_Q)+\|c\|_1\epsilon_0$ and in the second line we have used the fact that $|u(t)-\|c\|_1\|u_0\| \left( \bra{0^O 0^Q} \otimes \mathbb{I} \right) \tilde{U}_u \tilde{U}_0 \left( \ket{0^O 0^Q 0^S} \otimes \mathbb{I} \right)|\leq\epsilon'$ implies that $|\left( \bra{0^O 0^Q} \otimes \mathbb{I} \right) \tilde{U}_u \tilde{U}_0 \left( \ket{0^O 0^Q 0^S} \otimes \mathbb{I} \right)|\geq \frac{|u(t)|-\epsilon'}{\|u_0\|\|c\|_1}$ by the reverse triangle inequality.
\end{proof}

We discuss how the above result is useful for an important aspect of simulating differential equations on a quantum computer: norm estimation. In most applications of such algorithms, it is important to extract an estimate of some property of the solution state $u(t)$.
By encoding the desired property into an amplitude estimate (see e.g. \cite{penuel2024feasibility}), the quantum computer can only give an estimate of the desired property for the normalized vector $u(t)/\|u(t)\|$.
So, the resulting estimate must be rescaled by $\|u(t)\|$ to yield an estimate of the desired property.
Fortunately, in some cases the quantum computer can be used to obtain an estimate of $\|u(t)\|$.
In particular, the block encoding provided in Lemma \ref{lem:lchs_be} can be used to estimate $\|u(t)\|$ as long as sufficiently accurate estimates of $\|u_0\|$ is known ($\|c\|_1$ can be computed exactly).
This is achieved as follows.
A Hadamard test or, more efficiently, amplitude estimation \cite{brassard2000quantum} can be used to estimate $\left\|\left( \bra{0^O 0^Q} \otimes \mathbb{I} \right) \tilde{U}_u \tilde{U}_0 \left( \ket{0^O 0^Q 0^S} \otimes \mathbb{I} \right)\right\|^2$.
The proof of Lemma \ref{lem:lchs_be} shows that an estimate of the above quantity is an estimate of 
$\|\mathcal{L}u_0\|$ with bias no larger than the controllable amount $(\|u_0\|(\epsilon_{c}+\|c\|_1\epsilon_Q)+\|c\|_1\epsilon_0)$ as given by Eq. \ref{eq:L_bias}.
Furthermore, this also gives an estimate of
$\|u(t)\|$ with bias at most $\epsilon'$, which can also be controlled.
In this way, the parameters $\|\mathcal{L}u_0\|$ and $\|u(t)\|$, which may be needed for different purposes can be estimated using the block encoding from Lemma \ref{lem:lchs_be}.

We also address a potential concern one might have in regards to how the block encoding subnormalization might scale with $t$.
As we will show in the subsequent theorem, the subnormalization of this ``unamplified'' block encoding of $u(t)$ governs the cost of using amplitude amplification to increase the success probability of preparing the target state.
The concern might then be that the normalization $\|c\|_1 = \sum_{j=1}^M |c_j|$ scales poorly with $t$, causing the cost of amplitude amplification to scale poorly with $t$.
This is a valid concern because the number of terms in $c$ scales as
$M \in O\left (t \|L\| \log \left ( \frac{1}{\epsilon}\right )^{1+1/\beta} \right )$. However, applying [\cite{an2023linear}, Lemma 11], regardless of the value of $M$ the nature of approximating a smooth function with Gaussian quadrature implies that $\|c\|_1 \in O(1)$. This implies that long time simulations requiring large $M$ will not lead to a corresponding increase in cost of amplitude amplification. To give accurate costs for the entire approach, we show how to numerically compute these coefficients exactly in Section \ref{sec:numerics} as the goal of this work is to compute query counts that are sensitive to all constants.

Now, the final step of the constant factor analysis is to calculate the number of calls to LCHS, and therefore the block encoding $U_A$, in order to ensure a failure rate below $\epsilon$ for the algorithm.

\begin{theorem}[Amplified output] \label{thm:amplified_prep}
We can construct a $(\|\mathcal{L}u_0\|, [m_{c}+m_A+5,m_N+m_{c}+m_A+5], \epsilon)$-block encoding of $u(t)$
using the following block encoding subroutines with the associated query counts: 
\begin{itemize}
\item $(O_{c,l}, O_{c,r})$, an $m_{c}$-qubit $\epsilon_{c}$ state preparation pair for the coefficient vector $c$ (see Definition \ref{def:state_prep_pair}) and $U_0$, a $(\|u_0\|,[0,m_N],\epsilon_0)$ block encoding of $u_0$, both with query count
\begin{align}
\label{eq:aa_cost}
C_{_{LCHS}}(\Delta,\epsilon_{\mathrm{AA}})=\ceil{\sqrt{8\ceil{\frac{4}{\Delta^2}\ln(8/\pi \epsilon_{\mathrm{AA}}^2) e^2 }\ln\left (\frac{64\frac{\sqrt{2}}{\Delta}\ln^{1/2}(8/\pi \epsilon_{\mathrm{AA}}^2)}{3\sqrt{\pi} \epsilon_{\mathrm{AA}}}\right )}+1},\,\,\,\, \textup{where}\,\, \Delta=2\frac{\|u(t)\|-\epsilon_{_{LCHS}}}{\|u_0\|\, \|c\|_1}
\end{align}
\item $c$-$U_A$, an $(\alpha_A, m_A, \epsilon_A)$-block encoding of $A$ with query count
\begin{align}
C_A=C_{_{LCHS}}(\Delta,\epsilon_{\mathrm{AA}})\left(e\sqrt{1+K^2}\alpha_A t + 2\log(\frac{2 \eta}{\epsilon_{\mathrm{exp}}})\right),
\end{align}
where $\eta= 4(\sqrt{2\pi}e^{1/13})^{-1} \approx 1.47762.$
\item $\epsilon_R$-accurate multi-controlled $R_Z(\phi)$ gates, with query count
\begin{align}
C_R = MC_A
\end{align}
as long as the error parameters $\epsilon_{A}$, $\epsilon_{\text{exp}}$, $\epsilon_{c}$, $\epsilon_0$, $\epsilon_{\text{trunc}}$, $\epsilon_{\text{disc}}$, and $\epsilon_{\mathrm{AA}}$ satisfy
\begin{align}
    \|u_0\|(\epsilon_{\mathrm{trunc}}&+\epsilon_{\mathrm{disc}})+\left(\|u(t)\|+\|u_0\|(\epsilon_{\mathrm{trunc}}+\epsilon_{\mathrm{disc}}) \right)\left(\epsilon_{\mathrm{AA}}+\frac{9C_{_{LCHS}}(\Delta,\epsilon_{\mathrm{AA}})}{2\|c\|_1\|u_0\|}\left( \|u_0\|\left(\epsilon_{c}\right. \right. \right. \nonumber\\ & \left. \left. \left. +\|c\|_1(\epsilon_{\mathrm{exp}}+\sqrt{1+K^2}t\epsilon_A+ 2Mt\sqrt{1+K^2}\alpha_A\epsilon_R)\right)+\|c\|_1\epsilon_0\right)\right)\leq \epsilon.
\end{align}
and as long as $\epsilon\leq 3C_{_{LCHS}}(\Delta,\epsilon_{\mathrm{AA}})\left(\|u(t)\|+\|u_0\|(\epsilon_{\mathrm{trunc}}+\epsilon_{\mathrm{disc}}) \right)/8$, which is used to ensure the performance of the amplitude amplification step.
\end{itemize}
\end{theorem}
\begin{proof}
Our goal is to establish a quantum algorithm $U$ such that $\|u(t)-\alpha\bra{0}^aU\ket{0}^b\|\leq \epsilon$, for some subnormalization $\alpha$ and some $a$- and $b$-qubit registers.
Letting $\mathcal{L}u_0$ be the LCHS approximation to $u(t)$, then by the triangle inequality we have that
$\|u(t)-\alpha\bra{0}^aU\ket{0}^b\|\leq \|u(t)-\mathcal{L}u_0\|+\|\mathcal{L}u_0-\alpha\bra{0}^aU\ket{0}^b\|\equiv\epsilon_v+\epsilon_{\mathrm{qc}}$.
Lemma \ref{lem:class_error} ensures that $\mathcal{L}u_0$ is an $\epsilon_v$-accurate approximation to $u(t)$, provided that $\|u_0\|(\epsilon_{\mathrm{trunc}}+\epsilon_{\mathrm{disc}})\leq\epsilon_v$.
Thus, an $(\alpha, [a,b],\epsilon_{\textup{qc}})$ block encoding of $\mathcal{L}u_0$ is also an $\epsilon$-accurate block encoding of $u(t)$ by ensuring that $\|u_0\|(\epsilon_{\mathrm{trunc}}+\epsilon_{\mathrm{disc}})+\epsilon_{\mathrm{qc}}\leq\epsilon$. Later in the proof we will account for the costs associated with $\epsilon_{\mathrm{trunc}}$ and $\epsilon_{\mathrm{disc}}$.

The quantum algorithm $U$ uses amplitude amplification to produce an amplified version of the conditional state that is output from LCHS (as described in Lemma \ref{lem:lchs_be}).
We use Lemma \ref{lem:robust_fp_obl_amp} for the costs of robust fixed-point oblivious amplitude amplification.
The costs of amplitude amplification depend on the given state preparation routine, which in this case is the conditional state block encoding properties.
As we've already accounted for the LCHS error $\epsilon_v$ above, we will cost out amplitude amplification in targeting the approximate solution $v(t)=\mathcal{L}u_0$, as opposed to $u(t)$.
This is achieved by setting $\epsilon_v$ of Lemma \ref{lem:lchs_be} to zero.
Lemma \ref{lem:lchs_be} establishes a
\begin{align}
(\|c\|_1\|u_0\|,[m_{c}+m_Q,m_{c}+m_Q+m_N], \epsilon_{_{LCHS}})
\end{align}
block encoding of $\mathcal{L}u_0$ 
with an overlap lower bound of $\frac{\Delta}{2} \leq \frac{\|u(t)\|-\epsilon_{_{LCHS}}}{\|u_0\|\, \|c\|_1} \leq \sqrt{\text{Pr}(\ket{0}^{m_N})}$, where, to ensure the correct algorithm performance, the amplitude amplification gap parameter $\Delta$ (see Lemma \ref{lem:robust_fp_obl_amp}) is set to twice the overlap lower bound.
Therefore, applying the amplitude amplification results of Lemma \ref{lem:robust_fp_obl_amp},
we obtain a 
\begin{align}
(\|\mathcal{L}u_0\|_1,[m_{c}+m_Q+1,m_{c}+m_Q+m_N+1], \epsilon_{\mathrm{qc}})
\end{align}
block encoding of $\mathcal{L}u_0$ using 
\begin{align}
C_{_\mathrm{LCHS}}(\Delta,\epsilon_{\mathrm{AA}})=\ceil{\sqrt{8\ceil{\frac{4}{\Delta^2}\ln(8/\pi \epsilon_{\mathrm{AA}}^2) e^2 }\ln\left (\frac{64\frac{\sqrt{2}}{\Delta}\ln^{1/2}(8/\pi \epsilon_{\mathrm{AA}}^2)}{3\sqrt{\pi} \epsilon_{\mathrm{AA}}}\right )}+1}
\end{align}  
calls to the LCHS block encoding, where $\Delta$ is chosen such that $\Delta\leq2\frac{\|u(t)\|-\epsilon_{_{LCHS}}}{\|u_0\|\, \|c\|_1}$, and we ensure that the error parameters satisfy
\begin{align}
\|\mathcal{L}u_0\|\left(\epsilon_{\mathrm{AA}} +\frac{9C_{_{LCHS}}(\Delta,\epsilon_{\mathrm{AA}})\epsilon_{_{LCHS}}}{2\|c\|_1\|u_0\|}\right)\leq
(\|u(t)\| + \epsilon_v) \left ( \epsilon_{\mathrm{AA}} + \frac{{9\epsilon_{\mathrm{LCHS}}}}{2\|c\|_1 \|u_0\| } C_{\mathrm{LCHS}}(\Delta, \epsilon_{\mathrm{AA}})\right )\leq\epsilon_{\mathrm{qc}},
\label{eq:looser_error_bound}
\end{align}
where we have used a triangle inequality and the definition $\epsilon_v=\|u(t) - \mathcal{L}u_0\|$ to establish the first inequality.
Then, with the assumption $\epsilon\leq 3C_{_{LCHS}}(\Delta,\epsilon_{\mathrm{AA}})(\|u(t)\|+\epsilon_v)/8$ from the theorem statement and the fact that $\epsilon_{\textup{qc}}\leq\epsilon$, this ensures the bound of $\epsilon_{_{LCHS}}/\|c\|_1\|u_0\|\leq 1/12$ as required by Lemma \ref{lem:robust_fp_obl_amp}.

Next, we compile the LCHS block encoding into its subroutines according to Lemma \ref{lem:lchs_be}.
It uses a single call to $U_Q$, an $(\alpha_Q, m_Q, \epsilon_Q)$ block encoding of $\mathcal{Q}(t) = \sum_{j=1}^M \ketbra{j}{j} \otimes e^{-it(k_j L + H)}$, a single call to $(O_{c,l}, O_{c,r})$, an $m_{c}$-qubit $(\epsilon_{c})$ state preparation pair for the coefficient vector $c$ (see Definition \ref{def:state_prep_pair}), and a single call to $U_0$, a $(\|u_0\|,[0,m_N],\epsilon_0)$ block encoding of $u_0$.
Lemma \ref{lem:lchs_be} states that their error parameters must satisfy
\begin{align}
\|u_0\|(\epsilon_{c}+\|c\|_1\epsilon_Q)+\|c\|_1\epsilon_0\leq \epsilon_{LCHS}. 
\end{align}

Finally, we compile $U_Q$ into its subroutines according to Lemma \ref{lem:lcu_select}.
$U_Q$ is a $(1,m_A+4,\epsilon_Q)$-block encoding of the LCU SELECT operator 
$\mathcal{Q}(t) = \sum_{j=1}^M \ketbra{j}{j} \otimes e^{-it(k_j L + H)}=e^{-it\mathcal{S}}$.
It uses 
\begin{align}
e\sqrt{1+K^2}\alpha_A t + 2\log\left (\frac{2\eta}{\epsilon_{\text{exp}}}\right )
\end{align}
queries to controlled $U_A$ and $U_A^\dagger$ plus 4 additional controlled queries (trivial contribution), for the $(\alpha_A, m_A, \epsilon_A)$-block encoding $U_A$,
it uses $M$ times that many calls to $\epsilon_R$-accurate multi-controlled $R_Z(\phi)$ gates,
and it requires that the error parameters $\epsilon_{\mathrm{exp}}$, $\epsilon_A$, and $\epsilon_R$ satisfy
\begin{align}
\epsilon_{\mathrm{exp}}+\sqrt{1+K^2}t\epsilon_A+ 2Mt\sqrt{1+K^2}\alpha_A\epsilon_R\leq \epsilon_Q.
\end{align}
Putting all of these results together, 
$U$ is a 
$(\|\mathcal{L}u_0\|, [m_{c}+m_A+5,m_N+m_{c}+m_A+5], \epsilon)$-block encoding of $u(t)$,
using the following total number of calls to controlled versions of $U_A$,
\begin{align}
C_A=C_{_{LCHS}}(\Delta,\epsilon_{\mathrm{AA}})\left(e\sqrt{1+K^2}\alpha_A t + 2\log\left (\frac{2 \eta}{\epsilon_{\mathrm{exp}}}\right )\right),
\end{align} 
and ensuring that the error parameters satisfy
\begin{align}
    \|u_0\|(\epsilon_{\mathrm{trunc}}&+\epsilon_{\mathrm{disc}})+\left(\|u(t)\|+\|u_0\|(\epsilon_{\mathrm{trunc}}+\epsilon_{\mathrm{disc}}) \right)\left(\epsilon_{\mathrm{AA}}+\frac{9C_{_{LCHS}}(\Delta,\epsilon_{\mathrm{AA}})}{2\|c\|_1\|u_0\|}\left( \|u_0\|\left(\epsilon_{c}\right. \right. \right. \nonumber\\ & \left. \left. \left. +\|c\|_1(\epsilon_{\mathrm{exp}}+\sqrt{1+K^2}t\epsilon_A+ 2Mt\sqrt{1+K^2}\alpha_A\epsilon_R)\right)+\|c\|_1\epsilon_0\right)\right)\leq \epsilon.   
\end{align}
\end{proof}

The above result is stated in terms of several error parameters that need only satisfy a single inequality. This allows the user to budget this error as they see fit. 
However, it may be useful in some cases to simply use a pre-budgeted result. 
That is, the decisions about how much error to allow each subroutine are pre-set such that the total query costs are given in terms of a single error parameter, the overall error $\epsilon$.
The following corollary provides this.
\begin{corollary}[Error-Budgeted Version] \label{coro:error_budgeted}
We can construct a $(\|\mathcal{L}u_0\|, [m_{c}+m_A+5,m_N+m_{c}+m_A+5], \epsilon)$-block encoding of $u(t)$ using the following block encoding subroutines with the associated query counts: 
\begin{itemize}
\item $(O_{c,l}, O_{c,r})$, an $m_{c}$-qubit $\left(\frac{\|c\|_1\epsilon}{36\|v(t)\|C_{_{LCHS}}}\right)$-accurate state preparation pair for the coefficient vector $c$ (see Definition \ref{def:state_prep_pair}) and $U_0$, a $(\|u_0\|,[0,m_N],\|u_0\|\epsilon/36\|u(t)\|C_{_{LCHS}})$ block encoding of $u_0$, both with query count
\begin{align}
C_{_{LCHS}}(\Delta,\epsilon/8\|v(t)\|)=\ceil{\sqrt{8\ceil{\frac{4}{\Delta^2}\ln(512\|u(t)\|/\pi \epsilon^2) e^2 }\ln\left (\frac{512\|v(t)\|\frac{\sqrt{2}}{\Delta}\ln^{1/2}(512\|v(t)\|/\pi \epsilon^2)}{3\sqrt{\pi} \epsilon}\right )}+1},
\end{align}
where $\Delta=2\frac{\|v(t)\| -2\|c\|_1\|u_0\|\epsilon/\|v(t)\|}{\|c\|_1\|u_0\|}$.
\item $c$-$U_A$, an $\left(\alpha_A, m_A, \frac{\epsilon}{36\|v(t)\|C_{_{LCHS}}\sqrt{1+K^2}t}\right)$-block encoding of $A$ with query count
\begin{align}
C_A = C_{_{LCHS}}(\Delta,\epsilon/8\|v(t)\|)\left(e\sqrt{1+K^2}\alpha_A t + 2\log\left(\frac{72\eta \|v(t)\|C_{\mathrm{LCHS}}}{\epsilon}\right )\right)
\end{align}
\item $\left(\frac{\epsilon}{72\|v(t)\|C_{_{LCHS}}Mt \sqrt{1+K^2}\alpha_At}\right)$-accurate multi-controlled $R_Z(\phi)$ gates, with query count
\begin{align}
C_R = MC_A
\end{align} 
as long as $\epsilon\leq 3C_{_{LCHS}}\|v(t)\|/8$, which is used to ensure the performance of the amplitude amplification step. 
\end{itemize}
\end{corollary}
\begin{proof}
Starting from the total error budget 
\begin{align}
    \|u_0\|(\epsilon_{\mathrm{trunc}}&+\epsilon_{\mathrm{disc}})+\|v(t)\|\left(\epsilon_{\mathrm{AA}}+\frac{9C_{_{LCHS}}(\Delta,\epsilon_{\mathrm{AA}})}{2\|c\|_1\|u_0\|}\left( \|u_0\|\left(\epsilon_{c}\right. \right. \right. \nonumber\\ & \left. \left. \left. +\|c\|_1(\epsilon_{\mathrm{exp}}+\sqrt{1+K^2}t\epsilon_A+ 2Mt\sqrt{1+K^2}\alpha_A\epsilon_R)\right)+\|c\|_1\epsilon_0\right)\right)\leq \epsilon  
\end{align}
we can choose to evenly distribute $\epsilon$ among the following eight contributions:
\begin{align}
    \|u_0\|\epsilon_{\mathrm{trunc}} &= \epsilon/8 & 
    \|u_0\|\epsilon_{\mathrm{disc}} &= \epsilon/8 \\
    \|v(t)\|\epsilon_{\mathrm{AA}} &= \epsilon/8 & 
    9\|v(t)\|C_{_{LCHS}}\epsilon_{c}/2\|c\|_1 &= \epsilon/8 \\
    9\|v(t)\|C_{_{LCHS}}\epsilon_{0}/2\|u_0\| &= \epsilon/8 &    
    9\|v(t)\|C_{_{LCHS}}\epsilon_{\mathrm{exp}}/2 &= \epsilon/8 \\
    9\|v(t)\|C_{_{LCHS}}\sqrt{1+K^2}t\epsilon_A/2 &= \epsilon/8 &  9\|v(t)\|C_{_{LCHS}}Mt\sqrt{1+K^2}\alpha_A\epsilon_R &= \epsilon/8
\end{align}
Note that we have used $\|v(t)\|$ above in place of the $\|u(t)\|+\epsilon_v$ used in Theorem \ref{thm:amplified_prep}.
This is a valid replacement as can be seen in Eq. \ref{eq:looser_error_bound}, where $\|\mathcal{L}u_0\|=\|v(t)\|$, and it also simplifies the budgeting of error to $\epsilon_{\textup{trunc}}$ and $\epsilon_{\textup{disc}}$.
This sets the error requirement for each of the operations as given in the theorem statement.
With this error budgeting and using the result in Theorem \ref{thm:amplified_prep}, the total number of calls to $U_A$ as a function of $\epsilon$ is
\begin{align}
C_A = C_{_{LCHS}}(\Delta,\epsilon_{\mathrm{AA}})\left(e\sqrt{1+K^2}\alpha_A t + 2\log(\frac{72\eta \|v(t)\|C_{_{LCHS}}(\Delta,\epsilon_{\mathrm{AA}})}{\epsilon})\right),
\end{align}

where $C_{_{LCHS}}(\Delta,\epsilon_{\mathrm{AA}})$, given in Equation \ref{eq:aa_cost},
is the number of calls to the state preparation pair $(O_{c,l}, O_{c,r})$ and the initial state preparation $U_0$
and, using Lemma \ref{lem:class_error},
\begin{align}
    k_{\text{max}} = \left (\frac{2\beta}{\cos (\beta \pi /2)} W_0\left ( \left (\frac{8B_\beta\|u_0\|}{\epsilon} \right )^{1/\beta} \frac{\cos(\beta \pi /2)}{2\beta} \right ) \right )^{1/\beta}.
\end{align}
The total number of multi-controlled rotations (due to $U_S$) is
\begin{align}
MC_A,
\end{align}
where, using Lemma \ref{lem:class_error},
\begin{align}
    M &= t \|L\| e k_{\text{max}} \ceil{-\frac{\log(e)}{4} W_{-1}\left (\frac{-3 C_\beta}{2\pi e^{1/3} \log(e)} \frac{\epsilon}{\left (\frac{2\beta\|u_0\|m}{\cos (\beta \pi /2)} W_0\left ( \left (\frac{8B_\beta\|u_0\|}{\epsilon} \right )^{1/\beta} \frac{\cos(\beta \pi /2)}{2\beta} \right ) \right )^{1/\beta}}\right )}.
\end{align}
Since each call to the LCHS block encoding uses two calls to the state preparation oracles (one for $O_{c,l}$ and one for $O_{c,r}$), the total number of calls to these oracles is that given in the theorem statement.
\end{proof}

The above results are developed in a modular way that allows for individual lemmas to be updated or used independently.
Yet, for pedagogical purposes, it may be helpful to compile these results into a self-contained presentation, allowing the reader to more-easily follow the overall flow of the analysis.
Below we provide such a self-contained analysis of how the total error can be accounted for from the errors in the constituent parts of the algorithm.

\begin{theorem}[Self-contained error analysis] \label{thm:error}
    Given the solution to Equation \ref{eq:ode} $u(t) = e^{-At}u_0$, such that $\|A\| \leq \alpha $, let $\widetilde{U}_v$ be a $(\gamma_c, [m_{c} + m_Q, m_{c} + m_Q], \epsilon_{LCHS})$ block encoding of the LCHS propagator $\cal{L}$, and ${U}_0$ be the initial state preparation unitary such that $v(t) := \mathcal{L}U_0 \ket{0}$, and let $\|c\|_1 \leq \gamma_C$. If the imperfect output vector from the quantum computer is written $\tilde{v}(t)$, then the total error has the bound: 
    \begin{align}
        \|u(t) - \tilde{v}(t)\| \leq  \epsilon_v + (\|u(t)\| + \epsilon_v) \left ( \epsilon_{\mathrm{AA}} + \frac{{9\epsilon_{\mathrm{LCHS}}}}{2\|c\|_1 \|u_0\| } C_{\mathrm{LCHS}}(\Delta, \epsilon_{\mathrm{AA}})\right ),
    \end{align}
    where 
    \begin{equation}
         \|u_0\| \biggr( \|c\|_1 \Bigr ( t\alpha_A\sqrt{(1+K^2)}\epsilon_A + 2tM\sqrt{(1+K^2)}\alpha_A \epsilon_R + \epsilon_{\mathrm{exp}} \Bigr) +\epsilon_{c} \biggr) + \gamma_C \epsilon_0  \leq \epsilon_{\mathrm{LCHS}} ,
    \end{equation}
    and 
    \begin{equation}
        \epsilon_v \leq  \|u_0\|({\epsilon_{\mathrm{trunc}}} + {\epsilon_{\mathrm{disc}}} ),
    \end{equation}
and $C_{\mathrm{LCHS}}(\Delta, \epsilon_{\mathrm{AA}})$, defined in Equation \ref{eq:aa_cost}, represents the number of calls to the LCHS block encoding made by amplitude amplification.
\end{theorem}
\begin{proof}
We begin by considering a quantum circuit $\widetilde{U}_{AA}$ which is a block encoding of the amplified approximate solution vector $\tilde{v}(t)$ described in Lemma \ref{lem:robust_fp_obl_amp}. Noting that $e^{-At}u_0 = \int_{\mathbb{R}}g(k)U(t,k)\dd k u_0 = u(t)$, and $v(t)$ is obtained via applying the block encoding of the approximate LCHS propagator and approximate state prep unitaries $\widetilde{U}_v$ and $\widetilde{U}_0$, we bound the following:
{\allowdisplaybreaks
\begin{align}
    \|u(t) - \tilde{v}(t)\| \equiv& \left \|\int_{\mathbb{R}}g(k)U(t,k)\dd k u_0  - \|v(t)\| \bra{0}^{m_Q+m_{c}} \widetilde{U}_{AA} \ket{0}^{m_Q+m_{c}} \ket{0}^{\ceil{\log(d)}}\right \| \\
    & \leq \left \|\int_{\mathbb{R}}g(k)U(t,k)\dd k u_0 - v(t) \right \| + \left \|v(t) -  \|v(t)\| \bra{0}^{m_Q+m_{c}} \widetilde{U}_{AA} \ket{0}^{m_Q+m_{c}} \ket{0}^{\ceil{\log(d)}} \right \|
\end{align}}
Now, the second term is exactly in the form of Lemma \ref{lem:robust_fp_obl_amp}, so we can use this bound directly to obtain 

\begin{align}
    \|u(t) - \tilde{v}(t)\| &\leq \underbrace{\left \|\int_{\mathbb{R}}g(k)U(t,k)\dd k u_0 - v(t) \right \|}_{:= \epsilon_v} +  \|v(t)\| \left ( \epsilon_{\mathrm{AA}} + \frac{9\epsilon_{\mathrm{LCHS}}}{2\alpha_{v(t)}} C_{\mathrm{LCHS}}(\Delta, \epsilon_{\mathrm{AA}})\right ) \\
    &= \epsilon_v + (\|u(t)\| + \epsilon_v) \left ( \epsilon_{\mathrm{AA}} + \frac{{9\epsilon_{\mathrm{LCHS}}}}{2\|c\|_1 \|u_0\| } C_{\mathrm{LCHS}}(\Delta, \epsilon_{\mathrm{AA}})\right ),
\end{align}

where $C_{\mathrm{LCHS}}(\Delta, \epsilon_{\mathrm{AA}})$ is the number of calls made by amplitude amplification (see Lemma \ref{lem:robust_fp_obl_amp}) to the LCHS block encoding (see Lemma \ref{lem:lchs_be}), in the second line, we upper bound $\|v(t)\| \leq \epsilon_v + \|u(t)\|$, and we replace $\alpha_{v(t)} = \|c\|_1 \|u_0\|$. Here $\epsilon_{\mathrm{LCHS}}$ contains all the errors in imperfectly implementing the algorithm (ie. due to Hamiltonian simulation, imperfect state preparation pairs, and initial state prep). It thus remains to bound these two errors. Let us begin with $\epsilon_v$, noting that $v(t) := \sum_j c_j U(t,k_j) u_0$:
\begin{align}
    \epsilon_v \equiv & \left | \left |\int_\mathbb{R} g(k) U(k) \dd k \; u_0  - \sum_j c_j U_j u_0   \right | \right | \\
    \leq & \underbrace{\left | \left | \int_{-K}^{K} g(k) U(k) \dd k - \sum_j c_j U_j \right | \right |}_{\epsilon_{\mathrm{disc}}} \|u_0\|  \\
    & + \underbrace{\left | \left | \int_\mathbb{R} g(k) U(k) \dd k - \int_{-K}^K g(k) U(k) \dd k \right | \right |}_{\epsilon_{\mathrm{trunc}}} \|u_0\| \\ 
    \leq & \|u_0\|({\epsilon_{\mathrm{trunc}}} + {\epsilon_{\mathrm{disc}}} ).
\end{align}
Now it remains to bound $\epsilon_{LCHS}$. If we recall the definition of the block encoding $\widetilde{U}_v$ from before, and the imperfect state preparation  unitary $\widetilde{U}_0$, then $\epsilon_{LCHS}$ is understood as the upper bound on the block encoding error $\|v(t) - \tilde{v}(t)\| \leq \epsilon_{\mathrm{LCHS}}$. We proceed by bounding this block encoding error 
\begin{align}
     \|v(t) - \tilde{v}(t)\| =& \left \|\sum_j c_j U_j u_0 - (\bra{0}^{m_{c} + m_Q} \otimes \mathbb{I}) \widetilde{U}_v (\ket{0}^{m_{c} + m_Q} \otimes \mathbb{I}) \widetilde{U}_0 \ket{0}^{m_0}\right \|
    \\
    \leq & \Biggr \| \|c\|_1 \bra{0}^{m_{c}}(O_{c,l}^\dagger \otimes \mathbb{I}_Q ) \mathcal{Q} (O_{c,r} \otimes \mathbb{I}_Q )\ket{0}^{m_{c}} -  \|c\|_1 \bra{0}^{m_{c}}(O_{c,l}^\dagger \otimes \mathbb{I}_Q ) \widetilde{\mathcal{Q}}(\widetilde{A}, t) (O_{c,r} \otimes \mathbb{I}_Q )\ket{0}^{m_{c}} \Biggr \| \|u_0 \| \\
    & + \Biggr \| \|c\|_1 \bra{0}^{m_{c}}(O_{c,l}^\dagger \otimes \mathbb{I}_Q ) \widetilde{\mathcal{Q}}(\widetilde{A}, t) (O_{c,r} \otimes \mathbb{I}_Q )\ket{0}^{m_{c}} u_0 - (\bra{0}^{m_{c} + m_Q} \otimes \mathbb{I}) \widetilde{U}_v (\ket{0}^{m_{c} + m_Q} \otimes \mathbb{I}) \widetilde{U}_0 \ket{0}^{m_0}\Biggr \| \\
    \leq & \|c\|_1 \underbrace{\left \| \mathcal{Q}(A,t) - \widetilde{Q}(\widetilde{A}, t) \right \|}_{\epsilon_{\mathrm{Q}}} \|u_0\| \\
    & + \Biggr \| \|c\|_1 \bra{0}^{m_{c}}(O_{c,l}^\dagger \otimes \mathbb{I}_Q ) \widetilde{\mathcal{Q}}(\widetilde{A}, t) (O_{c,r} \otimes \mathbb{I}_Q )\ket{0}^{m_{c}} u_0 - (\bra{0}^{m_{c} + m_Q} \otimes \mathbb{I}) \widetilde{U}_v (\ket{0}^{m_{c} + m_Q} \otimes \mathbb{I}) \widetilde{U}_0 \ket{0}^{m_0}\Biggr \| ,
\end{align}
\noindent where in the last line we utilize the fact that we can use sub-multiplicativity to pull out $\bra{0}^{m_{c}}(O_{c,l}^\dagger \otimes \mathbb{I}_Q)$ on the left, which has norm one, and perform the same on the right with the corresponding oracle. Now, we bound ${\epsilon_Q}$:
\begin{align}
    \| \mathcal{Q}(t, A) - \widetilde{\mathcal{Q}}(t, \widetilde{A}) \| & \leq \|\mathcal{Q}(t, A) - {\mathcal{Q}}(t, \widetilde{A})\| + \|\mathcal{Q}(t, \widetilde{A}) - \widetilde{\mathcal{Q}}(t, \widetilde{A}) \| \\
    & = \|e^{-it\cal{S}} - e^{-it\widetilde{\cal{S}}} \| + \epsilon_{\mathrm{exp}} \\
    & \leq t \|\cal{S} - \widetilde{\cal{S}}\| +\epsilon_{\mathrm{exp}} \\
    &= t\alpha_A\sqrt{2(1+k^2_{\max})}\epsilon_A + tMk_{\max}\alpha_A \epsilon_R + \epsilon_{\mathrm{exp}} := \epsilon_Q,
\end{align} 
where we use $\widetilde{\cal{Q}}$ to mean the approximate unitary produced by qubitization, and $\widetilde{\cal{S}}$ to mean the imperfect block encoding of $\cal{S}$. The error $\epsilon_{\mathrm{exp}}$ is therefore the error originating from the imperfect unitary produced by qubitization, and therefore dictates the query cost of qubitization in Lemma \ref{lem:lcu_select}.  In the third line we use the inequality $\| e^{-iHt} - e^{-i\widetilde{H}t}\| \leq t\|H - \widetilde{H}\|$ from Ref. \cite{chakraborty2018power}, and then recognize that this is the block encoding error of $\cal{S}$ in the fourth line. Plugging this back in, we resume to bounding of the total algorithmic error:

{\allowdisplaybreaks
\begin{align}
    \|v(t) - \tilde{v}(t)\| \leq & \|u_0\|\|c\|_1 \epsilon_Q \\
    &+ \Biggr \| \|c\|_1 \bra{0}^{m_{c}}(O_{c,l}^\dagger \otimes \mathbb{I}_Q ) \widetilde{\mathcal{Q}}(\widetilde{A}, t) (O_{c,r} \otimes \mathbb{I}_Q )\ket{0}^{m_{c}} u_0 \\
    &- (\bra{0}^{m_{c} + m_Q} \otimes \mathbb{I}) \widetilde{U}_v (\ket{0}^{m_{c} + m_Q} \otimes \mathbb{I}) \widetilde{U}_0 \ket{0}^{m_0}\Biggr \| \\
    =& \|u_0\| \|c\|_1 \epsilon_Q
    + \Biggr \| \sum_j c_j \widetilde{U}(t, k_j) u_0 - \gamma_c \left (\sum_j \frac{\tilde{c}_j}{\|\tilde{c}\|_1} \widetilde{U}(t, k_j)  \right ) \widetilde{U}_0 \ket{0}^{m_0}\Biggr \| \\
    \leq & \|u_0\| \|c\|_1 \epsilon_Q + \Biggr \| \sum_j c_j \widetilde{U}(t, k_j) u_0 - \gamma_c \left (\sum_j \frac{\tilde{c}_j}{\|\tilde{c}\|_1} \widetilde{U}(t, k_j)  \right ) u_0 \Biggr \| \\
    &+ \Biggr \| \gamma_c \left (\sum_j \frac{\tilde{c}_j}{\|\tilde{c}\|_1} \widetilde{U}(t, k_j)  \right ) u_0  - \gamma_c \left (\sum_j \frac{\tilde{c}_j}{\|\tilde{c}\|_1} \widetilde{U}(t, k_j)  \right ) \widetilde{U}_0 \ket{0}^{m_0}\Biggr \| \\
    \leq & \|u_0\| \|c\|_1 {\epsilon_Q} + \underbrace{\sum_j \Biggr |  \left (c_j  - \gamma_c \frac{\tilde{c}_j}{\|\tilde{c}\|_1}   \right )\Biggr |}_{\epsilon_{c}} \|u_0\| + \gamma_c \underbrace{\|u_0 - \widetilde{U}_0\ket{0}^{m_0} \|}_{\epsilon_0} \\
    \leq & \|u_0\| \biggr( \|c\|_1 \Bigr ( t\alpha_A\sqrt{(1+K^2)}\epsilon_A + 2tM\sqrt{(1+K^2)}\alpha_A \epsilon_R + \epsilon_{\mathrm{exp}} \Bigr) +\epsilon_{c} \biggr) + \gamma_C \epsilon_0 \leq \epsilon_{\mathrm{LCHS}}
\end{align}}
Plugging in the values of the $\epsilon_v$ and $\epsilon_{\mathrm{LCHS}}$ completes the proof.
\end{proof}
Above the error in the imperfect state preparation pair $(\widetilde{O}_{c,l}, \widetilde{O}_{c,r})$ is bounded using the worst case error in the coefficients from either oracle to avoid cumbersome notation. For example, we considered individual errors from each of the pairs we would require including $(\tilde{c}_{j,l}, \tilde{c}_{j,r})$, and the block encoding parameter in the bound is replaced with $\sqrt{\gamma_{c,l}\gamma_{c,r}}$. However, this notation is unnecessarily cumbersome for our purposes, and in practice, each coefficient would likely be subject to the same decimal rounding error anyway. \\

Notice that the robust fixed point amplitude amplification algorithm given in Lemma  \ref{lem:robust_fp_obl_amp} gives an implicit relationship for the error. Intuitively this is because the larger the $\epsilon_{\mathrm{LCHS}}$ of the incoming state, the more total error that will accumulate as one applies Grover iterates to meet the target amplitude amplification precision $\epsilon_{\mathrm{AA}}$. Unfortunately, this also complicates the study of the total error in the end to end implementation. In Section \ref{sec:numerics}, we provide strategies to study the complexity of the simulation in presence of non-linear constraints.

\section{Constant Factor Query Counts} \label{sec:numerics}
In this Section we provide a numerical analysis of the constant factor bounds presented in this paper. Throughout the section we take $\epsilon_A = \epsilon_0 = \epsilon_R = 0$, meaning that we assume access to a perfect block encoding, initial state prep oracle, and perfect rotation gates. The reason for these choices is that in practice, a choice of these errors are motivated by some lower-level details of the implementation that do not have a significant bearing on the higher-level algorithmic comparison that we are aiming for.
With these choices, the result of Theorem \ref{thm:error} now states:
\begin{equation}\label{eq:constraint}
    \epsilon_v + (\|u(t)\| + \epsilon_v) \left ( \epsilon_{\mathrm{AA}} + \frac{9}{2} \frac{{\epsilon_{\mathrm{LCHS}}}}{\|c\|_1 \|u_0\|} C_{\mathrm{LCHS}}(\Delta, \epsilon_{\mathrm{AA}})\right ) \leq \epsilon,
\end{equation}
where the LCHS error simplifies to
\begin{equation}
     \epsilon_{LCHS} = \|c\|_1 \|u_0\|  \epsilon_{\mathrm{exp}}, 
\end{equation}
yielding the overall constraint:
\begin{equation}
    \epsilon_v + (\|u(t)\| + \epsilon_v) \left ( \epsilon_{\mathrm{AA}} + \frac{9}{2} {\epsilon_{\mathrm{exp}}} C_{\mathrm{LCHS}}(\Delta, \epsilon_{\mathrm{AA}})\right ) \leq \epsilon.
\end{equation}

In Appendix \ref{app:polynomial_degree_bound}, we provide an explicit expression for $C_{\mathrm{LCHS}}(\Delta, \epsilon_{\mathrm{AA}})$ in Theorem \ref{thm:sgn_approx}. The goal of this Section is therefore to minimize the cost of the algorithm, defined as the number of queries to $U_A$, where the cost follows the form of Theorem \ref{thm:amplified_prep}:
\begin{equation} \label{eq:cost_A}
C_A=C_{_{LCHS}}(\Delta,\epsilon_{\mathrm{AA}})\ceil{e \sqrt{1+K^2}\alpha_A t + 2\ln \left (\frac{2\eta}{\epsilon_{\mathrm{exp}}}\right ) }    
\end{equation}

Finding the minimum of $C_A$ is therefore a constrained optimization problem. Surprisingly, we find in this section that the constraints can be lifted by exactly solving $\ref{eq:constraint}$, again, in terms of Lambert-W functions. Through lifting the constraints, we are then able to perform a simple optimization of the cost function in the subsequent subsection.

\subsection{Lifting Optimization Constraints}
Without further approximation, the derived formulas for the cost and error in Equations \ref{eq:cost_A} and \ref{eq:constraint} can be not only difficult to interpret analytically, but also hard to deal with numerically given that they admit a constrained optimization problem.  Note that it is not straightforward to solve this expression either due to the fact that, in accordance with the construction presented in Ref. \cite{martyn2021grand}, $\Delta$ is a lower bound on the amplitude of the quantum state outputted by the LCHS algorithm, meaning 
\begin{equation} \label{eq:Delta}
    \Delta = \frac{2(\|u(t)\|-(\|c_1\|  \|u_0\| \epsilon_{\mathrm{exp}}) - \epsilon_v))}{\|u_0\|\|c\|_1}
\end{equation}
which is clearly a function of various error parameters. Here, we show that the constraints can be lifted by solving these equations using two different strategies. The first approach is to directly solve this equation for $\epsilon_{\mathrm{AA}}$ again by using a Lambert-W function as in the sections before. To accomplish this, we must manipulate $C_{_{LCHS}}$ such that it takes the functional form of $C_{_{LCHS}}(\Delta,\epsilon_{\mathrm{AA}}) = \mathrm{poly}(\Delta^{-1})\log(\epsilon^{-1}_{\mathrm{AA}})$, which is achieved via an upper bound in Proposition \ref{prop:degree_upper} 
\begin{equation}
    C_{\mathrm{LCHS}} \leq \frac{8e}{\Delta}\sqrt{1+\frac{1}{e}}\ln\left (\frac{64\sqrt{2}}{3 \sqrt{\pi} \Delta \epsilon} \right )
\end{equation}
Saturating this equality then yields a solvable expression for $\epsilon_{\mathrm{AA}}$ which we state below:
\begin{equation}
    \epsilon_{\mathrm{AA}} = -\frac{16e\sqrt{1+1/e}}{\Delta}\epsilon_{\mathrm{exp}} W_{-1}\left [\frac{-16\sqrt{2}}{ 27e\sqrt{\pi + \pi/e}} \frac{1}{\epsilon_{\mathrm{exp}}}\exp\left(-\frac{(\epsilon - \epsilon_{v})}{(\|u(t)\| + \epsilon_{v})\epsilon_{\mathrm{exp}}} \frac{\Delta}{36e\sqrt{1+1/e}} \right ) \right ].
\end{equation}
This expression can be obtained in the same fashion as was shown in previous proofs with Lambert-W functions. Now, if we write the above expression like $\epsilon_{\mathrm{AA}} = f(\epsilon, \epsilon_{LCHS})$, then we can rewrite the cost function in terms of the above function such that $C_A(\epsilon_v, \epsilon_{LCHS}, \epsilon_{\mathrm{AA}}) = C_A(\epsilon_v,  \epsilon_{LCHS}, f(\epsilon, \epsilon_{LCHS}))$. By then fixing a choice of total error tolerance $\epsilon$, we have effectively lifted the constraints and can optimize the cost function with a wide variety of optimization algorithms. We label this approach $\mathtt{Sol(\epsilon_{\mathrm{AA}})}$, since we have found a solution to Equation \ref{eq:constraint} for $\epsilon_{\mathrm{AA}}$ to lift the constraints. Note that the optimizer still must be configured to ensure that the required properties of $W_{-1}(x)$ are satisfied, i.e. $x\in(-1/e, 0)$ to ensure $W_{-1}(x) \in \mathbb{R}$ and uniquely defined. \\ 

Another approach is to further lower-bound $\Delta$ (thus upper-bounding $C_{\mathrm{LCHS}}$) to remove its dependence on other error parameters we may wish to solve for. Using the fact that $\Delta \geq \frac{2(\|u(t)\|-\epsilon)}{\|u(0)\|\|c\|_1}$, by using the fact that the total error $\epsilon$ is greater than the error terms that subtract the numerator in Equation \ref{eq:Delta}. This now allows us to solve the constraint in Equation \ref{eq:constraint} for 
\begin{equation}
    \epsilon_{\mathrm{exp}} = \left [\frac{\epsilon-\epsilon_v}{\|u(t) \| +\epsilon_v} - \epsilon_{\mathrm{AA}} \right ]\frac{2}{9C_{\mathrm{LCHS}}}, 
\end{equation}
where in this case $C_{\mathrm{LCHS}}$ need not be simplified using Proposition \ref{prop:degree_upper}. This method is simpler, however, it creates a self-referential inequality by introducing the dependence $\Delta(\epsilon)$. For our purposes, we can still obtain a function for $\epsilon_{\mathrm{exp}}$ and configure our optimizer so so as to not explore parameter regimes where this becomes problemat(i.e.i.e. choosing negative $\epsilon_v$). This method is labeled $\mathtt{Sol(\epsilon_{\mathrm{exp}})}$. 

\subsection{Numerical Results} \label{subsec:numeric_results}
Here we perform numerical analysis of the LCHS cost function derived in this paper (Equation \ref{eq:cost_A}), where cost is defined as the queries to $U_A$. We take two different approaches to calculating this cost; one sets the error of each imperfect subroutine to the same value, and slowly tunes this parameter until Equation \ref{eq:constraint} is satisfied, while the other performs an optimization of cost function to find an optimal error balancing. In the latter case, both of the methods $\mathtt{Sol(\epsilon_{\mathrm{AA}})}$ and $\mathtt{Sol(\epsilon_{\mathrm{exp}})}$ are employed to lift the constraints of the optimization problem, and the results are compared. In the optimized cases, the parameter $\beta \in (0,1)$ is also optimized, whereas in the former approach we choose $\beta=0.75$. This choice is made based on the numerical observations that the optimal $\beta$ should be roughly between $[0.7, 0.8]$, which agree with the conclusions made by \cite{an2023quantum}. To perform the optimization, we utilize the Gradient Boosted Regression Trees (GBRT) from $\mathtt{SciKit \: Optimize}$ \cite{pedregosa2011scikit}. The reasons for this choice are that the algorithm can optimize arbitrary cost functions, the user can easily set bounds on parameters for optimization, and the algorithm has seen success in optimizing Hamiltonian simulation algorithms in the past in Ref. \cite{pocrnic2024composite}. Setting bounds on possible parameter values is especially important given the restrictions on $\beta$, as well as the Lambert-W function used in $\mathtt{Sol(\epsilon_{\mathrm{AA}})}$. 

We compare our results to Ref. \cite{jennings2024cost}. We compute their cost formula for the more general case of \textit{unstable dynamics}, as perviously discussed in Section \ref{sec:introduction}. In this case, no assumptions are made on the log-norm of the generator, preventing any fast-forwardability. 
We generated the RLS data using the following \href{https://www.desmos.com/calculator/wwryu4lhnx}{Desmos module}.
This interactive module also allows for exploration of the LCHS costs.
We take measures to make this comparison as fair as possible. For instance, we consider all the same errors present in that analysis up to errors that do not apply to their method (such as the integral truncation error). For example, their analysis is done under the assumptions that the block encoding synthesis and state preparations are done perfectly, as is the approach taken in this section. Without other well-motivated numerical choices, we also choose $\|u_0\| = \|u(t)\| = 1$ for each algorithm. The final choice to make is that of the $C_{\max}$ parameter that appears in Equation 13 of \cite{jennings2024cost}. In Appendix \ref{sec:Cmax} we provide a discussion on how this is done for a fair comparison. The final parameters to fix are $\alpha = \|L\| = 1$, as they have the same scaling as $t$, and we set the total error to $\epsilon=10^{-10}$ as in Ref. \cite{jennings2024cost}. The results of these numerics are shown in Figure \ref{fig:cost_plot}. 
\FloatBarrier
\begin{figure}[htbp!] 
    \centering\includegraphics[width=0.85\textwidth]{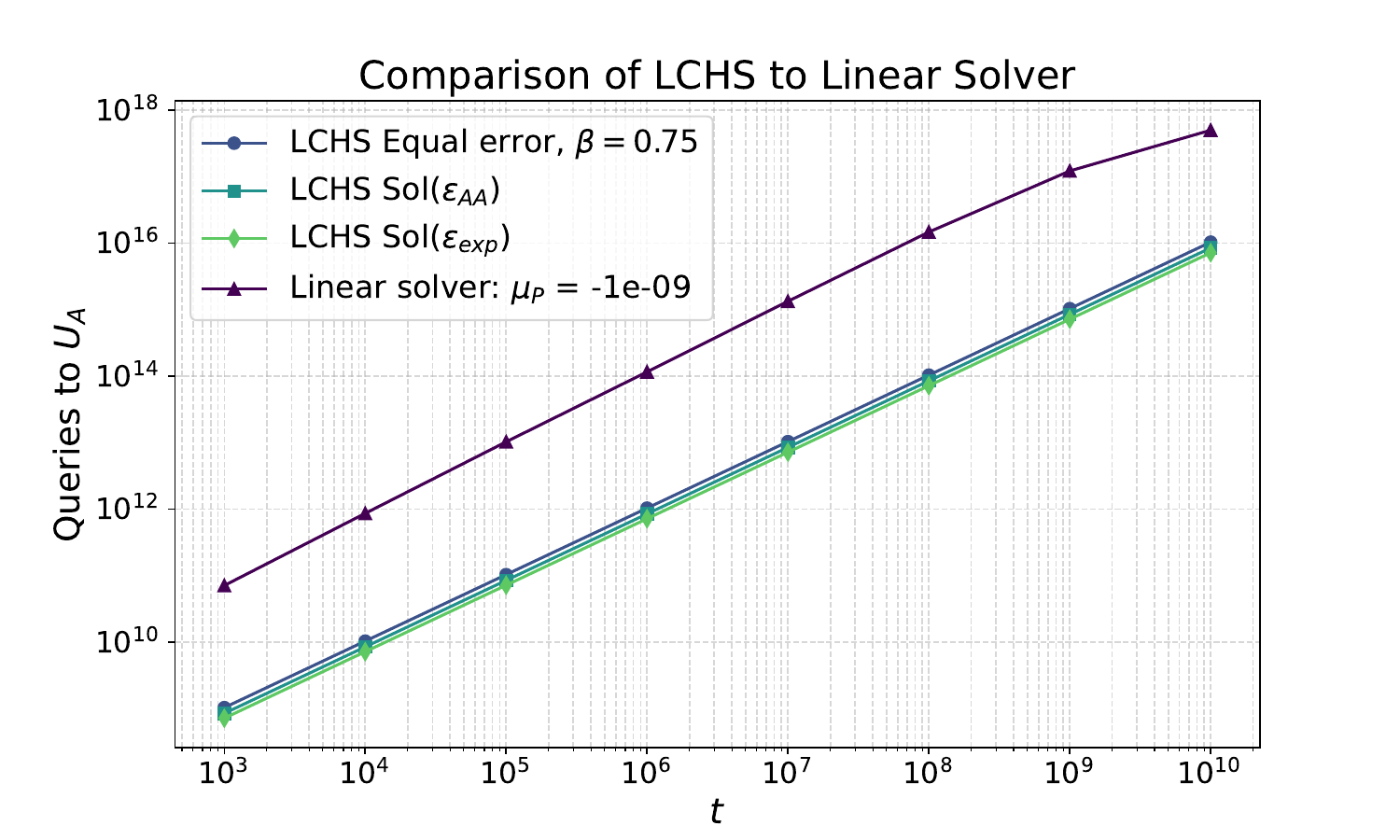}
    \caption{A comparison of the LCHS constant factor bounds and the randomization linear solver (RLS) method from Ref \cite{jennings2024cost}. The numerics are conducted such that the total error satisfies $\epsilon\leq 10^{-10}$, and the block encoding constant is $\alpha=1$, though we expect comparable speedups to persist when varying these parameters. For the equal-error data, the errors for each subroutine are evenly distributed such that the total error is less than $\epsilon\leq 10^{-10}$, and for the other cases the optimizer finds a more efficient distribution. In the optimized cases it is found that $\beta \in [0.7403, 0.8127]$, yielding $\|c\|_1 \in [1.385,1.585]$. We note that, although the randomization linear solver cost data appears to become sub-linear, the curve actually eventually returns to being linear, though with a smaller intercept; the location of this transition is governed by the $\mu_P$ parameter in \cite{jennings2024cost}. This is likely a numerical artifact given that we are limited in our accuracy regarding how closely $\mu_P$ can approach zero. Most of the data on this plot can be explored interactively at the following \href{https://www.desmos.com/calculator/wwryu4lhnx}{link}.} \label{fig:cost_plot}
\end{figure} 
In terms of the comparison, we find that the LCHS algorithm, based on the formulas derived in this paper, requires significantly less queries to $U_A$ compared to the randomization linear solver (RLS). The naive equal error budgeting LCHS implementation achieves an advantage over RLS that is nearly two orders in magnitude, whereas for the optimized version, the improvement is even greater. Due to our high demand on the precision and output success, the cost of amplitude amplification is roughly accounting for 2 order of magnitude in the cost, which can be further decreased if the user is willing to take on some probability of failure. This, however, does not affect the magnitude of the speedup, as the same amplitude amplification routine is being used for the RLS data. Note that in the less general case where the generator $A$ is strictly positive definite, it is possible for RLS to be fast-forwarded, achieving $\sqrt{t}$ scaling. In this case, as seen in Figure 1 of Ref. \cite{jennings2024cost}, the approaches will become more comparable in cost with the advantage of LCHS asymptotically shrinking, and eventually disappearing when the cost of state preparation is not taken into account. However, if such an analogous speedup is achieved for LCHS, then the constant factor bounds for LCHS are expected to be strictly better in all cases, unless the analysis of the RLS can be significantly improved. The latter case should by no means be excluded from the realm of possibility, given that the analysis of the quantum linear system algorithm continues to improve, noting the recent improvement by Dalzell \cite{dalzell2024shortcut}. It should also be noted that while our bounds give significantly smaller $U_A$ query estimates, they do not imply that LCHS will necessarily perform strictly better in practice. However, the bounds provide the means for improving the resource estimates of quantum ODE solutions, and encourage better analyses from competing approaches. 

Proceeding, it is also important to factor the cost of state preparation into the total cost of an ODE solution. We make this comparison by considering the time complexity of the state preparation oracle $T(U_0)$ as a fraction of the time complexity of synthesizing the block encoding $T(U_A)$. Then given the expressions for the number of calls to initial state preparation in each algorithm, we can compare how the total cost scales with this fraction. What we wish to then plot is the following:

\begin{equation}
    \ell = \frac{C^{\mathrm{Linear}}_A + \chi C^{\mathrm{Linear}}_0}{C^{\mathrm{LCHS}}_A + \chi C^{\mathrm{LCHS}}_0}, 
\end{equation}
where $\ell$ is effectively the factor of speedup of LCHS over RLS as a function of $\chi:= T(U_0)/ T(U_A)$, the ratio between the time cost of initial state prep to that of the block encoding synthesis. We plot these results in Figure \ref{fig:speedup_plot}. It is sufficient to choose a general notion of time complexity $T$ to make our point here, whereas in practice, $T$ can serve as a measure of $T$-gate count or surface code cycles \cite{beverland2022assessing}, active volume \cite{litinski2022active} or any circuit metric that increases with the number of calls to the relevant sub-circuits. In addition, we can also study the figure of merit $\ell(t)$ as a function of the simulation time. For means of simple comparison, here we consider the case where the cost of state preparation is negligible $\chi \to 0$. We plot these results in Figure \ref{fig:ratio_plot} and find that the advantage of LCHS over the randomized linear solver can be approximately quantified as $\ell(t) = 21.28\log_{10}(\alpha t) + 35.41$ with fixed $\epsilon = 10^{-10}$. This indicates that in Figure \ref{fig:speedup_plot}, not only is there a constant factor advantage, but the slope of the LCHS cost is smaller by a logarithmic factor (given the plot is log-scale). Therefore, the speedup is expected to slightly grow with time in this fixed-parameter regime. 

\begin{figure}[htbp!] 
    \centering
    \begin{subfigure}[b]{0.49\textwidth}
        \includegraphics[width=1\textwidth]{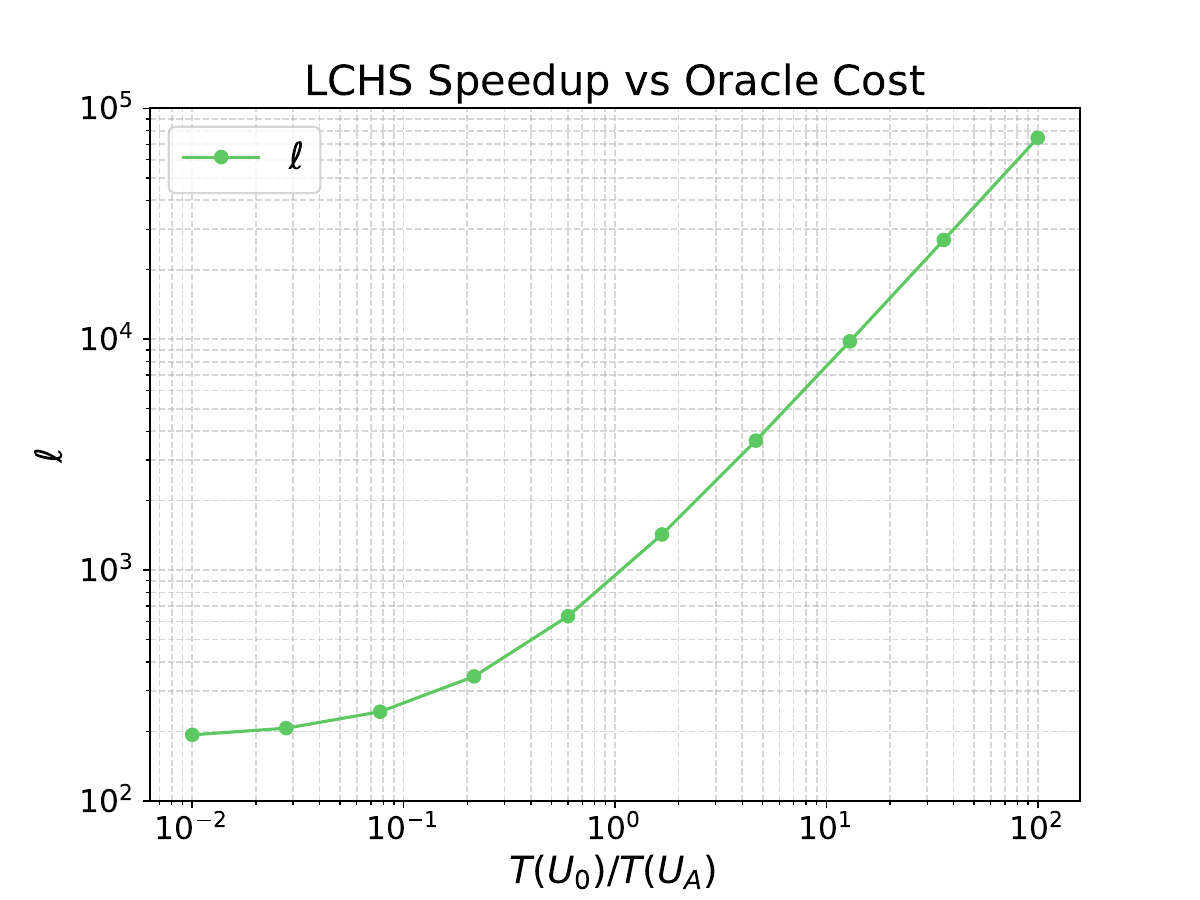}
        \caption{}\label{fig:speedup_plot}
    \end{subfigure}
    \begin{subfigure}[b]{0.49\textwidth}
        \includegraphics[width=1\textwidth]{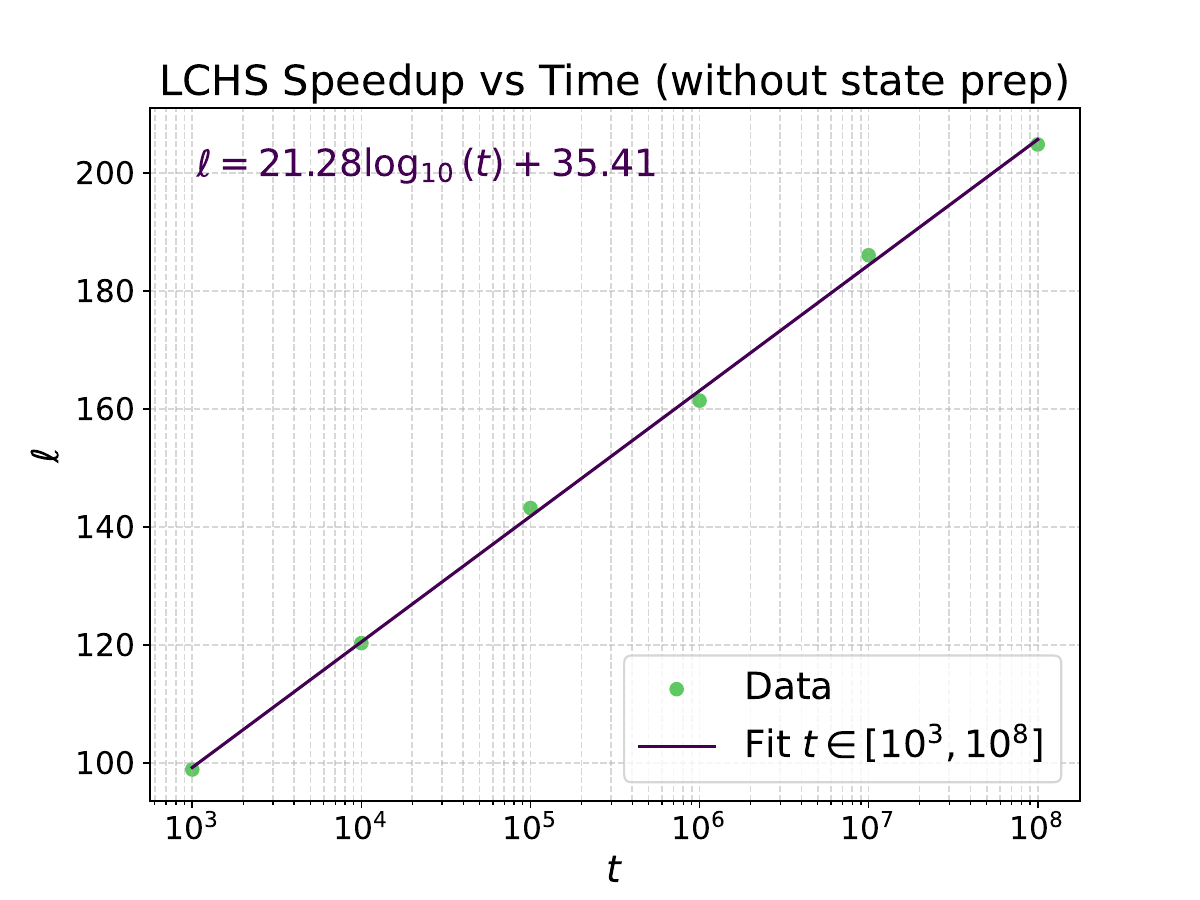}
        \caption{}\label{fig:ratio_plot}
    \end{subfigure}
    \caption{\textit{The Speedup of LCHS over the randomization linear solver (RLS) approach \cite{jennings2024cost} as a function of (a) the ratio of the state preparation oracle $U_0$ time cost to the block encoding $U_A$ time cost, and (b) the simulation time in units of the block encoding constant ($\alpha =1$) given that state preparation is free ($\chi$=0). We perform this fit only over the region where the $\mu_P \to 0^-$ limit is valid, namely where $t\mu_P << 1$. The most performant LCHS simulation data is used to compute $\ell$. As an examples of time cost, one could consider them to be measured in terms of number of $T$ gates \cite{beverland2022assessing} or the active volume \cite{litinski2022active} of the compiled quantum circuits.}} 
\end{figure}

A possible initial reaction to this speedup might be to ask whether LCHS is just a space-time tradeoff, given that we are summing a large number of terms $M$, where the corresponding coefficients are stored in an ancilla register. To counter this claim, for the largest time evaluated numerically ($t=10^{10}$) the number of qubits used by RLS is 51 (calculated using Theorem 12 of Ref. \cite{jennings2024cost}). In comparison, for LCHS the Qubit counts can also be roughly estimated as $\ceil{\log_2M} + m_A$, where $\ceil{\log_2M}$ is the number of qubits needed to represent the ancillary space and $m_A$ is the number required to block encode the generator $A$. For the corresponding point in time ($t=10^{10}$), LCHS approximately requires a 48 qubit ancillary space. Note that this is not a precise count; additional ancilla may be required in the \textit{prepare} operation to handle negative and complex numbers, and in practice, one may wish to introduce additional ancilla to reduce the $T$-depth via QROM \cite{babbush2018encoding}, or utilize dirty qubit state preparation via QROAM \cite{low2024trading, berry2019qubitization}. Depending on the state preparation subroutine, the total qubit count will be altered. This, however, can lead to intricate problem-dependent tradeoffs that require analysis on a case by case basis.

\section{Discussion and Outlook}
\label{sec:discussion_and_outlook}

Our work provides a constant factor analysis of the linear combination of Hamiltonian simulation (LCHS) algorithm \cite{an2023quantum, an2023linear}.
We used this analysis to make a numerical comparison of query costs between LCHS and the randomization linear solver (RLS) algorithm \cite{jennings2024cost}, the previous state-of-the-art.
As shown in Figure \ref{fig:cost_plot}, we found that LCHS approximately yields a 100-200x speedup over RLS, where this advantage seems to increase logarithmically with $\alpha t$ argued by Figure \ref{fig:ratio_plot}.
These findings can enable researchers to make more-optimized estimates of the absolute costs of solving differential equations on a quantum computer.
This is significant because such resource estimates can be used to more-accurately assess the viability of quantum algorithms for differential equation solving.

We discuss the implications of these findings, starting with the practicality of solving differential equations on a quantum computer.
As an example of resource estimation, we can compare the costs for LCHS to the costs of modern quantum computing approaches to solving problems in quantum chemistry.
Taking a state-of-the-art method \cite{low2025fast} for estimating the ground state energy of a molecule, the number of calls to the primary block encoding per circuit is $10^4$ to $10^5$ (i.e. the ratio of Total Toffoli to $C_{\text{B.E.}}$ in Table V \cite{low2025fast}).
Here, each block encoding itself has a Toffoli gate cost in the range of $10^3$ to $10^5$ (see $C_{\text{B.E.}}$ in Table V \cite{low2025fast}).
In contrast, Figure \ref{fig:cost_plot} shows that the number of calls to the block encoding per circuit ranges from $10^{9}$ to $10^{16}$. 
As for the cost per block encoding in differential equations applications, few works have detailed such accounting, but a recent study \cite{penuel2024feasibility} on estimating drag in a computational fluid dynamics setting reported $T$ gate counts\footnote{The cost of implementing Toffoli gates and $T$ gates with a fault-tolerant quantum computer are within an order of magnitude of one another using modern methods \cite{gidney2019efficient}.} per block encoding of order $10^6$. Moreover, in order to use LCHS or any other quantum differential equation solver in a practical setting, one must extract information from the output state, which requires repeatedly calling the differential equation solver circuit (see \cite{penuel2024feasibility} as an example). 
This shows that, despite the improvements of LCHS over RLS, the costs of quantum algorithms for differential equations is still relatively high.

There are several important caveats to such a comparison.
First, the comparison of number of calls to the block encoding is not quite fair because the number of block encodings for LCHS is proportional to the total simulation time $t$, which might vary significantly from application to application.
More research is needed to determine what total time $t$ is required by classically-intractable instances.
Second, while such algorithms in quantum chemistry have had the benefit of nearly a decade of algorithmic and compilation improvements (see, e.g. Table I in \cite{low2025fast}), quantum algorithms for differential equations, and in particular their circuit compilation, have had relatively less attention.
As has been the case for quantum algorithms in quantum chemistry (Table I in \cite{low2025fast} shows a $10^5$ reduction in gate count over 8 years), we anticipate that researchers will devise methods to exploit structure in the data encoding and data processing of specific differential equations to improve costs significantly.

We discuss the theoretical advances that supported our numerical findings. Our two main improvements to LCHS were to tighten the truncation and discretization bounds of the kernel integral in Lemmas \ref{lem:K_formula} and \ref{lem:gauss} and to improve the data encoding mechanism by utilizing a diagonal block encoding that exploits structure in the LCHS SELECT operator \ref{lem:SH}. 
Furthermore, we applied the recently developed generalized quantum signal processing method \cite{berry2024doubling} to give an efficient implementation of robust Hamiltonian simulation in Lemma \ref{lem:lcu_select}. 
All of these contributions help to provide a more
accurate and optimized accounting of the 
costs of LCHS compared to the original approach.
It is difficult to give a quantitative comparison to the original LCHS approach because the goal of that work was not to provide a detailed compilation with constant factor costs. 
However, we expect that our detailed accounting lays the ground work for
further improvements 
that can be tracked quantitatively.

Our work included another important theoretical advance that is both useful for our analysis and likely of general interest in quantum algorithms.
Lemma \ref{lem:robust_fp_obl_amp} gives a constant factor analysis of robust fixed point oblivious amplitude amplification.
The ``unamplified'' versions of LCHS (Lemma \ref{lem:lchs_be}) and RLS have different subnormalizations and 
thus prepare the solution state with different post-selection probabilities. So, to make a fair comparison between the methods, we considered the respective costs of preparing the solution states with near-unit probability, which is achieved using amplitude amplification. Specifically, we accounted for the additional cost of amplifying the amplitude of the approximate solution vector $\tilde{u}(t)$ by deriving a constant factor bound on fixed-point oblivious amplitude amplification in Appendix \ref{app:robust_AA}.
This cost is based on a constant factor bound on the degree of a polynomial approximation of the $\textup{sign}$ function, which we establish in Theorem \ref{thm:sgn_approx}. 
Because amplitude amplification is a ubiquitous quantum algorithm, we anticipate that this result will be broadly of interest to researchers making resource estimates and designing quantum algorithms.

Next, we discuss a feature of the LCHS algorithm that can enable further improvements over RLS.
The LCHS algorithm only requires a single call to the initial state preparation $U_0$ per simulation, while RLS must call $U_0$ multiple times per simulation.
In Section \ref{sec:numerics} we explored the total cost dependence on the relative cost of the initial state preparation $U_0$ and the block encoding $U_A$. For LCHS, in cases where $U_0$ has a negligible cost compared to the cost of $U_A$, then this feature of LCHS does not significantly impact the overall costs. However, as the cost of $U_0$ is increased relative to that of $U_A$, the overall speedup of LCHS relative to RLS grows. Figure \ref{fig:speedup_plot} shows that, roughly, for each order of magnitude increase in the relative cost of $U_0$ over $U_A$, there is a proportional increase in the speedup of LCHS over the randomization linear solver method.  

Before describing several future directions, we re-emphasize an often-neglected aspect of quantum differential equation solvers: the requirement to estimate the norm $\|u(t)\|$ of the time-evolved state. 
Such an estimate is needed, for example, to appropriately rescale an amplitude estimate involving the output state.
LCHS provides the means for this norm estimation given the fact that when the LCHS propagator $\cal L$ is successfully implemented, it is flagged via measuring $\ket{0}$ in the $\mathrm{PREP}$ register. As shown in Lemma \ref{lem:lchs_be}, the probability of measuring $\ket{0}$ is roughly equal to $\frac{(\|u(t)\| - \epsilon)^2}{\|u_0\|\|c\|_1}$. Given that $\|c\|_1$ is a known quantity that is calculated classically, then if we have knowledge of $\|u_0\|$ beforehand, then $\|u(t)\|$ is straightforward to estimate using probability estimation techniques with simple postprocessing. While this introduces an additional step, our analysis suggests that norm estimation can be performed efficiently using amplitude estimation or sampling-based techniques. This is important for real ODE applications such as in Ref. \cite{penuel2024feasibility}.  While we have discussed the bias in such estimates below Lemma \ref{lem:lchs_be}, we leave a more detailed treatment of this component to future work.\\

Finally, we list several future directions that would provide more insight on the prospects of using quantum computers to address problems in differential equations:
\begin{itemize}
    \item An advantage of the randomization linear solver method is that it is able to exploit properties (specifically, stability) of the $A$ matrix to improve the scaling with respect to $t$. Thus, even if with the constant factor LCHS improvements, there are regimes where the randomization linear solver would be more favorable. An important future direction is to explore if the LCHS method can be adapted to also leverage the case where $A$ is stable yielding fast forward-able dynamics. 

    \item Another direction to investigate is whether or not the structure in the Gaussian quadrature points, or those present in some other quadrature methods, can be exploited to reduce the cost of implementing $D$ -- the diagonal matrix that stores the points of integration. 

    \item This compilation of LCHS used qubitization with GQSP for an efficient implementation of the time evolution operator. There are known cases where product formulas can outperform qubitization in practice \cite{rubin2024quantum}, which raises the question of whether this may hold for certain types of differential equation.
    
    \item Given that the randomization linear solver provides a history state encoding, whereas LCHS output the final time solution vector $u(t)$, are there instances where the larger costs of the Linear solver are worth it for the purposes of accessing the solution at multiple points in time? An example of this may be for solving Maxwell's equations with a retarded potential.
    
    \item With constant factor formulas now worked out for multiple quantum differential equation solvers, the potential for algorithmic optimization is more amenable to those developing specific applications. Future work should be devoted to optimizing existing quantum ODE solvers, costing out existing asymptotic results, and proposing problem instances where these are useful.   
\end{itemize}

This work has provided a series of tools that can be used to systematically track errors and cost-out real ODE problem instances. Given estimates of gate counts for the block encoding and state preparation, our results can be used to estimate gate counts for the total quantum circuit.  
We hope that this paper will motivate further investigations into constant factor analyses for ODEs in the spirit of advancing the search for useful applications of quantum computers.

\section{Acknowledgments}
MP acknowledges funding from Mitacs and the NSERC Post-Graduate Doctoral Scholarship.  NW's work on this project was supported by the U.S. Department of Energy,
 Officeof Science, National Quantum Information Science Research Centers, Co-design Cen
ter for Quantum Advantage (C2QA) under contract number DE- SC0012704 (PNNL FWP
 76274) and PNNL’s Quantum Algorithms and Architecture
 for Domain Science (QuAADS) Laboratory Directed Research and Development (LDRD)
 Initiative. The Pacific Northwest National Laboratory is operated by Battelle for the
 U.S. Department of Energy under Contract DE-AC05-76RL01830. 
We thank Dong An and Danial Motlagh for insightful conversations, as well as Yixuan Liang for providing useful feedback that improved the quality of the paper.

\appendix

\section{$C_{\textup{max}}$ comparison} \label{sec:Cmax}
Here we prove a minor equivalence condition between the assumptions in Ref. \cite{jennings2024cost}, and those required by LCHS. Note that their generator $A$ is defined without a leading minus sign, so some of the inequalities therein are flipped when compared to this section. In \cite{jennings2024cost}, $C_{\textup{max}}$ is defined to be an upper bound on the norm of the propagator like so
\begin{equation}
    \|e^{-At}\| \leq C_{\textup{max}}.
\end{equation}
 
Therein, classes of \textit{stable dynamics} are investigated where the inequality $\min_i \mathfrak{Re} \lambda_i(A) > 0 $ is assumed to hold. This is called \textit{stable} as small perturbations do not lead to arbitrarily large increases in the norm of the solution vector. When this is the case, one can find positive operators $P\succeq 0$ such that $PA + A^\dagger P \succ 0$. In \cite{jennings2024cost}, for a linear ODE where $A$ is enforced to be positive definite, the complexity is dependent on the log norm: 
\begin{equation}
    \mu_P = \min_{\|x\| \neq 0}\mathfrak{Re}\frac{x^\dagger A^\dagger Px}{xPx},
\end{equation}
as the following inequality is used to provide a tighter analysis based on $\mu_P$:
\begin{equation}
    \|e^{At}\| \leq \sqrt{\kappa_P}e^{-t\mu_P}, 
\end{equation}
where $\kappa_P \equiv \|P\|\|P^{-1}\|$ (see \cite{jennings2024cost} for more details). 
In this work we are interested in the boundary case of \textit{unstable} dynamics, or the weakest conditions necessary for the LCHS formula to hold: $L \succeq 0$ which is equivalent to the real part of $A$ being positive semi-definite. In this scenario, the norm of the solution vector $\|u(t)\|$ will be non-increasing. This is the most general class of dynamics that are both suitable to quantum algorithms and \textit{non-fast-forwardable}. In the context of the inequality above, this occurs in the limit where $\mu_P \to 0^+$ and $\kappa_P \to C_{\textup{max}}^2$, and the boundary case is obtained by setting $C_{\textup{max}} \equiv1$. Now since we wish to make accurate comparisons to the numerics in Ref. \cite{jennings2024cost}, we need to establish an equivalence between our matrix assumptions. In the following Lemma, we show that by enforcing the matrix $L$ to be positive semi-definite, we recover the condition that $C_{\mathrm{max}} = 1$. This is not an additional assumption, but a requirement for the LCHS integral identity to hold in the first place.

\begin{lemma}[$C_{\mathrm{max}}$ equivalence]
Given a matrix $A = L+iH$ as previously defined, if $L\succeq 0$, then the limit that $\lambda_{\mathrm{min}(L)} \to 0$ implies that $C_{\text{max}} = 1$. 
\end{lemma}
\begin{proof}
We begin by showing that the norm of the propagator $\|e^{-At}\|$ is upper bounded by the dynamics generated by the real part of the generator $\|e^{-Lt}\|$, since $L = \mathfrak{Re}(A)$ :
    \begin{align}
        \|e^{-tA}\| &= \|e^{-t(L+iH)}\| \\
        &= \lim_{r \to \infty} \|e^{-\frac{t}{r}(L+iH)}\|^r \\
        & \leq \lim_{r \to \infty} \|e^{-\frac{t}{r}L}\|^r \|e^{-\frac{t}{r}iH}\|^r \\
        & = \|e^{-tL}\| \\
        & = e^{-t \lambda_{\text{min}}(L)}.
    \end{align}
    Now we take the following limit
    \begin{equation}
        \lim_{\lambda_{\text{min}(L)} \to 0} e^{-t \lambda_{\text{min}}} = 1,
    \end{equation}
    which then yields the total expression
    \begin{equation}
        \lim_{\lambda_{\mathrm{min}}(\mathfrak{Re}(A)) \to 0} \|e^{-tA}\| \leq 1.
    \end{equation}
    This is equivalent to the condition of setting $C_{\text{max}} \equiv 1$.  
\end{proof}
Interestingly, the LCHS formula also holds if the initial state $u(0)$ is supported on a positive semi-definite subspace of $L$ \cite{an2023quantum}.

\section{Gaussian Quadrature and Calculation of $\|c\|_1$} \label{sec:gauss}
In this Section, Gaussian quadrature is briefly reviewed up to the point where the calculation of $\|c\|_1$ can be clearly illustrated. \\

Gaussian quadrature is a numerical integration scheme that is exact for polynomials of degree $2n-1$ when $n$ points are used. For example, for some degree $2n-1$ polynomial $p_{2n-1}(x)$, Gaussian quadrature formula provides the means to exactly represent its integral as a sum such that $\int_{-1}^1 p_{2n-1}(x) \dd x =\sum_{i=1}^n \omega_i \; p_{2n-1}(x_i)$. Here, $x_i$ is the $i$-th root of $n$-th Legendre polynomial $P_n(x)$, and the coefficients $\omega_i$ are weights that have the form 
\begin{equation}
    \omega_i = \frac{2}{(1-x_i^2)(P_n'(x_i))^2}.
\end{equation}
Given that this construction is exact for polynomials, these formulas yield rapid convergence for functions that can be accurately approximated by polynomials. In this paper, the functions of interest are
\begin{equation}
    g(k) = \frac{1}{C_\beta (1-ik) e^{(1+ik)^\beta}},
\end{equation}
for $\beta \in (0,1)$, and the unitaries
\begin{equation}
    U(t, k) = e^{-it(kL+H)},
\end{equation}
each of which can be well approximated by polynomials via Taylor series. For our discretization, the starting point is below
\begin{equation} 
    \int_{-K}^K g(k) U(t,k)dk = \sum_{m=-K/h}^{K/h-1}\int_{mh}^{(m+1)h} g(k) U(t,k) dk \approx \sum_{m=-K/h}^{K/h-1} \sum_{q=1}^Q c_{q,m} U(t,k_{q,m}). 
\end{equation}
The interval for Gaussian quadrature is [-1,1], which we need to rescale arbitrarily to [a,b] for our formula. We perform the following transformation
\begin{equation}
    \int_a^b f(x) \dd x = \int_{-1}^1 f \left(\frac{b-a}{2}\zeta + \frac{a+b}{2} \right ) \frac{\dd x}{\dd \zeta} \dd \zeta, 
\end{equation}
where $\frac{\dd x}{\dd \zeta}=\frac{b-a}{2}$. For a single integral time step of our formula becomes 
\begin{equation}
    \frac{h}{2}\int_{-1}^1 f\left(\frac{h}{2}\zeta + \frac{(2m+1)h}{2} \right ) \dd \zeta. 
\end{equation}
Now after applying the Gaussian quadrature formula, we are then left with 
\begin{align}
    \frac{h}{2}\int_{-1}^1 f\left(\frac{h}{2}\zeta + \frac{(2m+1)h}{2} \right ) \dd \zeta &\approx \frac{h}{2} \sum_{q=1}^Q \omega_q g\left(\frac{h}{2}\zeta_q + \frac{(2m+1)h}{2} \right ) U(t,k_{q,m}), \\
    & = \sum_{q=1}^Q c_{q,m} U(t,k_{q,m})
\end{align}
where $m$ is an integer to index the outer sum, $\omega_q$ are the Gaussian quadrature weights, and $g(k)$ is our chosen kernel function for LCHS. Now to calculate $\|c\|_1$, we need to sum the modulus of all factors that multiply $U(t,k_{q,m})$. This gives $\|c\|_1$ the following definition:
\begin{equation} \label{eq:c1_def}
    \|c\|_1 := \frac{h}{2} \sum_{m=-K/h}^{K/h-1} \sum_{q=1}^Q \left | \omega_q g\left(\frac{h}{2}\zeta_q + \frac{(2m+1)h}{2} \right ) \right |.
\end{equation}
Conveniently, $\mathtt{NumPy}$ has functions to compute the $\omega_q$ factors for the canonical interval [-1,1] which is exactly what we require after performing the above transformation. Therefore, this provides the means to simply compute the quantity $\|c\|_1$. 

\section{Constant Factor Degree Polynomial Approximation to the Sign Function}\label{app:polynomial_degree_bound}
In this Section we study the constant factor degree $n$ of a polynomial $p_{n,\Delta, \epsilon}(x)$ that approximates the signum or sign function $\sign(x)$ in an interval parameterized by $\Delta$ up to precision $\epsilon$. We follow the strategy of Ref. \cite{low2017hamiltonian} in that we first approximate $\sign(x)$ with $\erfunc(x)$, and then given an entire polynomial approximation to $\erfunc(x)$, we show that this well approximates $\sign(x)$ over a restricted domain. Much of our work here is to convert the asymptotic results of Ref. \cite{low2017hamiltonian} into constant factor bounds. This Section will also rely heavily on the analysis presented in Ref. \cite{sachdeva2014faster}. In addition, we first bound the degree of this polynomial in terms of elementary functions, and then once again use Lambert-W functions to tighten the analysis at the end of the section. 

\begin{lemma}[Approximating the signum function with the error function (adapted from Ref. \cite{low2017hamiltonian}).] \label{lem:signum_with_erf} $\forall x \in \mathbb{R}$ let $k=\frac{\sqrt{2}}{\Delta}\log^{1/2}(2/\pi \epsilon^2)$ and let $\Delta \in \mathbb{R}^+$ define an interval where the function $f_{\mathrm{sgn},k, \Delta, \epsilon} := \erfunc(kx) \; \mathbf{s.t.} \; |x| \geq \Delta/2$ well approximates $\sign(x)$ such that 
\begin{equation} \label{eq:sig_eps}
    \max_{|x|\geq \Delta/2} \left |f_{\mathrm{sgn},k, \Delta, \epsilon}(x) - \sign(x) \right | \leq \epsilon, 
\end{equation}
where 
\begin{equation}
\sign(x) := 
     \begin{cases} 
      1, & x > 0, \\
      0, & x=0, \\
      -1, & x < 0, 
   \end{cases}
\end{equation}
provided that $\epsilon \in (0, \sqrt{2/e\pi}]$. 
\end{lemma}
\begin{proof}
    A proof is contained in  [\cite{low2017hamiltonian}, Lemma 10]. 
\end{proof}

Next we wish to find a polynomial approximation of $\erfunc(x)$. In pursuit of this goal, we first find a constant factor degree approximation to the Gaussian, and then use this result to bound the degree of the polynomial required to approximate the error function. The polynomial approximation for the Gaussian is presented in the following Lemma.  

\begin{lemma}[Entire Approximation to the Gaussian Function] \label{lem:gauss}
    There exists a polynomial $p_{\mathrm{gauss}, k, n}(x)$ of degree $ n=\ceil{\sqrt{8\ceil{\max\{\ln(2/\epsilon), k^2 e^2 /2 \} }\ln(4/\epsilon)}}$ that well approximates the Gaussian function $f_{\mathrm{gauss}}(x) = e^{-(k x)^2}$ within precision $\epsilon$ such that 
    \begin{equation}
        \max_{x\in [-1,1]}\left | p_{\mathrm{gauss}, k, n}(x) - e^{-(k x)^2} \right | \leq \epsilon.
    \end{equation}
\end{lemma}
\begin{proof}
    Our starting point is Lemma 4.2 from \cite{sachdeva2014faster} where a polynomial approximation to the function $f(y) = e^{-\lambda (y+1)}$ is analyzed. The reason for this choice of function is that studying this function over the canonical interval $y\in [-1,1]$ is equivalent to analyzing $g(y) = e^{- y} \; | \; y\in[0, t]$ where $\lambda = t/2$. As noted in \cite{low2017hamiltonian}, this analysis can be repurposed for the Gaussian function via the change of variables $y\to 2x^2-1$, and $\lambda \to k^2/2$. Our proof here follows these works quite closely, however, we write it out in detail to carefully account for all constants. Now a straightforward way to approximate the exponential function with Chebyshev polynomials is to expand it via Taylor series and implement a Chebyshev approximation to the monomial at each order like so 
    \begin{equation}
        e^{-\lambda (y+1)} = e^{-\lambda} \sum_{j=0}^\infty \frac{(-\lambda)^j}{j!}y^j \approx e^{-\lambda} \sum_{j=0}^\infty \frac{(-\lambda)^j}{j!} p_{\mathrm{mon}, j, d}(y),
    \end{equation}
    where $p_{\mathrm{mon}, j, n}(y)$ is an at most degree-$n$ Chebyshev polynomial approximation to the monomial. The work of Ref. \cite{sachdeva2014faster} bounds the error in this approximation in the following way: 
    \begin{align}
    \sup_{y \in [-1, 1]} & \left | e^{-\lambda (y+1)} - e^{-\lambda} \sum_{j=0}^\infty \frac{(-\lambda)^j}{j!} p_{\mathrm{mon}, j, d}(y)\right | \\
    & \leq \sup_{y \in [-1, 1]}  \left |e^{-\lambda} \sum_{j=0}^t \frac{(-\lambda)^j}{j!} \left (y^j - p_{\mathrm{mon}, j, d}(y) \right ) \right | + \sup_{y \in [-1, 1]}  \left | e^{-\lambda} \sum_{j=t+1}^\infty \frac{(-\lambda)^j}{j!} y^j\right | \\
    &\leq \sup_{y \in [-1, 1]}  e^{-\lambda} \sum_{j=0}^t \frac{(-\lambda)^j}{j!} \left |y^j - p_{\mathrm{mon}, j, d}(y) \right |  + e^{-\lambda} \sum_{j=t+1}^\infty \frac{(\lambda)^j}{j!} .
    \end{align}
    For the first term, Ref. \cite{sachdeva2014faster} uses the properties of Chebyshev polynomials to show that $$\sup_{y \in [-1, 1]}  e^{-\lambda} \sum_{j=0}^t \frac{(-\lambda)^j}{j!} \left |y^j - p_{\mathrm{mon}, j, n}(y) \right |  \leq 2e^{-n^2/2t}.$$ For the second term, they use $j! \geq (j/e)^j$ to bound $e^{-\lambda} \sum_{j=t+1}^\infty \frac{(\lambda)^j}{j!} \leq e^{-\lambda} \sum_{j=t+1}^\infty \left (\frac{\lambda e}{j}\right )^j$. This summand can be further upper bounded using $\frac{\lambda^t e^t}{t^t} \leq e^{-t}$, which comes directly from assuming that $\lambda e^2 \leq t$ as can be seen from a few lines of algebra. We then have that $$e^{-\lambda} \sum_{j=t+1}^\infty \frac{(\lambda)^j}{j!} \leq e^{-\lambda} \sum_{j=t+1}^\infty \left (\frac{\lambda e}{j}\right )^j \leq e^{-\lambda} \sum_{j=t+1}^\infty e^{-j} =\frac {e \;e^{-\lambda -t}}{1-1/e} \leq e^{-\lambda -t},$$
    where the second last equality comes from the geometric series and we lower bounding the denominator by 1/2. The total error now reads 
    \begin{equation} \label{eq:poly_error}
        \sup_{y \in [-1, 1]} \left | e^{-\lambda (y+1)} - e^{-\lambda} \sum_{j=0}^\infty \frac{(-\lambda)^j}{j!} p_{\mathrm{mon}, j, n}(y)\right | \leq 2e^{-n^2/2t} + e^{-\lambda -t}. 
    \end{equation}
    Now, applying the transformation $y\to 2x^2-1$ and $\lambda \to k^2/2$, we can convert this bound to a bound on the error of a polynomial approximation to the Gaussian function. Examining the Taylor series of the Gaussian $e^{-x^2k ^2} = \sum_{j=0}^\infty \frac{(-1)^jk^{2j}}{j!}x^{2j}$ the Chebyshev approximation to the monomials are doubled in degree and shifted yielding the bound 
    \begin{equation} \label{eq:gauss_error}
        \epsilon_{\mathrm{gauss}, k, n} := \sup_{x \in [-1, 1]} \left | e^{-(x k)^2} - e^{-k^2/2} \sum_{j=0}^\infty \frac{(-k^2/2)^j}{j!} p_{\mathrm{mon}, j, n/2}(x)\right | \leq 2e^{-n^2/8t} + e^{-k^2/2 -t}.
    \end{equation}
    Choosing $n=\ceil{\sqrt{8t\ln(4/\epsilon)}}$ and $t = \ceil{\max\{\ln(2/\epsilon), k^2 e^2 /2 \} }$ ensures that $2e^{-n^2/8t} + e^{-k^2/2 -t} \leq \epsilon/2 + e^{-k^2/2}\epsilon/2 \leq \epsilon$ which completes the proof. Note that the second choice of $t$ came from the constraint used to bound the Taylor series.
\end{proof}
Despite this analysis being in terms of a Taylor series of polynomial approximations to monomials, this is not the series expansion used in Ref \cite{low2017hamiltonian}. Instead, the work implements the truncated Jacobi-Anger expansion of the exponential function which has the form
\begin{equation}
     e^{-\lambda(y+1)} \approx p_{\mathrm{exp}, n}(y)  = e^{-\lambda}\left (I_0(\lambda) +  2\sum_{j=0}^n I_j(\lambda) T_j(-y) \right ),
\end{equation}
where $T_j(y)$ are the $j-$th order Chebyshev functions of the first kind, and $I_j(\lambda)$ are modified Bessel functions of the first kind. The reason for the use of this expansion is that Jacobi-Anger expansions have favorable properties for QSP, specifically in that they satisfy the boundedness condition that the range of the polynomial is $\in [-1.1]$, which is required for unitarity. However, despite this difference, in [Ref. \cite{low2017hamiltonian}, Lemma 14] it is shown that 
\begin{equation}
    \max_{y\in [-1,1]}|p_{\mathrm{exp}, n}(y) - e^{-\lambda(y+1)}| \leq 2e^{-n^2/2t} + e^{-\lambda -t},
\end{equation}
which is the same bound that was obtained in Equation \ref{eq:poly_error}, meaning we can proceed to bound the degree of the Jacobi-Anger series using the analysis in Lemma \ref{lem:gauss}. Next, we continue with this analysis to bound the degree of a polynomial approximation of the shifted error function. 

\begin{lemma}[Entire Approximation to Shifted Error Function] \label{lem:erf_approx} There exists a polynomial $p_{\mathrm{erf}, n, k, s}(x)$ of degree $n$ that well approximates the shifted error function $f(x) = \erfunc(k(x-s)),$ $s\in[-1,1]$ to error $\epsilon$ such that 
    \begin{equation}
        \max_{x\in [-1,1]}|p_{\mathrm{erf}, n, k, s}(x) - \erfunc(k(x-s))| \leq \frac{8k}{n\sqrt{\pi}} \left ( 2e^{-(n-1)^2/8t} + e^{-2k^2 -t} \right ),
    \end{equation}
    where $n, t \in \mathbb{Z^+}$.
\end{lemma}
\begin{proof}
    We work with the quantity $\epsilon_{\mathrm{gauss}, k, n}$ defined in Equation \ref{eq:gauss_error} of Lemma \ref{lem:gauss}. In Corollaries 4 and 5 of Ref. \cite{low2017hamiltonian} it is proven that 
    \begin{equation}
        \max_{x\in [-1,1]}|p_{\mathrm{erf}, n, k, s}(x) - \erfunc(k(x-s))| \leq \frac{8k}{n\sqrt{\pi}}\epsilon_{\mathrm{gauss}, 2k, n-1},
    \end{equation}
    which allows us to obtain the bound simply by using the definition of $\epsilon_{\mathrm{gauss}, k, n}$ and updating the coefficient and shift to $k$ and $n$ respectfully. 
\end{proof}

With a polynomial approximation to the error function with a constant factor error bound we are in position to analyze an approximation to the sign function. 

\begin{theorem}[Polynomial Approximation to the Sign Function] \label{thm:sgn_approx}
    There exists an odd degree-$n$ polynomial $p_{\mathrm{erf}, n, k, s, \Delta}(x)$ that well approximates the function $f(x) = \sign(x-s)$ such that 
    \begin{equation}
        \max_{x \in [-1,\; s-\Delta/2 ] \cup [s+ \Delta/2,\; 1]} \left |p_{\mathrm{erf}, n, k, s, \Delta}(x) - \sign(x-s) \right | \leq \epsilon,  
    \end{equation}
    where the degree of the polynomial is 
    \begin{equation}
        n = \ceil{\sqrt{8\ceil{\max \left \{ \ln\left (\frac{32\frac{\sqrt{2}}{\Delta}\ln^{1/2}(8/\pi \epsilon^2)}{3\sqrt{\pi} \epsilon}\right ), \frac{4}{\Delta^2}\ln(8/\pi \epsilon^2) e^2 \right  \}}\ln\left (\frac{64\frac{\sqrt{2}}{\Delta}\ln^{1/2}(8/\pi \epsilon^2)}{3\sqrt{\pi} \epsilon}\right )}+1}
    \end{equation}
    such that $\Delta \in \mathbb{R}^+$ parametrizes the interval of approximation, and that $\epsilon \in (0, 2\sqrt{2/e\pi})$.  
\end{theorem}
\begin{proof}
    We start from the definition of the error and apply the triangle inequality:
    \begin{align}
        \max_{x \in [-1,\; s-\Delta/2 ] \cup [s+ \Delta/2,\; 1]} & \left |p_{\mathrm{erf}, n, k, s, \Delta}(x) - \sign(x-s) \right | \\ 
        \leq &\max_{x \in [-1,\; s-\Delta/2 ] \cup [s+ \Delta/2,\; 1]}|p_{\mathrm{erf}, n, k, s, \Delta}(x) - \erfunc(k(x-s))| \notag \\
        &+ \max_{x \in [-1,\; s-\Delta/2 ] \cup [s+ \Delta/2,\; 1]}|\erfunc(k(x-s)) - \sign(x-s)| \label{eq:poly_triangle} \\
        \leq & \frac{8k}{n\sqrt{\pi}} \left ( 2e^{-(n-1)^2/8t} + e^{-2k^2 -t} \right ) + \frac{\epsilon}{2} 
    \end{align}
    where in the second line, the first of the terms can be bounded directly with the result of Lemma \ref{lem:erf_approx}. We bound the second term by $\epsilon/2$, which by Lemma \ref{lem:signum_with_erf}, can be done by setting $k=\frac{\sqrt{2}}{\Delta}\ln^{1/2}(8/\pi \epsilon^2)$. Now if we upper bound the multiplicative $\frac{1}{n} \leq \frac{1}{3}$, given that certainly the degree of the polynomial $n>1$, and is necessarily odd, then by choosing $n = \ceil{\sqrt{8t\ln\left (\frac{64k}{3\sqrt{\pi} \epsilon}\right )}+1}$ and $t = \ceil{\max \left \{ \ln\left (\frac{32k}{3\sqrt{\pi} \epsilon}\right ), 2k^2 e^2 \right  \}}$ we obtain 
    \begin{align}
        \frac{8k}{n\sqrt{\pi}} \left ( 2e^{(n-1)^2/8t} + e^{-2k^2 -t} \right ) + \frac{\epsilon}{2} &\leq \frac{8k}{3\sqrt{\pi}} \left(2\frac{3\sqrt{\pi} \epsilon}{64k} +  \frac{3\sqrt{\pi} \epsilon}{32k} \right) + \frac{\epsilon}{2} \\
        & = \frac{\epsilon}{4} + \frac{\epsilon}{4} + \frac{\epsilon}{2} = \epsilon.
    \end{align}
    Substituting the values for $t$ and $k$ into the formula for $n$ completes the proof. 

We can simplify the degree of the polynomial for the case of amplitude amplification with the following Corollary. 
\begin{corollary}
\label{coro:degree_bound}
    For $\epsilon \in (0, 2\sqrt{2/e\pi}]$ and $\Delta \in (0, 2]$ the degree of the polynomial approximation to the sign function $p_{\mathrm{erf}, n, k, s, \Delta}(x)$ simplifies to 
\begin{equation}
      n = \ceil{\sqrt{8\ceil{\frac{4}{\Delta^2}\ln(8/\pi \epsilon^2) e^2 }\ln\left (\frac{64\frac{\sqrt{2}}{\Delta}\ln^{1/2}(8/\pi \epsilon^2)}{3\sqrt{\pi} \epsilon}\right )}+1}.
\end{equation}
\end{corollary}
For the case of amplitude amplification $\Delta\leq 2$, which allows us to show that the \textit{max} of $t$ in Theorem \ref{thm:sgn_approx} is always taken by $\frac{4}{\Delta^2}\log(8/\pi \epsilon^2) e^2$. Let $x=\frac{4e^2}{\Delta^2}\log(8/\pi \epsilon^2)$. Then, for $\epsilon>0$, the max expression becomes
\begin{align}
    \max \left \{ \frac{1}{2}\ln x +\frac{1}{2}\ln\left (\frac{64}{9e^2}\frac{8}{\pi\epsilon^2}\right ), x \right  \}.
\end{align}
We show that, for $0<\Delta\leq 2$,
\begin{align}
    \frac{1}{2}\ln x +\frac{1}{2}\ln\left (\frac{64}{9e^2}\frac{8}{\pi\epsilon^2}\right )\leq x .
\end{align}
Since $0<\Delta\leq 2$ implies $\Delta^2\leq 4$, we have
\begin{align}
    x=\frac{4e^2}{\Delta^2}\ln(8/\pi \epsilon^2)\geq e^2\ln(8/\pi \epsilon^2).
\end{align}
This implies that
\begin{align}
    \frac{1}{2}\ln x +\frac{1}{2}\ln\left (\frac{64}{9e^2}\frac{8}{\pi\epsilon^2}\right )&=  \frac{1}{2}\ln x +\frac{1}{2}\ln\left (\frac{64}{9e^2}\right )+\frac{1}{2e^2}e^2\ln\left (\frac{8}{\pi\epsilon^2}\right )\\
    &\leq   \frac{1}{2}\ln x +\frac{1}{2}\ln\left (\frac{64}{9e^2}\right )+\frac{1}{2e^2}x\\
    &\leq   \frac{\ln x+x}{2}\\
    &\leq x.
\end{align}
This establishes that for $0<\Delta\leq 2$, and for $\epsilon>0$, the result of the max function is $\frac{4}{\Delta^2}\log(8/\pi \epsilon^2) e^2$. 
\end{proof}

In practice, the use of the ceiling and the factor of unity represent trivial additions to the complexity, and in our numerical studies, can be ignored. In this case, the result can then be further simplified with an upper bound that is convenient for analysis in Section \ref{sec:numerics} to lift the constraints of an optimization problem.

\begin{proposition} \label{prop:degree_upper}
    For $\epsilon \in (0, 2\sqrt{2/e\pi}]$ and $\Delta \in (0, 2]$ the function
    \begin{equation}
        n={\sqrt{{\frac{32}{\Delta^2}\ln(8/\pi \epsilon^2) e^2 }\ln\left (\frac{64\frac{\sqrt{2}}{\Delta}\ln^{1/2}(8/\pi \epsilon^2)}{3\sqrt{\pi} \epsilon}\right )}}.
    \end{equation}
    can be upper bounded as
    \begin{align}
        n &\leq \frac{8e}{\Delta}\sqrt{1+\frac{1}{e}}\ln\left (\frac{64\sqrt{2}}{3 \sqrt{\pi} \Delta \epsilon} \right ). 
    \end{align}
\end{proposition}
\begin{proof}
    Our starting point is the expression 
    \begin{equation}
        n = \sqrt{\frac{32}{\Delta^2}\ln(8/\pi \epsilon^2) e^2 \ln\left (\frac{64\frac{\sqrt{2}}{\Delta}\ln^{1/2}(8/\pi \epsilon^2)}{3\sqrt{\pi} \epsilon}\right )},
    \end{equation}
    which can be simplified by defining $y:=\ln(8/\pi \epsilon^2)$ and rewriting :
    \begin{align}
        n &= \sqrt{\frac{32e^2}{\Delta^2} y  \ln\left (\frac{64\frac{\sqrt{2}}{\Delta}\sqrt{y}}{3\sqrt{\pi} \epsilon}\right )} \\
        &= \sqrt{\frac{16e^2}{\Delta^2} y \left [ \ln\left (\frac{8192 }{9 \pi \Delta^2 \epsilon^2}\right ) + \ln(y)\right ]}, \\
    \end{align}
    and if we note that $\ln\left (\frac{8192 }{9 \pi \Delta^2 \epsilon^2}\right ) \geq y$ for the region of interest $\Delta \leq 2$, then we can further upper bound $n$ using this fact. Letting $a := \frac{8192 }{9 \pi \Delta^2 \epsilon^2}$
    \begin{align}
        n \leq \sqrt{\frac{16e^2}{\Delta^2} \left [\ln^2(a) + \ln(a)\ln\ln(a) \right ]}.
    \end{align}
    Now we want to find some constant $C$ such that $\ln^2(a) + \ln(a)\ln\ln(a) \leq C\ln^2(a) \; \forall \; a>1$ in order to obtain the form given in the Proposition. To see this, set $z = \ln(a)$:
    \begin{align}
        z\ln(z) + z^2 &\leq Cz^2 \\
        \frac{z\ln(z)}{z^2} + 1 &\leq C, 
    \end{align}
    and notice that the function on the left hand side has a maximum at $1+1/e \; \forall \; a>1$. Substituting this value for the constant and taking the root yields the expression claimed. 
\end{proof}

Note that the above bounds can be further improved if, sharing in the theme with the rest of the paper, we use Lambert-W functions in the analysis. 
\begin{theorem}[Sign Function Approximation with Lambert Bound] \label{thm:tighter_sgn_approx}
    There exists an odd degree-$n$ polynomial $p_{\mathrm{erf}, n, k, s, \Delta}(x)$ that well approximates the function $f(x) = \sign(x-s)$, $s\in [-1,1]$ such that 
    \begin{equation}
        \max_{x \in [-1,\; s-\Delta/2 ] \cup [s+ \Delta/2,\; 1]} \left |p_{\mathrm{erf}, n, k, s, \Delta}(x) - \sign(x-s) \right | \leq \epsilon,  
    \end{equation}
    where for some constant $b \in \mathbb{R}_{>1}$ , $k=\frac{1}{\Delta}\sqrt{2W_0\left(\frac{2b^2}{\pi \epsilon^2} \right)}$, and $W_0(\cdot)$ is the Lambert-W function, the degree of the polynomial is 
    \begin{equation}
        n = \ceil{2\sqrt{2}ke\sqrt{W_0\left[\frac{1}{2k^2e^2}\left ( \frac{(b-1)\sqrt{\pi}\epsilon}{8kb} - \frac{e^{-2k^2(1+e^2)}}{\bar{n}}\right )^{-2} \right ]}+1},
    \end{equation}
    given that the following holds 
    \begin{equation}
        \Delta \leq \sqrt{8(1+e^2) - \frac{64}{e(b-1)^2\bar{n}^2}} \leq 8.1917.
    \end{equation}
\end{theorem}
\begin{proof}
    The starting point of the proof is Equation \ref{eq:poly_triangle}, from the proof of Theorem \ref{thm:sgn_approx}, where instead of an equal error split, we choose an arbitrary allocation
    \begin{align}
            \max_{x \in [-1,\; s-\Delta/2 ] \cup [s+ \Delta/2,\; 1]} 
            \leq &\max_{x \in [-1,\; s-\Delta/2 ] \cup [s+ \Delta/2,\; 1]}|p_{\mathrm{erf}, n, k, s, \Delta}(x) - \erfunc(k(x-s))| \notag \\
            &+ \max_{x \in [-1,\; s-\Delta/2 ] \cup [s+ \Delta/2,\; 1]}|\erfunc(k(x-s)) - \sign(x-s)| \\
            \leq &  \left |p_{\mathrm{erf}, n, k, s, \Delta}(x) - \sign(x-s) \right | \leq \frac{8k}{n\sqrt{\pi}} \left ( 2e^{-(n-1)^2/8t} + e^{-2k^2 -t} \right ) + \frac{\epsilon}{b},
    \end{align}
    where this choice implies that 
    \begin{equation}
        \frac{8k}{n\sqrt{\pi}} \left ( 2e^{-(n-1)^2/8t} + e^{-2k^2 -t} \right ) \leq \frac{(b-1)\epsilon}{b},
    \end{equation}
    since the total error adds to $\epsilon$. Now the first term above has an apparent solution for $n$ related to the Lambert-W function. We use upper bounds to obtain this form:
    \begin{equation}
        \frac{2e^{-(n-1)^2/8t}}{n} + \frac{e^{-2k^2 -t}}{n} \leq \frac{2e^{-(n-1)^2/8t}}{n-1} + \frac{e^{-2k^2 -t}}{\bar{n}},
    \end{equation}
    where $\bar{n}$ is the minimum degree of the polynomial ($n\geq \bar{n}$). Subbing this in to the prior equation and rearranging we obtain:
    \begin{equation}
        \frac{2e^{-(n-1)^2/8t}}{n-1}  \leq \left( \frac{(b-1)\sqrt{\pi}\epsilon}{8kb} - \frac{e^{-2k^2 -t}}{\bar{n}} \right ).
    \end{equation}
    By saturating this inequality, the equation can be solved for $n$: 
    \begin{equation}
        n = \ceil{2\sqrt{tW_0\left[t^{-1}\left ( \frac{(b-1)\sqrt{\pi}\epsilon}{8kb} - \frac{e^{-(2k^2+t)}}{\bar{n}}\right )^{-2} \right ]}+1}. 
    \end{equation}
    Next, it remains a question of how to determine $t$, which recalling from \ref{lem:erf_approx}, is the degree of the truncation of the Taylor series for the exponential function. From before, we have the condition $t\geq 2k^2e^2$, and our result above imposes another condition to ensure $W_0(x)$ is a valid solution (we require $x\in(0,\infty)$ Looking at the argument of the function it is clear that to maintain positivity we require
    \begin{equation}
        t \geq \ln \left( \frac{8kbe^{-2k^2}}{\bar{n}(b-1)\sqrt{\pi}\epsilon}\right ), 
    \end{equation}
    which yields the condition that $t = \ceil{\max\left \{2k^2e^2, \ln \left( \frac{8kbe^{-2k^2}}{\bar{n}(b-1)\sqrt{\pi}\epsilon}\right )  \right \} }$. Next we wish to find criteria whereby we can resolve this max expression. Given that thew first term under the \textit{max} operation scales quadratically with $k$ whereas the left term scales only logarithmically, we expect the left term to dominate in almost all cases except for large $\Delta$, which is usually not the case of interest. Therefore, we consider the inequality
        \begin{equation}
        \ln \left( \frac{8kbe^{-2k^2}}{\bar{n}(b-1)\sqrt{\pi}\epsilon} \right ) \leq 2k^2e^2,   
    \end{equation}
    expand and isolate the $k$-dependent terms
    \begin{equation}
        \ln \left( \frac{8b}{\bar{n}(b-1)\sqrt{\pi}\epsilon} \right ) + \ln(\epsilon^{-1}) - (1+e^2)2k^2 + \ln(k) \leq 0.
    \end{equation}
    Recall that from \ref{lem:signum_with_erf} $k \geq \frac{\sqrt{2}}{\Delta} \ln^{1/2}\left(\frac{2b^2}{\pi \epsilon^2}\right)$, where we have a factor of $b^2$ from the substitution $\epsilon \to \epsilon/b$ from earlier in the proof. Next we define $\ell :=\ln^{1/2}\left(\frac{2b^2}{\pi \epsilon^2}\right)$ and simplify
    \begin{equation}
        \underbrace{\ln\left (\frac{8}{\Delta (b-1) \bar{n}} \right ) + \left(\frac{1}{2} - \frac{4(1+e^2)}{\Delta^2} \right )\ell + \frac{1}{2}\ln \ell}_{:=f(l)} \leq 0.
    \end{equation}
    We then define $f(\ell)$, and fixing $\Delta, b, \bar{n}$, take its first and second derivative:
    \begin{align}
        \frac{\partial f}{\partial \ell}&=  \left(\frac{1}{2} - \frac{4(1+e^2)}{\Delta^2} \right ) + \frac{1}{2\ell}\\
        \frac{\partial^2 f}{\partial \ell^2} &=  - \frac{1}{2\ell^2},
    \end{align}
    given that $\frac{\partial^2 f}{\partial \ell^2}<0$, $f(l)$ is strictly concave. Now using elementary calculus principles, we solve for the maximum of the function, and find under which conditions it is negative, as the concavity then guarantees the satisfaction of the inequality. Setting $f'(\ell) = 0$ we find a maximum at $\ell_0 = \frac{\Delta^2}{8(1+e^2)-\Delta^2}$. Then taking $f(\ell_0)\leq 0$ we obtain the condition 
    \begin{equation} \label{eq:valid_delta}
        \Delta \leq \sqrt{8(1+e^2) - \frac{64}{e(b-1)^2\bar{n}^2}}.
    \end{equation}
    Finally, from Lemma 10 of \cite{low2017hamiltonian}, we have $k\geq \frac{\sqrt{2}}{\Delta}\ln^{1/2}\left(\frac{2b^2}{\pi \epsilon^2}\right) \geq \frac{\sqrt{2}}{\Delta} \sqrt{W_0\left (\frac{2b^2}{\pi \epsilon^2} \right )}$. Given that we proved Equation \ref{eq:valid_delta} for the case that the constants must satisfy the middle bound, we can further tighten the result by using the right most inequality for $k$ directly which completes the proof.  
\end{proof}
Note that proving a condition like Equation \ref{eq:valid_delta} would've not been possible were the Lambert bound used directly from the outset. 
\section{Error-budgeting version of robust Hamiltonian simulation}

Here we prove a slightly tighter and generalized version of the Robust block-Hamiltonian simulation result in Corollary 62 of the preprint version of \cite{gilyen2019quantum}.
The value of this generalization is that 
it accommodates the setting where the error of the input block encoding of $H$ (i.e. $\epsilon_{be}$) is not necessarily a fixed multiple of the target error (i.e. $\epsilon$).
That is, it describes the costs of robust Hamiltonian simulation when the user-required precision $\epsilon$ is not linked to the error in the Hamiltonian block encoding $\epsilon_{be}$.
In the original version, the allocation of the error budget was made for the user.

\begin{lemma}[Robust Hamiltonian Simulation by Qubitization with GQSP]
\label{lem:flex_ham_sim}
Let $t \in \mathbb{R}^+$, $\epsilon \in (0,1)$, and let $\widetilde{U}$ be a $(\alpha, a, \epsilon_{be})$-block-encoding of the Hamiltonian $H$. Then we can implement a unitary $\widetilde{V}$ that is a $(1,a+2,\epsilon)$-block-encoding 
of \( e^{-itH} \), using $n+2$ applications of the controlled $\widetilde{U}$ or its controlled inverse $\widetilde{U}^\dagger$ where 
\begin{equation}
    n= \ceil{\frac{e}{2}\alpha t +  \ln \left (\frac{2\eta}{\epsilon_{\mathrm{exp}}}\right ) }
\end{equation}
given that $\epsilon_{\mathrm{exp}} +  t\epsilon_{be} \leq \epsilon$, where the constant $\eta = 4(\sqrt{2\pi}e^{1/13})^{-1} \approx 1.47762.$ 
\end{lemma}
\begin{proof}
Let $H$ be a Hamiltonian with eigenvectors and eigenvalues defined such that $H\ket{\lambda_j} = \lambda_j \ket{\lambda_j}$. Utilizing the \textit{prepare} and \textit{select} oracles previously defined, it is possible to construct a walk operator 
\begin{equation}
    W = -(\mathbb{I} - 2\textup{PREP}\ketbra{0}{0}\textup{PREP}^\dagger) \textup{SELECT}
\end{equation} which effectively block diagonalizes $H$ into a series of 2-dimensional subspaces with eigenvalues $e^{\pm i\arccos(\lambda_j/\alpha)}$ as shown in \cite{motlagh2024generalized,low2017optimal}. In order to implement the time evolution operator $e^{-itH}$, one can utilize general quantum signal processing (GQSP) \cite{motlagh2024generalized} to implement the following truncated Jacobi-Anger series 
\begin{equation}
    e^{-it\cos(\theta)} \approx \sum_{k=-n}^{n} (-i)^k J_k(t) e^{ik\theta}.
\end{equation}
The art of GQSP lifts the standard parity restrictions of QSP and allows for a direct implementation of the polynomial of indefinite parity by constructing the following signal operator 
\begin{equation}
\mathcal{A = }
    \begin{bmatrix}
    W & 0 \\
    0 & \mathbb{I}
    \end{bmatrix},
\end{equation}
and allowing for arbitrary $\textup{SU}(2)$ rotations in the QSP sequence \cite{motlagh2024generalized}. If the degree of the polynomial $n$ is chosen to be sufficiently large this leads to an $\epsilon$-precise approximation of the time evolution operator. Noting that the series index runs from $[-n, n]$, implementing said polynomial with GQSP requires $2n$ queries to the block encoding of $H$ accessed via $\textup{PREP}$ and $\textup{SELECT}$. Furthermore, recent work \cite{berry2024doubling} has shown it is possible to improve the number of queries to just $n$. This is achieved by instead defining the signal operator 
\begin{equation}
\mathcal{B = }
    \begin{bmatrix}
    W & 0 \\
    0 & W^\dagger
    \end{bmatrix}
\end{equation}
to encode a polynomial of definite parity, and then apply an additional 2 controlled queries of $W$ to return to a degree of indefinite parity. This is shown to result in a phase doubling, which takes the indices of the truncated series to $[-n/2, n/2]$, thereby reducing the overall block encoding queries to $n$. In Ref. \cite{jennings2023efficient}, it is shown that the degree of the polynomial can be chosen to be $$n = \ceil{\frac{e}{2}\alpha t + \log\left(\frac{2\eta}{\epsilon}\right )},$$ in order to achieve an $\epsilon$-precise approximation to the time evolution operator. 
Now to prove robustness, we need to consider 2 errors in this setting; the block encoding error $\epsilon_{be}$ and the error that comes from implementing the operator exponential imperfectly $\epsilon_{\mathrm{exp}}$ (due to truncation). If $\widetilde{V}(\widetilde{H})$ is the evolution operator that is implemented, where the tildes are used to indicate an imperfect operator exponential of an imperfect block encoding, we can bound the following 
\begin{align}
    \|\widetilde{V}(\widetilde{H})-e^{-iHt}\| &\leq \|{V}(\widetilde{H})-e^{-iHt}\| + \|\widetilde{V}(\widetilde{H})-{V}(\widetilde{H})\| \\
    &=\|e^{-it\widetilde{H}} - e^{-iHt}\| + \epsilon_{\mathrm{exp}}\\
    &\leq \|\widetilde{H}-H\|t + \epsilon_{\mathrm{exp}} \\
    &=  t \epsilon_{be} + \epsilon_{\mathrm{exp}} \leq \epsilon.
\end{align}
In the second to last line we utilize an inequality from Ref. \cite{chakraborty2018power}, and note that $$\|\widetilde{H}-H\| = \|\alpha(\bra{0}^a\otimes \mathbb{I}) \widetilde{U} (\ket{0}^a \otimes \mathbb{I}) - H\|$$ is a block encoding error. In terms of the number of qubits, we require $a$ qubits for the block encoding, 1 additional qubit for the controlled block encoding, and 1 more for the controlled signal operators of GQSP. 
\end{proof}

We note that the additional 2 queries to the block encoding are a trivial addition to constant factor bounds we calculate, and exclude this from our numerical examples. In addition, a possible criticism of this GQSP doubling approach is that queries to the controlled block encoding can be more expensive than queries to the block encoding itself. However, in the case of LCHS, our construction of the select operator $S$ already requires access to a controlled version of the block encoding $U_A$, so under this assumption, making $S$ a controlled operation can be done with a single ancilla qubit. 

\section{Robust oblivious fixed-point amplitude amplification}
\label{app:robust_AA}
 
In this section we establish analytical upper bounds on the cost of amplitude amplification. The result is used in Theorem \ref{thm:amplified_prep} and Theorem \ref{thm:error} to yield the overall query counts of the LCHS solution preparation.
\begin{lemma}
    [Robust fixed-point oblivious amplitude amplification] \label{lem:robust_fp_obl_amp}
     We can construct a $(\|x\|,[m_x+1,m_x+n+1],\epsilon)$-block encoding of the vector $x\in\mathbb{C}^{2^n}$ that is amplified (i.e. does not need post-selection) from
    $U_x$, an $(\alpha_x,m_x,\epsilon_x)$-block encoding of $x$ that is conditional (i.e. succeeds only conditionally based on post-selection) using a number of queries to $U_x$ no greater than
\begin{align}
C_{x}(\Delta,\epsilon_{\mathrm{AA}})=\ceil{\sqrt{8\ceil{\frac{4}{\Delta^2}\ln(8/\pi \epsilon_{\mathrm{AA}}^2) e^2 }\ln\left (\frac{64\frac{\sqrt{2}}{\Delta}\ln^{1/2}(8/\pi \epsilon_{\mathrm{AA}}^2)}{3\sqrt{\pi} \epsilon_{\mathrm{AA}}}\right )}+1} 
\end{align}
as long as a value $\Delta$ satisfying $\Delta\leq 2\|x\|/\alpha_x\leq 9/5$ is known and the error parameters $\epsilon_{\mathrm{AA}}$ and $\epsilon_x$ satisfy
$\|x\|\left(\epsilon_{\mathrm{AA}} +
\frac{9\epsilon_x}{2\alpha_x}C_{x}(\Delta,\epsilon_{\mathrm{AA}})\right)\leq \epsilon$ and $\epsilon_x/\alpha_x\leq 1/12$ (i.e. $\|x/\alpha_x- \bra{0^m_x}U_x\ket{0^m_x}\ket{0^n}\|\leq \epsilon_x/\alpha_x\leq 1/12$).
\end{lemma}
\begin{proof}
As mentioned in \cite{gilyen2019quantum}, one can ``easily derive a fixed-point version of oblivious amplitude
amplification'' by leveraging their results. 
We follow their suggestion here. 
We start by considering the fixed point amplitude amplification algorithm of Theorem 27 in the pre-print version of \cite{gilyen2019quantum}. 
This result constructs from 
a perfect $U_x$, satisfying $\left\|\ket{0^m}\otimes(\|x\|\ket{x})-\alpha_x(\ketbra{0^m}{0^m}\otimes \mathbb{I})U_x\ket{0^m}\ket{0^n}\right\|=0$, a circuit $\overline{U}$ that satisfies
$\left\| \ket{0^m}\ket{x} - \overline{U} \ket{0^m}\ket{0^n} \right\| \leq \epsilon_{\mathrm{AA}}$ where $\ket{x}=x/\|x\|$, using no more than 
\begin{align}
C_{x}(\Delta,\epsilon_{\mathrm{AA}})=\ceil{\sqrt{8\ceil{\frac{4}{\Delta^2}\ln(8/\pi \epsilon_{\mathrm{AA}}^2) e^2 }\ln\left (\frac{64\frac{\sqrt{2}}{\Delta}\ln^{1/2}(8/\pi \epsilon_{\mathrm{AA}}^2)}{3\sqrt{\pi} \epsilon_{\mathrm{AA}}}\right )}+1} 
\end{align}
calls to $U_x$, if promised that $\Delta\leq 2\|x\|/\alpha_x$, where we have used Corollary \ref{coro:degree_bound} to establish a constant factor bound for the degree used in Theorem 27 of the pre-print version of \cite{gilyen2019quantum}.
Our contribution is to work out the case where $U_x$ is imperfect
\begin{align}
\left\|\ket{0^m}\otimes(\|x\|\ket{x})-\alpha_x(\ketbra{0^m}{0^m}\otimes \mathbb{I})U_x\ket{0^m}\ket{0^n}\right\|\leq\epsilon_x,
\end{align}
which is the block encoding assumption of the lemma statement.
We will bound the error in an imperfect implementation of an amplitude amplification unitary $\tilde{U}$.
This unitary is imperfect due to both the algorithm itself being imperfect (i.e. $\overline{U}$ is imperfect) as well as the fact that $\tilde{U}$ makes use of the imperfect block encoding $U_x$.
Using a triangle inequality to accommodate these error contributions, we can bound the total error as
\begin{align}
    \left\| x - \|x\| \bra{0^m}\tilde{U}\ket{0^m}\ket{0^n}\right\|&\leq \|x\|\left\| \ket{0^m}\otimes\left(\ket{x} - \bra{0^m}\overline{U}\ket{0^m}\ket{0^n}\right)\right\|+\|x\|\left\| \bra{0^m}\overline{U}\ket{0^m}\ket{0^n}-\bra{0^m}\tilde{U}\ket{0^m}\ket{0^n}\right\|
\end{align} 
The first term on the right-hand side is $\|x\|\epsilon_{\mathrm{AA}}$ 
and so the remaining task, which constitutes addressing robustness in $U_x$, is to upper bound the second term on the right-hand side of the above inequality.

For this robustness setting we follow the proof technique of Theorem 28 in the preprint version of \cite{gilyen2019quantum}.
Define $A=(\|x\|/\alpha_x)(\mathbb{I}\otimes\ketbra{x}{0^n})$ and $\tilde{A}=(\ketbra{0^m}{0^m}\otimes \mathbb{I})U_x(\ketbra{0^m}{0^m}\otimes\ketbra{0^n}{0^n})$.
With this we establish the condition $
\left\| A - \tilde{A} \right\| + \left\| \frac{A + \tilde{A}}{2} \right\|^2 \leq 1$
as required in Lemma 23, ``Robustness of singular value transformation 3'', of the preprint version of \cite{gilyen2019quantum}.

From this, $\left\|A+\tilde{A}\right\|\leq 2\left\|A\right\|+\left\|A-\tilde{A}\right\| = 2\|x\|/\alpha_x + \epsilon_x/\alpha_x \leq \frac{9}{5}+\frac{1}{12}\leq 113/60$, where we've used the assumptions that $\|x\|/\alpha_x\leq 9/10$ and $\epsilon_x/\alpha_x\leq 1/12$.
Thus, $\left\| A - \tilde{A} \right\| + \left\| \frac{A + \tilde{A}}{2} \right\|^2\leq\epsilon_x/\alpha_x + \frac{(1+\epsilon_x/\alpha_x)^2}{4}\leq 13969/14400 < 1$ and $
\sqrt{\frac{2}{1 - \left\| \frac{A + \tilde{A}}{2} \right\|^2}} 
\leq 
\sqrt{\frac{28800}{1631}}<9/2$.
With these conditions established, Lemma 23 of the preprint version of \cite{gilyen2019quantum} lets us bound the norm of the difference between the fixed point amplitude amplification algorithms using the ideal state preparation unitary (i.e. $\overline{U})$ and the imperfect one (i.e. $\tilde{U}$); defining $\overline{X}=(\ketbra{0^m}{0^m}\otimes \mathbb{I})\overline{U}(\ketbra{0^m}{0^m}\otimes \ketbra{0^n}{0^n})$ and $\tilde{X}=(\ketbra{0^m}{0^m}\otimes \mathbb{I})\tilde{U}(\ketbra{0^m}{0^m}\otimes \ketbra{0^n}{0^n})$, we have
$\left\| \overline{X}-\tilde{X}\right\| = \left\| P^{(SV)}(A)-P^{(SV)}(\tilde{A}) \right\| \leq \frac{9\epsilon_x}{2\alpha_x} C_{x}(\Delta,\epsilon_{\mathrm{AA}})$.

Finally, using the triangle inequality we have 

\begin{align}
    \left\| x - \|x\| \bra{0^m}\tilde{U}\ket{0^m}\ket{0^n}\right\|&\leq \|x\|\left\| \ket{0^m}\otimes\left(\ket{x} - \bra{0^m}\overline{U}\ket{0^m}\ket{0^n}\right)\right\|+\|x\|\left\| \bra{0^m}\overline{U}\ket{0^m}\ket{0^n}-\bra{0^m}\tilde{U}\ket{0^m}\ket{0^n}\right\|\\
    &= \|x\|\left\| \ket{0^m}\ket{x} - \overline{U}\ket{0^m}\ket{0^n}\right\|+\|x\|\left\| \overline{X}-\tilde{X}\right\|\\
    &\leq \|x\|\left(\epsilon_{\mathrm{AA}} +\frac{9\epsilon_x}{2\alpha_x}C_{x}(\Delta,\epsilon_{\mathrm{AA}})\right).
\end{align} 
Therefore, as long as $\|x\|\left(\epsilon_{\mathrm{AA}} +\frac{9\epsilon_x}{2\alpha_x}C_{x}(\Delta,\epsilon_{\mathrm{AA}})\right)\leq \epsilon$, then $\left\| x - \|x\| \bra{0^m}\tilde{U}\ket{0^m}\ket{0^n}\right\|\leq \epsilon$.
\end{proof}

\bibliography{bib}
\end{document}